\documentclass[12pt, a4paper]{article}

\usepackage{amsmath,amssymb,array}
\usepackage{siunitx}
\usepackage{longtable,tabularx}
\usepackage{algorithm}
\usepackage{algorithmic}
\usepackage{multirow}
\usepackage{color}
\usepackage{bm}
\usepackage[pdftex]{graphicx}
\DeclareMathOperator{\tr}{\mathrm{Tr}}

\newtheorem{theorem}{Theorem}
\newtheorem{definition}{Definition}
\newtheorem{lemma}[theorem]{Lemma}%
\newtheorem{remark}{Remark}%
\newtheorem{proposition}[theorem]{Proposition}%

\def\E{{\mathbb E}}

\def\cE{{\mathcal{E}}}

\def\QED{\mbox{\rule[0pt]{1.5ex}{1.5ex}}}

\newenvironment{proofof}[1][\proofname]{\paragraph{Proof of #1:}}{\hfill\QED}

\usepackage{bbm}
\usepackage{physics}
\usepackage{mathtools}
\usepackage{enumitem}
\usepackage{comment}

\usepackage[utf8]{inputenc}
\usepackage[T1]{fontenc}
\usepackage{graphicx}
\usepackage{longtable}
\usepackage{wrapfig}
\usepackage{rotating}
\usepackage[normalem]{ulem}
\usepackage{amsmath}
\usepackage{amssymb}
\usepackage{capt-of}
\usepackage{hyperref}
\usepackage{physics}
\usepackage{dsfont}
\usepackage[margin=2cm]{geometry}

\usepackage{authblk}

\author[1]{Farzin Salek}
\author[2]{Patrick Hayden}
\author[3]{Masahito Hayashi}

\affil[1]{Department of Mathematics, Technical University of Munich, Munich, Germany}
\affil[2]{Department of Applied Physics, Stanford University, Stanford, CA, USA}
\affil[2]{Department of Physics, Stanford University, Stanford, CA, USA}
\affil[3]{School of Data Science, The Chinese University of Hong Kong, Shenzhen, China}
\affil[3]{International Quantum Academy (SIQA),
Futian, Shenzhen, China}
\affil[3]{Graduate School
of Mathematics, Nagoya University, Nagoya, Japan}
\allowdisplaybreaks

\begin{document}
\title{Three-Receiver Quantum Broadcast Channels:\\Classical Communication with Quantum Non-unique Decoding
}

\date{}

\maketitle

\begin{abstract}
In network communication, it is common in broadcasting scenarios for there to exist a hierarchy among receivers based on information they decode due, for example, to different physical conditions or  premium subscriptions. This hierarchy may result in varied information quality, such as higher-quality video for certain receivers. This is modeled mathematically as a degraded message set, indicating a hierarchy between messages to be decoded by different receivers, where the default quality corresponds to a common message intended for all receivers, a higher quality is represented by a message for a smaller subset of receivers, and so forth. We extend these considerations to quantum communication, exploring three-receiver quantum broadcasts channel with two- and three-degraded message sets. Our technical tool involves employing quantum non-unique decoding, a technique we develop by utilizing the simultaneous pinching method. We construct one-shot codes for various scenarios and find achievable rate regions relying on various quantum Rényi mutual information error exponents. Our investigation includes a comprehensive study of pinching across tensor product spaces, presenting our findings as the asymptotic counterpart to our one-shot codes. By employing the non-unique decoding, we also establish a simpler proof to Marton's inner bound for two-receiver quantum broadcast channels without the need for more involved techniques. Additionally, we derive no-go results and demonstrate their tightness in special cases.

\end{abstract}


\section{Introduction and Problem Setup}
More than five decades ago, Cover \cite{cover-broadcast} formulated the problem of simultaneous communication from a single transmitter to several receivers via a single broadcast channel.
 In classical information theory, the broadcast channel is characterized by a conditional probability distribution $p(y_1, y_2, \cdots,y_k | x)$, with a single input $x$ for the single sender and multiple disjoint outputs $(y_1,\ldots,y_k)$, one for each receiver. Thus, each marginal channel $p(y_i\vert x)$ connecting receiver $Y_i$ to the transmitter undergoes a statistically different noise process.
 The capacity region of this channel consists of all possible rate pairs $(R_1,\ldots,R_k)$ which are simultaneously attainable with arbitrary reliability, where $R_i$ denotes the rate of the $i$-th receiver. 
  Extensive research has been dedicated to understanding the capacity of this channel with only two receivers, and especially in scenarios where the channel is assumed to be memoryless and available for many independent uses \cite{elgamal-kim,1056046,1056302,6145471}. 
Although this problem has been resolved in several special cases \cite{elgamal-kim,1054980,korner-marton}, the most general scenario where a sender wishes to transmit private messages to each receiver, has remained open for over five decades. Here, private does not refer to confidentiality, but rather to individualized message tailored for each receiver.

 To have an intuitive understanding of why this problem is notoriously difficult, consider giving a lecture to two groups with \textit{different backgrounds} and \textit{independent interests}; for example, the information intended for one group is about a math problem and for the other is about a piece of literature analysis (so that each group discards the information of the other even if they get to learn it). So the lecturer aims at providing independent information to the groups. One possible strategy is focusing solely on the math problem, ignoring the other group; another is to focusing solely on literature, disregarding the group interested in the math problem.  If the rate region is a portion of the positive quarter plane, these extreme cases correspond to two points on the $x$ and $y$ axes. A more practical approach would be if the lecturer divides his time between the two groups; so based on the dedicated time slot for each group, this $time-sharing$ strategy corresponds to a straight line connecting the two aforementioned extreme points on $x$ and $y$ axes. This strategy and similar ones are already established in the seminal paper \cite{cover-broadcast} to be sub-optimal. Instead, it was shown in the same paper that one should distribute high-rate information across low-rate messages, a technique known as superposition coding. In superposition coding, the messages are encoded into correlated random variables, such that one message is piggybacked on top of the other. Nonetheless, encoding private messages into correlated variables has proven to not always be effective. To address this issue, Marton proposed a technique in which random variables encoding two private messages are generated independently, but the variables become correlated through a particular binning structure \cite{1056046,1056302}. However, this approach also could not establish the capacity of the broadcast channels. Subsequently, it was shown that combining Marton's binning with superposition coding can largen the achievable rate region \cite{5673931,6809158}. Although converse bounds have come close to the latter rate region \cite{1055883,4797571}, the optimality remains an open question. 

The first scenario, where the capacity was established,
corresponds to a situation where one of the receivers is at a certain disadvantage compared to the other; for example, the lecturer's accent is more understandable to one group of pupils. Another example could be that one group is seated closer to the lecturer, so that they could hear the lecturer more clearly.
This is modeled, strictly speaking, as a degraded broadcast channel \cite{1054980,gallager_degraded,1056029}, where the Markov chain $X-Y_1-Y_2$ holds between the transmitter random variable $X$ and those of the two receivers $(Y_1,Y_2)$.

Another scenario where the capacity is established corresponds to a situation where one piece of information is intended for both groups of pupils, while some private information meant solely for one group. For example, a basic math course could be supplemented by an advanced topic, so that both groups tune in for the basic course but only one group receives the advanced course. (The other group may or may not learn this advanced course. That is not a concern in this setting.) Another example could be TV information broadcasting, where a default quality is guaranteed for all receivers. However, a higher-quality is guaranteed for receivers with a premium subscription. This scenario is modeled with degraded message sets: A common message will be received by all receivers, a private message is meant for a subset of receivers, another private message for a smaller subset of receiver, and so forth. For two-receiver broadcast channels, a two-degraded message set is relevant, where a common message is decoded by both receivers and a private message solely by one receiver \cite{korner-marton}. 

The situation becomes more challenging for three-receiver broadcast channels ${p(y_1,y_2,y_3\vert x)}$:
Transmitting private messages to each receiver becomes out of the question. Additionally, based on the relative statistical noise levels, one can define different hierarchies among receivers. Furthermore, one can define various two- and three-degraded message sets, giving rise to different scenarios. Among various emerging scenarios, there is only one situation where the capacity of the three-receiver broadcast channel has been established \cite{Nair-Elgamal}. This scenario pertains to a two-degraded message set, where a common message is intended for all three receivers and a private message solely to one receiver. Moreover, one of the receivers is assumed to be degradable with respect to another receiver, i.e. $X-Y_1-Y_2$; this model is referred to as three-receiver multilevel broadcast channel. The primary technique used by \cite{Nair-Elgamal} is called non-unique or indirect decoding. In this decoding scheme, a receiver tries to decode a certain message which is not intended for it. Therefore, in the case of an erroneous detection, there will be no further contribution to the overall error probability. The paper also suggests that achieving the capacity region without non-unique decoding is possible. However, the rate region without non-unique decoding becomes more complex, and establishing a matching outer bound may become very challenging. The technique of non-unique decoding was subsequently used in several other scenarios, and it was shown that the same rate regions can be attained without non-unique decoding \cite{6774863}. However, the question of whether this is generally the case remains unresolved.
The paper \cite{Nair-Elgamal} also investigates more involved models with two- and three-degraded message sets. However, there is no matching no-go theorem to establish the capacity region.

\medskip
Nearly four decades after Cover's seminal work on broadcast channels \cite{cover-broadcast}, quantum broadcast channels emerged, incorporating the principles of quantum mechanics \cite{q-yard}. A classical-quantum (cq-) broadcast channel describes a physical scenario where the transmitter is able to prepare any quantum state (described by a density matrix) in the laboratories of receivers conditioned on its classical letter $x$. In other words, the cq-broadcast channel pertains to a scenario where, upon sending a classical letter $x$ down the channel, each of the receivers obtains its respective quantum system, resulting in a collective state of the receivers in an entangled mixed state depending on $x$.
Due to the non-commutativity of operators in quantum mechanics, the problem of communication over quantum broadcast channels only becomes more challenging than the classical case. Therefore, a logical approach is to start by exploring quantum broadcast scenarios where their classical counterparts are resolved. Specifically,  one can attempt to develop codes or prove no-go results for situations where the classical counterpart is well-studied or its capacity is known.
This was also the approach taken by \cite{q-yard}. Utilizing the idea of superposition coding from classical information, the latter paper investigated two-receiver cq-broadcast channels when one of the receivers is degradable with respect to the other \cite{Devetak2005-kw} in addition to the scenario with one common and one private message (two-degraded message set). However, unlike the classical counterparts where the capacities are known, the no-go results for these two scenarios of cq-broadcast channels do not match the achievable rate region in general. These converse bounds are single-letter and present the same structure for the rate region, meaning that the number of equations and involved variables are similar. However, the difference with the achievable region stems from conditioning on quantum systems \cite{6634255,8370123}: The converse results in information theory usually involve identification of new random variables which depend on inputs and outputs. In the context of broadcast channels, all known identifications of random variables involve output systems, which are quantum. This implies that the no-go results presented in \cite{q-yard}, are, as of today, only known to coincide with the achievable region essentially in the classical setting. 

Naturally, the next step in this study was to generalize Marton's inner bound for quantum broadcast channels~\cite{savov-wilde}. However, unlike the superposition coding inner bound, generalization of Marton's code posed significant challenges. 
Initial attempts to generalize Marton's code to quantum information~\cite{savov-wilde-isit}, though promising, left some gaps. However, those challenges were overcome in \cite{7412732} by introducing the new idea of ``overcounting''. Subsequently, the gaps in the initial work were rectified in \cite{savov-wilde} by using overcounting idea. We discuss this further in Sec. \ref{marton-code}. Despite the difficulties faced, the quantum broadcast channels have remained an interesting topic over the years. For recent contributions and other important works, see \cite{5466521,Bäuml_2017,7438836,PhysRevD.109.065020,9153029,2023arXiv230412114C,cheng2023quantum}, as well as references within these papers for additional insights.

\medskip

In this work, our aim is to investigate three-receiver quantum broadcast channels, where an input alphabet is transmitted to three receivers. Each receiver observes a subsystem of a tripartite mixed state, which is determined by the letter transmitted through the channel.  
Formally, a discrete memoryless three-receiver classical-quantum broadcast channel consists of an input alphabet $\mathcal{X}$ and a set of density matrices $\{\rho_x^{B_1B_2B_3}\}_{x\in\mathcal{X}}$ on quantum systems $B_1B_2B_3$, where receiver $i$ obtains $\rho_x^{B_i}$ as output when the sender inputs classical symbol $x$ into the channel. We denote this channel as  $x\to\rho_x^{B_1B_2B_3}$. The goal is to communicate bits of information from the sender to the receivers. 
In the most general setting, the sender has three independent messages and each message is meant to be sent to a particular receiver. These messages are private or individualized in the sense that they only contain useful information for their intended receivers; otherwise, there are no secrecy requirements imposed on the messages. 

For the three-receiver quantum broadcast channel, we naturally study scenarios where the capacity of the classical counterpart is either known or is reasonably estimated \cite{Nair-Elgamal}. Specifically, since the capacity for the transmission of the private messages remains open even for the two-receiver broadcast channel, we instead prioritize studying the capacity for the degraded-message sets. 
A message set is called degraded when there exists a hierarchical structure among the messages intended for different receivers; that is, the smallest set of messages intended for all receivers, followed by additional sets that include the preceding ones in a nested fashion. For the three-receiver quantum broadcast channel, among various two- and three-degraded message sets, in this work, we focus on the following two scenarios:

\begin{itemize}
\item A two-degraded message set where there are two independent messages $(M_0,M_1)$ such that $M_0$ is a common message intended for all receivers and $M_1$ is intended solely for receiver $B_1$
    \item A three-degraded message set where there are three independent messages $(M_0,M_1,M_2)$ such that $M_0$ is a common message intended for all receivers, $M_1$ is intended for receivers $B_1$ and $B_2$ and $M_2$ is intended solely for receiver $B_1$.
\end{itemize}

\medskip

In the study of broadcast channels, it is natural to consider the situation where some receivers get better quality signals than other receivers. (This should not be confused with degraded message sets as the noise that each channel experiences is independent of the messages, although naturally a receiver with better conditions is can generally achieve higher rates.) A three-receiver multilevel quantum broadcast channel refers to a scenario where one of the receivers is degradable with respect to another receiver. (See the next section for the formal definition of degradable channels.) Here, we assume that receiver $B_2$ is degradable with respect to the receiver $B_1$, i.e. there exists a cptp map $\mathcal{M}^{B_1\to B_2}$ such that $\rho_x^{B_2}=\mathcal{M}(\rho_x^{B_1})$ for any $x\in\mathcal{X}$ (see Fig. \ref{multi-level.two-degraded}). The three-receiver multilevel quantum channel will be pivotal in our examination of the three-receiver quantum broadcast channel. 
We examine this channel with two-degraded message set in Sec. \ref{multilevel-two-degraded}. 
When there is no hierarchy among the quality of the received signals, we refer to the channel as the general three-receiver quantum broadcast channel. We study the general three-receiver quantum broadcast channel with two- and three-degraded message sets in Secs. \ref{general-two-degraded} and \ref{general-three-degraded}, respectively. When the hierarchy in signal quality is not relevant, we may inclusively say three-receiver quantum broadcast channel.

A one-shot code assumes that the channel is available for a single use. However, in the asymptotic regime, we assume that the three-receiver quantum broadcast channel is available for many independent uses. In this context, the channel is memoryless, meaning that the output at each use of the channel depends only on the input at that specific use and is conditionally independent of inputs and outputs at other times. In light of this, a one-shot codebook of rate $\widetilde{R}$, for example, 
$\{z(1),z(2),\ldots,z(2^{\widetilde{R}})\}$, drawn from a distribution $p(z)$, will be substituted with a codebook $\{z^n(1),z^n(2),\ldots,z^n(2^{nR})\}$. In this codebook, each codeword $z^n(i)$ is generated by sampling $n$ times from the distribution $p(z)$, resulting in a total of $2^{nR}$ codewords, and $R$ is the rate of the code. So in order to transmit the codeword $x^n(i)=(x_i(1),x_i(2),\ldots,x_i(2^{nR}))$, the channel is used $n$ times independently, the output will be $\rho_{x_i(1)}^{B_1B_2B_3}\otimes\rho_{x_i(2)}^{B_1B_2B_3}\otimes\cdots\otimes \rho_{x_i(2^{nR})}^{B_1B_2B_3}$. The decoding POVMs are constructed from pinching maps defined with respect to the spectral decomposition of the receiver states. As we will see later, the expectation of the one-shot average error probability scales linearly with the number of distinct eigenvalues of the corresponding pinching map. Therefore, if the number of distinct eigenvalues of a state grows exponentially with the number of tensor product Hilbert spaces, the asymptotic error probability may not vanish. We have demonstrated in Proposition \ref{pinching-asymptotic} that the number of eigenvalues grows only polynomially with the number of tensor product spaces. We will present our results in the asymptotic setting, considering the scenario where the channel can be utilized $n$ times, with $n\to\infty$. For this reason, we don't define a one-shot explicitly. Instead, we define an asymptotic code for the three-receiver quantum broadcast channel.

\medskip
\textbf{Results:} We summarize our contributions briefly as follows. We define the three-receiver multilevel quantum broadcast channel, a quantum analogue of the model studied in \cite{Nair-Elgamal}. We generalize the non-unique decoding from classical context to quantum information, and establish single-letter bounds on the capacity region of the three-receiver multilevel quantum broadcast channel. Our bounds are optimal up to the usual conditioning on quantum systems. We use our non-unique decoding to find an independent proof for Marton's inner bound with common message, without relying on previous ideas. Next, we study general three-receiver quantum broadcast channels with two- and three-degraded message sets. Our results are based on error exponents of one-shot codes which we derive using pinching techniques. Along the way, we present a comprehensive study of pinching over tensor product spaces, revealing novel insights. Ultimately, we present our principal findings as the asymptotic limit of our one-shot codes.

\medskip

The remainder of the paper is structured as follows. In the next section, we provide notations, definitions and briefly present our pinching techniques. For pedagogical reasons, we begin by presenting our novel proof of Marton's code in Sec. \ref{marton-code}. We then explore the three-receiver multilevel quantum broadcast channel with two-degraded message set in Sec. \ref{multilevel-two-degraded}. Following that, we study the scenario with two-degraded message set for the general three-receiver quantum broadcast channel in Sec. \ref{general-two-degraded}. Lastly, we investigate the case of general three-receiver quantum broadcast channel with three-degraded message set in Sec. \ref{general-three-degraded}. Sec. \ref{appendix} serves as an appendix, housing comprehensive discussion and proofs pertaining to pinching maps and hypothesis testing lemmas.

 \medskip

\section{Foundations and Methodology}\label{notation-definition}
%
For any positive integer $m$, we use the notation $[1:m]=\{1,...,m\}$.
Throughout the paper, $\log$ denotes by default the binary logarithm. 
The capital letters $U, V, X$, etc., denote random variables, whose realizations and the alphabets are shown by the corresponding lowercase ($u$, $x$, etc.) and calligraphic letters ($\mathcal{U}$, $\mathcal{X}$, etc.), respectively. We also use $d_U,d_X$, etc. for the cardinalities of these random variables, respectively. (This convention remains consistent except in the proof of no-go results, where it is explicitly specified which systems are quantum.) A general distribution is denoted by $(U,X)\sim p(u,x)$; we may also denote the latter by $p(U=u,X=x)$. Quantum systems $A$, $B$, etc., are associated with (finite-dimensional) Hilbert spaces and are denoted with the same letter, whose dimensions are denoted by $d_A$, $d_B$, etc. A quantum state is a Hermitian positive semidefinite operator in a Hilbert space whose trace equals one; this is also referred to as a density matrix. A density matrix corresponds to a pure state if and only if its ranks equals one. We use superscripts to denote the system on which an operator acts. Trace and partial trace are exceptions to this rule, there we use subscripts to denote the systems on which they act. Multipartite systems $AB\ldots Z$ are described by tensor product Hilbert spaces $A\otimes B\otimes \cdots \otimes Z$.  
The identity operator on Hilbert space $A$ is denoted by $\mathbbm{1}^A$.
 The purified distance between quantum state $\rho$ and $\sigma$ is defined as $P(\rho,\sigma)=\sqrt{1-F(\rho,\sigma)^2}$, with the fidelity $F(\rho,\sigma)=\|\sqrt{\rho}\sqrt{\sigma}\|_1$, where $\|\omega\|_1=\max_{0\leq\Lambda\leq I}\tr\Lambda\omega$ is the trace norm of Hermitian operator $\omega$. Alternatively, $\|\omega \|_1 \coloneqq \tr\sqrt{\omega^\dagger \omega}$.
For Hermitian operators $T$ and $O$, we define the operator $\{T\ge O\}$ as the projector onto the positive subspace of the operator $(T-O)$. 

 A quantum channel $\mathcal{N}^{A\to B}$ is a completely positive and trace-preserving (cptp) linear map (or superoperator) taking operators in $A$ to those in $B$. Consider two quantum channels $\mathcal{N}^{A\to B_1}$ and $\mathcal{M}^{A\to B_2}$. The channel $\mathcal{M}^{A\to B_2}$ is regarded as degraded or degradable with respect to $\mathcal{N}^{A\to B_1}$ if there exists a channel $\mathcal{J}^{B_2\to B_1}$, such that applying it after the channel $\mathcal{M}^{A\to B_2}$ results in $\mathcal{N}^{A\to B_1}$, i.e. $\mathcal{N}^{A\to B_1}=\mathcal{J}^{B_2\to B_1}\circ \mathcal{M}^{A\to B_2}$, where $\circ$ denotes channel concatenation.

Let $\rho^{ABC}$ be a tripartite state. We write $S(A)_\rho=S(\rho^A)=-\tr\rho^A\log\rho^A$ for the von Neumann entropy of the reduced density matrix $\rho^A$, dropping the subscript if the state is clear from the context. In analogy to classical Shannon theory, we define conditional entropy, mutual information, and conditional mutual information as follows, respectively:
\begin{align*}
    S(A\vert B)&= S(AB)-S(B),\\
    I(A;B)&=S(A)-S(A\vert B),\\
    I(A;B\vert C)&=S(A\vert C)-S(A\vert B,C).
\end{align*}
We will frequently work with classical-quantum (cq-) states. The statistical dependence between different random variables, as well as between them and quantum systems, will be explicitly reflected in the superscripts. For example, consider a cq-state $\rho^{UXB}$, with classical $(U,X)$ and quantum $B$ systems. When we have a Markov chain $U-X-B$, i.e. systems $U$ and $B$ are independent condition on $X$, we denote the cq-state by $\rho^{U-X-B}$, which is written as follows:
\begin{align*}
    \rho^{U-X-B} = \sum_{u,x} p(u,x) \ketbra{u}\otimes\ketbra{x}\otimes \rho_x^{B}.
\end{align*}

 For quantum states $\rho^{A}$ and $\sigma^{A}$, quantum relative entropy is defined as $D(\rho\|\sigma)=\tr\rho(\log\rho-\log\sigma)$ whenever $\rho\subseteq\text{supp}(\sigma)$, otherwise we set $D(\rho\|\sigma)=\infty$ \cite{Wilde_2013}. Here $\text{supp}(\rho)$ is the support of $\rho$, i.e., the span of eigenvectors of $\rho$ corresponding to non-zero eigenvalues.
Moving forward, for $\alpha \in (0, 1) \cup (1, \infty)$, the Petz R\'enyi relative entropy \cite{Petz1986QuasientropiesFF} and the sandwiched Rényi relative entropy \cite{Wilde_2014,M_ller_Lennert_2013}, are defined respectively as follows: If $\rho\subseteq\text{supp}(\sigma)$,
\begin{align*}
     D_{\alpha}(\rho \| \sigma) &\coloneqq \frac{1}{\alpha - 1} \log \operatorname{tr} \left\{\rho^{\alpha} \sigma^{1 - \alpha} \right\},\\
      \widetilde{D}_{\alpha}(\rho \| \sigma) &\coloneqq \frac{1}{\alpha - 1} \log \operatorname{tr} \left\{\left( \sigma^{\frac{1 - \alpha}{2\alpha}} \rho \sigma^{\frac{1 - \alpha}{2\alpha}} \right)^{\alpha} \right\},
\end{align*}
else, we set $D_{\alpha}(\rho \| \sigma) = \widetilde{D}_{\alpha}(\rho \| \sigma) = \infty$. Both of these R\'enyi relative entropies approach the quantum relative entropy $D(\rho\|\sigma)$ as $\alpha\to 1$.
\cite{M_ller_Lennert_2013,Wilde_2014}.
We define a variety of R\'enyi mutual information quantities as
\begin{align*}
I_\alpha^{\uparrow}(A;B)_{\rho^{AB}}&\coloneqq D_\alpha(\rho^{AB}\| \rho^A \otimes \rho^B)\\
\tilde{I}_\alpha^{\uparrow}(A;B)_{\rho^{AB}}&\coloneqq\widetilde{D}_\alpha(\rho^{AB}\| \rho^A \otimes \rho^B)
\\
\tilde{I}_\alpha^{\downarrow}(X;B)_{\rho^{XB}|\rho^X}&\coloneqq
\min_{\sigma^B} \widetilde{D}_\alpha(\rho^{XB}\| \rho^X \otimes \sigma^B) \\
\tilde{I}_\alpha^{\downarrow}(X;B|U)_{\rho^{UXB}|\rho^{UX}}&\coloneqq
\min_{\sigma^{X-U-B}: \sigma^{UX}=\rho^{UX}} 
\widetilde{D}_\alpha(\rho^{UXB}\| \sigma^{X-U-B}) .
\end{align*}
It is known that the aforementioned quantum R\'enyi mutual information terms are additive for tensor product states \cite{AHW,beigi-sandwiched}: For $\alpha\in (\frac{1}{2},1)\cup (1,\infty)$, we have
\begin{equation}
    \begin{aligned}
    \label{renyi-asymp}
        \tilde{I}_\alpha^{\uparrow}(A^n;B^n)_{(\rho^{AB})^{\otimes n}}&=n\tilde{I}_\alpha^{\uparrow}(A;B)_{\rho^{AB}},\\
        \tilde{I}_\alpha^{\downarrow}(X^n;B^n|U^n)_{(\rho^{UXB})^{\otimes n}|(\rho^{UX})^{\otimes n}}&=n\tilde{I}_\alpha^{\downarrow}(X;B|U)_{\rho^{UXB}|\rho^{UX}},
    \end{aligned}
\end{equation}
Max-relative entropy is defined as \cite{4957651}
\begin{align*}
    D_{\text{max}}(\rho\|\sigma)=\log\min\{\lambda: \rho\le \lambda\sigma\},
\end{align*}
and the smooth max-relative entropy as:
\begin{align*}
    D_{\text{max}}^{\varepsilon}(\rho\|\sigma)=\min_{\tilde{\rho}\in\mathcal{B}^{\varepsilon}(\rho)}D_{\text{max}}(\tilde{\rho}\|\sigma),
\end{align*}
where the $\varepsilon$-ball is defined using 
  the purified distance as $\mathcal{B}^{\varepsilon}(\rho)=\{\tilde{\rho}: P(\tilde{\rho},\rho)\le \varepsilon\}$.
In this work, we only need smooth max-relative entropy with classical systems. For a classical states $\rho^{UVX}$ and $\rho^{U-V-X}$, conditional smooth max-mutual information is defined as follows:
\begin{align*}
    I_{\text{max}}^{\varepsilon}(U;X|V)\coloneqq D_{\text{max}}^{\varepsilon}(\rho^{UVX}\|\rho^{U-V-X}).
\end{align*} 
It is known that the limit of this quantity on tensor product Hilbert spaces converges to to the (Shannon) conditional mutual information \cite{6574274}:
\begin{align}
\label{max-asymp}
    \lim_{n\to\infty}\frac{1}{n}I_{\text{max}}^{\varepsilon}(X^n;Y^n|U^n)_{\rho^{\otimes n}}=I(X;Y|U).
\end{align}

\medskip

Our achievability proofs rely on the Hayashi-Nagaoka inequality:
\begin{lemma}[Hayashi-Nagaoka inequality \cite{1207373}]
\label{hayashi-nagaoka}
    Consider the operators $0\le S\le I$ and $T\ge 0$. The following operator inequality holds:
    \begin{align*}
        I-(S+T)^{-\frac{1}{2}}S(S+T)^{-\frac{1}{2}}\leq 2(I-S)+4T.
    \end{align*}
\end{lemma}

\medskip

Pinching maps \cite{Hayashi2017-cv} are superoperators that play a central role in our technical derivations.
Consider the spectral decomposition of any Hermitian operator $H^A$ with $m\le d_A$ distinct eigenvalues ${h_1, \ldots, h_m}$:
\begin{align*}
    H = \sum_{i=1}^{m} h_i H_i, 
    \end{align*}
    where $H_i \,(=H_i^2)$ denotes the orthogonal projection onto the eigenspace corresponding to eigenvalue $h_i$ (so that the projectors $\{H_i\}_i$ are not necessarily rank one operators). 
 The pinching map with respect to the spectral decomposition of $H$ is defined as follows:
\begin{align*}
    \mathcal{E}_{H}(\rho) = \sum_i H_i \rho H_i.
\end{align*}
These maps are cptp channels so that $\sum_i H_i=\mathbbm{1}$. The action of the pinching map corresponds to the dephasing channel, where the off-diagonal elements of the density matrix are removed depending on the dephasing level. If the number of distinct eigenvalues equal the dimension of the system, i.e. $m=d_A$, the pinching corresponds to the completely dephasing channel where all off-diagonal elements are removed. On the other hand, if there is a single eigenvalue, so the underlying operator is $\mathbbm{1}/d_A$, the pinching map acts as the identity superoperator. When we work with cq-states, we may replace the subscript $H$ in $\mathcal{E}_{H}$ by a classical letter. For example, consider the cq-states $\rho^{UB}$ and $\sigma^{UB}$. For each $u$, let $\mathcal{E}_{u}^B$ be the pinching map with respect to the spectrum of $\sigma_u^B$. We can then define a pinching map on $UB$ as $\mathcal{E}^{UB}=\sum_u \ketbra{u}\otimes \mathcal{E}_u^B$. Then $\mathcal{E}^{UB}(\rho^{UB})$ commutes with $\sigma^{UB}$.
In general, the states $\rho^{AB}$ and $\mathbbm{1}^{A}\otimes \rho^B$ do not commute. However, they start commuting when we pinch the former with the spectrum of the operator $\rho^B$. 

In this work, we will work with nested pinching maps whose underlying operators are pinched density matrices on an $n$-fold tensor product Hilbert space. A crucial step consists of bounding the maximum number of distinct eigenvalues of these operators. We present these polynomial upper bounds in the following proposition. 

\begin{proposition}\label{pinching-asymptotic}
Consider the cq-states $\rho^{UVXB},\rho^{XV-U-B},\rho^{XU-V-B}$, and $\rho^{U-VX-B}$, classical on $U,V,X$ and quantum on $B$. Let $\nu,\nu_1,\nu_2$ and $\nu_3$ be defined as follows:
\begin{align*}
 \nu & = \left\{\text{distinct eigenvalues of}\hspace{0.1cm} (\rho^B)^{\otimes n}\right\},\\
    \nu_1 & = \max_{u^n\in\mathcal{U}^n} {\left\{\text{distinct eigenvalues of}\hspace{0.1cm} \mathcal{E}^{B^n}\big(\rho^{B^n}_{u^n}\big)\right\}},\\
     \nu_2 & = \max_{u^n\in\mathcal{U}^n,v^n\in\mathcal{V}^n} {\left\{\text{distinct eigenvalues of}\hspace{0.1cm} \mathcal{E}^{B^n}_{1|u^n}\big(\rho^{B^n}_{v^n}\big)\right\}},\\
     \nu_3 & = \max_{\substack{v^n\in\mathcal{V}^n,x^n\in\mathcal{X}^n\\u^n\in\mathcal{U}^n}} {\left\{\text{distinct eigenvalues of}\hspace{0.1cm} \mathcal{E}^{B^n}_{1|u^n}\big(\rho^{B^n}_{v^n,x^n}\big)\right\}}
\end{align*}
where $\rho^{B^n}_{u^n}= \rho_{u_1}^{B}\otimes \rho_{u_2}^{B}\otimes\cdots\otimes\rho_{u_n}^{B}$, $\rho^{B^n}_{v^n}= \rho_{v_1}^{B}\otimes \rho_{v_2}^{B}\otimes\cdots\otimes\rho_{v_n}^{B}$, $\rho^{B^n}_{v^n,x^n}=\rho_{u_1,v_1}^B\otimes\cdots,\otimes\rho_{u_n,v_n}^B$, and $\mathcal{E}^{B^n}$ and $\mathcal{E}^{B^n}_{1|u^n}$ are pinching maps corresponding to the spectral decomposition of the operators $(\rho^B)^{\otimes n}$ and $\mathcal{E}^{B^n}\left(\rho^{B^n}_{u^n}\right)$, respectively. Then the number of distinct eigenvalues are polynomial in the number of tensor product Hilbert spaces. More precisely,
\begin{align*}
\nu & \le (n+1)^{d_B-1} \\
\nu_1 & \le 
(n+1)^{d_U(d_B+2)(d_B-1)/2} \\
\nu_2 & \le (n+1)^{d_V d_U(d_B+2)(d_B-1)/2}\\
\nu_3 & \le (n+1)^{d_Xd_V d_U(d_B+2)(d_B-1)/2} .
\end{align*}
\end{proposition}
The proof of this proposition and further discussion on pinching are deferred to the appendix. (See Proposition \ref{pinching-asymptotic-appendix}.)

 \bigskip

\section{Quantum Marton Code with Common Message}
\label{marton-code}

\begin{figure}[t]
\includegraphics[width=15cm]{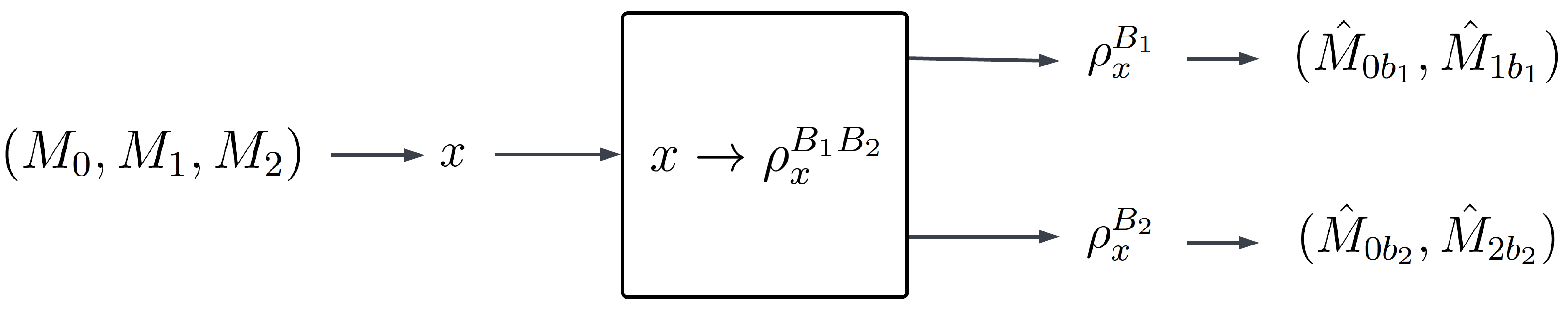}
\centering
\caption{Two-receiver quantum broadcast channel with private and common messages. A common message $M_0$ is intended for both receivers, while two private messages $M_1$ and $M_2$ are intended for receiver $B_1$ and $B_2$, respectively (here and throughout private does not imply confidentiality, but rather individualized). $M_{0b_i}$ denotes the estimate made by receiver $B_i$ about the message $M_0$. Similarly, $M_{1b_1}$ denotes the estimate made by receiver $B_1$ about the the message $M_1$ and $M_{2b_2}$ denotes the estimate made by receiver $B_2$ about the the message $M_2$}.
\label{Marton-fig}
\end{figure}

In this section, we provide a new proof of Marton's inner bound \cite{1056302,1056046,5673931} for the two-receiver quantum broadcast channel using the new idea of non-unique decoding. In particular, our proof does not rely on overcounting techniques \cite{7412732,savov-wilde}.
Generalization of the Marton's inner bound to quantum information involves new challenges that are unique to this problem. To estimate the error probability of a code, addressing certain error events becomes necessary. One such event concerns the likelihood of an independent codeword being mistaken for the correct one. In Marton's code, the codebook consists of bins, and given an output state, we are considering the independence of the output state from a codeword inside the message bin which is not the actual transmitted codeword. Although this codeword is not the actual transmitted one, it is not entirely independent from the output because both share a correlation stemming from belonging to the same bin. To put it another way, the output is generated from a codeword inside that bin, but we don't know which. This generates some sort of correlation. In the case where the POVMs seek to uniquely identify the codewords, the first term resulting from Hayashi-Nagaoka inequality (Lemma \ref{hayashi-nagaoka}) includes products of correlated operators, making the analysis challenging.
This was the primary obstacle in early attempts to generalize Marton's code to quantum information.

While the techniques of classical information theory can deal relatively easily with this problem, its analysis poses significant challenges in quantum information. The first attempts \cite{savov-wilde-isit} to adapt Marton's code to quantum information showed potential but had some gaps. However, these obstacles were tackled in \cite{7412732} by introducing the new concept of ``overcounting''. Later on, the issues in the initial research were fixed by applying the overcounting idea in \cite{savov-wilde}.
 The overcounting technique seems to have three caveats. First it is more involved. Second, the resulting preliminary rate inequalities do not mirror their classical counterparts. And third, it is not clear how it can be extended to the case with common message (though \cite{Sen2021} deals with Marton's code with common message). On the other hand, our proposed non-unique decoding technique avoids these complications and provides an analysis for the first term arising from the Hayashi-Nagaoka inequality which mirrors the classical result.

Here, we introduce a novel technique called non-unique decoding to establish Marton's inner bound for the two-receiver quantum broadcast channel without the necessity for overcounting ideas. In our technique, the decoders do not try to decode the transmitted codeword exactly, but only decide whether the transmitted codeword belongs to a specific bin. Note that using overcounting ideas, the decoder can decide which codeword was exactly transmitted. However, there is no need to learn the exact transmitted codeword because the goal is to decode the index of the bin, revealing the message, rather then the exact codeword. In the following, we formally define the problem and the code and present our result.

 Consider a two-receiver quantum broadcast channel $x\to\rho_x^{B_1B_2}$. A transmitter wants to send a common message $m_0\in[1:2^{nR_0}]$ to both receivers, and two private messages $m_1\in[1:2^{nR_1}]$ and $m_2\in[1:2^{nR_2}]$, to receivers $B_1$ and $B_2$, respectively; see Fig. \ref{Marton-fig}. 
We will present our results in the asymptotic setting, considering the scenario where the channel can be utilized $n$ times, with $n\to\infty$. For this reason, we don't define a one-shot explicitly. Instead, we define an asymptotic code as follows.

\begin{definition}\label{code-def-Marton}
A $(2^{nR_0},2^{nR_1},2^{nR_2},n,P^{(n)}_e)$ code $\mathcal{C}_n$ for the two-receiver quantum broadcast channel $x\to\rho_x^{B_1B_2}$ consists of the following components:
\begin{itemize}
    \item A triple of messages $(M_0,M_1,M_2)$ uniformly distributed over $[1:2^{nR_0}]\times[1:2^{nR_1}]\times[1:2^{nR_2}]$, where $R_0, R_1$ and $R_2$ represent the rates of the messages $M_0,M_1$ and $M_2$, respectively.
    \item An encoder that assigns a codeword $x^n(m_0,m_1,m_2)$ to each message triple $(m_0,m_1,m_2)\in[1:2^{nR_0}]\times[1:2^{nR_1}]\times[1:2^{nR_2}]$,
    \item Two POVMs, $\Lambda_{m_0m_1}^{B_1^{n}}$ on $B_1^{\otimes n}$, $\Lambda_{m_0m_2}^{B_2^{n}}$ on $B_2^{\otimes n}$ (indicating which messages are intended for decoding by each receiver) which satisfy
    \begin{align}
\tr{\rho_{x^n(m_0,m_1,m_2)}^{B_1^{n}B_2^{n}}\left(\Lambda_{m_0m_1}^{B_1^{n}}\otimes\Lambda_{m_0m_2}^{B_2^{n}}\right)}\geq 1-P^{(n)}_e(\mathcal{C}_n),
    \end{align}
    for every $(m_0,m_1,m_2)\in[1:2^{nR_0}]\times[1:2^{nR_1}]\times[1:2^{nR_2}]$. $P^{(n)}_e(\mathcal{C}_n)$ is the average probability of error defined as follows:
    \begin{align*}
        P^{(n)}_e(\mathcal{C}_n)= \text{Pr}\{\hat{M}_{0b_1},\hat{M}_{0b_2}\neq M_0 \,\text{or}\,\hat{M}_{1b_1},\neq M_1\,\text{or}\, \hat{M}_{2b_2}\neq M_2\},
    \end{align*}
 where $\hat{M}_{jb_i}$ represents the estimate made by receiver $B_i$ for the message $j$.
\end{itemize}
A rate triple $(R_0,R_1,R_2)$ is said to be achievable if there exists a sequence of codes $\mathcal{C}_n$ such that $P^{(n)}_e(\mathcal{C}_n)\to 0$ as $n\rightarrow \infty$. The capacity region is the closure of the set of all such achievable rate triples. 
\end{definition}

\medskip

\begin{theorem}
\label{Marton}
  Assuming a transmitter and two receivers have access to many independent uses of a two-receiver quantum broadcast channel $x\to\rho_x^{B_1B_2}$, the set of rate triples $(R_0,R_1,R_2)$ becomes asymptotically achievable if
\begin{align*}
    R_0 + R_1 & \leq I(U_0,U_1;B_1),\\
    R_0 + R_2 & \leq I(U_0,U_2;B_2),\\
    R_0 + R_1 + R_2 & \leq  I(U_0,U_1;B_1) + I(U_2;B_2|U_0) - I(U_1;U_2|U_0),\\
    R_0 + R_1 + R_2 & \leq   I(U_0,U_2;B_2) +I(U_1;B_1|U_0) - I(U_1;U_2|U_0), \\
    2R_0 + R_1 + R_2 & \leq   I(U_0,U_1;B_1) +I(U_0,U_2;B_2) - I(U_1;U_2|U_0),
\end{align*}
for some distribution $p(u_0,u_1,u_2)$ and function $x(u_0,u_1,u_2)$, giving rise to the state
\begin{align*}
    \rho^{U_0U_1U_2XB_1B_2} = \sum_{u_0,u_1,u_2} p(u_0,u_1,u_2)\ketbra{u_0}^{U_0}\otimes \ketbra{u_1}^{U_1}\otimes \ketbra{u_2}^{U_2}\otimes\rho_{x(u_0,u_1,u_2)}^{B_1B_2}.
\end{align*}

\end{theorem}

\medskip

The proof of this theorem will be presented as the asymptotic analysis of the one-shot code we construct below. Therefore, the proof is deferred until Subsection \ref{marton-asymptotic}. We break down the process and go through the details in distinct subsections.

\subsection{One-shot code construction}
The asymptotic analysis of the code developed here recovers Theorem \ref{Marton}. Fix a pmf $p(u_0,u_1,u_2)=p(u_0)p(u_1\vert u_1)p(u_2\vert u_0)$, i.e. the random variables $U_1$ and $U_2$ are conditionally independent given $U_0$. The key idea behind Marton's code is to construct a code such that the aforementioned pmf will be close to a general pmf $p(u_0,u_1,u_2)$ without conditional independence.  The goal is to send a common message $m_0\in[0,2^{R_0}]$ to both $B_1$ and $B_2$, and two private messages $m_1\in[0,2^{R_1}]$ and $m_2\in[0,2^{R_2}]$ to receivers $B_1$ and $B_2$, respectively. 
\subsubsection{Rate splitting}
We split the message $M_1$ into two independent parts, $M_{11}$ at rate $S_{11}$ and $M_{12}$ at rate $S_{12}$. Similarly, the message $M_2$ is split into two independent parts, $M_{21}$ at rate $S_{21}$ and $M_{22}$ at rate $S_{22}$. Hence, $R_1 = S_{11} + S_{12}$ and $R_2 = S_{21} + S_{22}$.

\subsubsection{Codebook generation}
Fix a pmf $p(u_0,u_1,u_2)=p(u_0)p(u_1\vert u_1)p(u_2\vert u_0)$, i.e. $U_1-U_0-U_2$. We generate a codebook of size $2^{(R_0+S_{11}+S_{21})}$ randomly and independent according to the marginal distribution $p(u_0)$ as follows: 
\begin{align*}
     \{u_0(m_0,m_{11},m_{21})\}_{m_{0}\in [1:2^{R_0}], m_{11}\in [1:2^{S_{11}}],m_{21}\in[1:2^{S_{21}}]}\subset\mathcal{U}_0.
\end{align*}

For each codeword $u_0(m_0,m_{11},m_{21})$ corresponding to the message triple $(m_0,m_{11},m_{21})$, we generate two codebooks randomly and conditionally independently using conditional distributions $p(u_1|U_0=u_0)$ and $p(u_2|U_0=u_0)$ as follows: 
\begin{align*}
    &\{u_1(m_0,m_{11},m_{21},m_{12},k_1)\}_{m_{12}\in [1:2^{S_{12}}],k_1\in[1:2^{r_1}]}\subset\mathcal{U}_1\\
    &\{u_2(m_0,m_{11},m_{21},m_{22},k_2)\}_{m_{22}\in [1:2^{S_{22}}],k_2\in[1:2^{r_2}]}\subset\mathcal{U}_2.
\end{align*}
These codeboks contain $2^{(S_{12}+r_1)}$ and $2^{(S_{22}+r_2)}$ elements, respectively. We divide each of these latter codebooks into $2^{S_{12}}$ and $2^{S_{22}}$ bins, respectively, so that each bin contains $2^{r_1}$ and $2^{r_2}$ elements, respectively \footnote{Note that this partitioning is different from random binning in the context of source compression. Here each bin contains equal number of elements, each generated conditionally independent from other elements.}. The indices $k_1$ and $k_2$ keep track of the codewords inside each bin while the index of the bins denote the messages, i.e. they correspond to message $m_{12}$ of rate $S_{12}$ and $m_{22}$ of rate $S_{22}$. Each product bin $(m_{12},m_{22})$ therefore contains $2^{r_1+r_2}$ elements. To ensure that at least one element $(u_1(m_0,m_{11},m_{21},m_{12},k_1),u_2(m_0,m_{11},m_{21},m_{22},k_2))$ inside the product bin $(m_{12},m_{22})$ is suitable for encoding, we require that 
\begin{align*}
    r_1+r_2 \geq I_{\text{max}}^{\varepsilon}(U_1;U_2|U_0)+2\log\frac{1}{\varepsilon}.
\end{align*}
This follows from mutual covering lemma \cite[Fact 2]{Sen2021} (see also \cite{7282697}).
Roughly speaking, this requirement means that there exists at least one pair $$(u_1(m_0,m_{11},m_{21},m_{12},k_1),u_2(m_0,m_{11},m_{21},m_{22},k_2))$$ in each product bin $(m_{12},m_{22})$ whose joint distribution is close to $p(u_0)p(u_1,u_2|u_0)$ although they are generated from $p(u_0)p(u_1|u_0)p(u_2|u_0)$. If there is more than one such pair, choose an arbitrary one among them. Finally based on the three chosen codewords, generate 
\begin{align*}
    x\big(u_0(m_0,m_{11},m_{21}),u_1(m_0,m_{11},m_{21},m_{12},k_1),u_2(m_0,m_{11},m_{21},m_{22},k_2)\big).
\end{align*}

\medskip

\subsubsection{Encoding}

To send the message triple $(m_0,m_1,m_2)$, send the function $x(u_0,u_1,u_2)$ over the channel.

\medskip

\subsubsection{Decoding and Analysis of error probability}
The failure of encoder contributes a constant additive term $f(\varepsilon)$ to the ultimate error probability. We go through the details of the decoding errors below assuming that the encoding was successful.
Each decoder decodes the common message by finding the unique codeword $u_0(m_0,m_{11},m_{21})$ corresponding to the transmitted common message. For the private messages, however, unique decoding is not necessary; as long as the bin index containing the transmitted codeword is decoded correctly, there are no errors, even if the decoder selects a codeword different from the one sent by the transmitter. 

In the following we assume the encoding step was successful, meaning that a suitable pair was found inside the product bin $(m_{12}, m_{22})$ satisfying the mutual covering lemma. 
Suppose thes selected pair is $(u_1(m_0,m_{11},m_{21},m_{12},k_1),u_2(m_0,m_{11},m_{21},m_{22},k_2))$. We now describe POVM construction and error probability analysis for each receiver.

\medskip

\textbf{\emph{Receiver $B_1$:}}
We define the following projectors:

\begin{align*}
    \Pi^{U_0U_1B_1}_1&\coloneqq \{\mathcal{E}_1^{U_0B_1}\left(\rho^{U_0U_1B_1}\right)\geq 2^{S_{12}+r_1}\mathcal{E}^{B_1}\left(\rho^{U_1-U_0-B_1}\right)\},\\
      \Pi^{U_0U_1B_1}_0&\coloneqq \{\mathcal{E}^{B_1}(\rho^{U_0U_1B_1})\geq 2^{R_0+S_{11}+S_{21}+S_{12}+r_1}\rho^{U_0U_1}\otimes\rho^{B_1}\},
\end{align*}
where $\mathcal{E}^{B_1}$ is the pinching map with respect to the spectral decomposition of the operator $\rho^{B_1}$ and $\mathcal{E}_1^{U_0B_1}=\sum_{u_0}\ketbra{u_0}\otimes\mathcal{E}_{1|u_0}^{B_1}$, where $\mathcal{E}_{1|u_0}^{B_1}$ is a pinching map with respect to the spectral decomposition of $\mathcal{E}^{B_1}(\rho_{u_0}^{B_1})$ for any $u_0\in\mathcal{U}_0$.
We further define the following operators:
\begin{align*}
    \Pi_{1\,u_0,u_1}^{B_1} = \bra{u_0,u_1}\Pi^{U_0U_1B_1}_1\ket{u_0,u_1},\\
    \Pi_{0\,u_0,u_1}^{B_1} = \bra{u_0,u_1}\Pi^{U_0U_1B_1}_0\ket{u_0,u_1},
\end{align*}
and the operator that we use in constructing the pretty-good measurement:
\begin{align*}
\Pi^{B_1}_{u_0,u_1}=\Pi_{1\,u_0,u_1}^{B_1}\Pi_{0\,u_0,u_1}^{B_1}
\end{align*}
The operators we have defined so far are not code-specific. For a specific code $\mathcal{C}_1$, we constrcut the following square-root measurement POVM as the decoding operation of receiver $B_1$:

\begin{align*}   \Omega^{B_1}_{m_0,m_{11},m_{21},m_{12}}&=
\left(\sum_{m_0'=1}^{2^{R_0}}\sum_{m_{11}'=1}^{2^{S_{11}}}\sum_{m_{21}'=1}^{2^{S_{21}}}\sum_{m_{12}'=1}^{2^{S_{12}}}\Pi_{u_1(m_0',m_{11}',m_{21}',m_{21}',m_{12}')}^{B_1}\right)^{-\frac{1}{2}}\\
&\hspace{2.5cm}\,\times\,\Pi_{u_1(m_0,m_{11},m_{21},m_{12})}^{B_1}\,\times\\\
&\hspace{3.5cm}\left(\sum_{m_0'=1}^{2^{R_0}}\sum_{m_{11}'=1}^{2^{S_{11}}}\sum_{m_{21}'=1}^{2^{S_{21}}}\sum_{m_{12}'=1}^{2^{S_{12}}}\Pi_{u_1(m_0',m_{11}',m_{21}',m_{21}',m_{12}')}^{B_1}\right)^{-\frac{1}{2}}
\end{align*}
where 
\begin{align}
\label{non-unique-povm-1}
\Pi_{u_1(m_0,m_{11},m_{21},m_{12})}^{B_1}=\sum_{k_1=1}^{2^{r_1}}\Pi_{u_1(m_0,m_{11},m_{21},m_{12},k_1)}^{B_1}.
\end{align}

The decoder finds the messages encoded into the variable $U_0$ uniquely. However, the message encoded into the variable $U_1$ is decoded non-uniquely up to index $k_1$. The non-unique decoding becomes relevant in finding the private message $m_{12}$ because an error in detecting $m_{12}$ occurs only if some codeword from a different bin $m_{12}'\neq m_{12}$ clicks. In other words, as long as the decoder find a codeword inside the bin $m_{12}$, there will be no errors, even though the decoded codeword may be $u(m_{0},m_{11},m_{21},m_{12},k_1')$ for some $k_1'$ different from the actual chosen $k_1$. 
The average error probability of the code $\mathcal{C}_1$ is bounded as follows. Recall that we use the variable $\hat{M}_{tb_i}$ to denote the estimate made by receiver $B_i$ about that message $t$. We obtain:
\begin{align*}
    P_{e,1}(\mathcal{C}_1)&=\text{Pr}\{\hat{M}_{0b_1}\neq M_0,\hat{M}_{1b_1}\neq M_1|\mathcal{C}_1\}\\
    &=\text{Pr}\{\hat{M}_{0b_1}\neq M_0,\hat{M}_{11b_1}\neq M_{11},\hat{M}_{12b_1}\neq M_{12}|\mathcal{C}_1\}\\
     &\le\text{Pr}\{\hat{M}_{0b_1}\neq M_0,\hat{M}_{11b_1}\neq M_{11},\hat{M}_{21b_1}\neq M_{21},\hat{M}_{12b_1}\neq M_{12}|\mathcal{C}_1\}\\
    &=\frac{1}{2^{R_0 + S_{11} + S_{21}+ S_{12}}}\sum_{m_0=1}^{2^{R_0}}\sum_{m_{11}=1}^{2^{S_{11}}}\sum_{m_{21}=1}^{2^{S_{21}}} \sum_{m_{12}=1}^{2^{S_{12}}}\tr(I-\Omega^{B_1}_{m_0,m_{11},m_{21},m_{12}})\rho_{u_1(m_0,m_{11},m_{21},m_{12},k_1)}^{B_1}\\
    &\leq
    \frac{1}{2^{R_0 + S_{11} + S_{21}+ S_{12}}}\sum_{m_0=1}^{2^{R_0}}\sum_{m_{11}=1}^{2^{S_{11}}}\sum_{m_{21}=1}^{2^{S_{21}}}\sum_{m_{12}=1}^{2^{S_{12}}} \tr\biggl\{\biggl(2\big(I-\Pi_{u_1(m_0,m_{11},m_{21},m_{12})}^{B_1}\big)\\
    &+4\big(\sum_{\substack{(m_0',m_{11}',m_{21}',m_{21}',m_{12}')\neq\\ (m_0,m_{11},m_{21},m_{12})}}\Pi_{u_1(m_0',m_{11}',m_{21}',m_{21}',m_{12}')}^{B_1}\big)\biggl)\rho_{u_1(m_0,m_{11},m_{21},m_{12},k_1)}^{B_1}\biggl\},
\end{align*}
where the first inequality follows since we add a new event to the error probability and the second corresponds to the Hayashi-Nagaoka inequality Lemma \ref{hayashi-nagaoka}. We analyze each of the two terms arose from the Hayashi- Nagaoka inequality separately. Note that we remove the explicit reference to the upper and lower summands where they are clear from the context. We begin with the first term:
    \begin{align}
    \nonumber
    &\frac{2}{2^{R_0 + S_{11} + S_{21}+ S_{12}}}\sum_{\substack{m_0,m_{11},\\m_{21},m_{12}}} \tr\big(I-\Pi_{u_1(m_0,m_{11},m_{21},m_{12})}^{B_1}\big)\rho_{u_1(m_0,m_{11},m_{21},m_{12},k_1)}^{B_1}\\ \nonumber
    &\stackrel{\text{(a)}}{=}\frac{2}{2^{R_0 + S_{11} + S_{21}+ S_{12}}}\sum_{\substack{m_0,m_{11},\\m_{21},m_{12}}} \tr\big(I-\sum_{k_1'=1}^{2^{r_1}}\Pi_{u_1(m_0,m_{11},m_{21},m_{12},k_1')}^{B_1}\big)\rho_{u_1(m_0,m_{11},m_{21},m_{12},k_1)}^{B_1}\\ \nonumber
    &=\frac{2}{2^{R_0 + S_{11} + S_{21}+ S_{12}}}\sum_{\substack{m_0,m_{11},\\m_{21},m_{12}}} \tr\bigg\{\biggl(I-\Pi_{u_1(m_0,m_{11},m_{21},m_{12},k_1) }^{B_1}\\ \nonumber
    &\hspace{4.2cm}-\sum_{k_1'\neq k_1}\Pi_{m_{21}),u_1(m_0,m_{11},m_{21},m_{12},k_1')}^{B_1}\biggl)\rho_{u_1(m_0,m_{11},m_{21},m_{12},k_1)}^{B_1}\bigg\}\\  \nonumber
    &\stackrel{\text{(b)}}{\le}  \frac{2}{2^{R_0 + S_{11} + S_{21}+ S_{12}}}\sum_{\substack{m_0,m_{11},\\m_{21},m_{12}}} \tr\big(I-\Pi_{u_1(m_0,m_{11},m_{21},m_{12},k_1)}^{B_1}\big)\rho_{u_1(m_0,m_{11},m_{21},m_{12},k_1)}^{B_1}\\   \nonumber
    &\stackrel{\text{(c)}}{=} \frac{2}{2^{R_0 + S_{11} + S_{21}+ S_{12}}}\sum_{\substack{m_0,m_{11},\\m_{21},m_{12}}} \tr\bigl\{\big(I-\Pi_{0\,u_1(m_0,m_{11},m_{21},m_{12},k_1)}^{B_1}\Pi_{1\,u_1(m_0,m_{11},m_{21},m_{12},k_1)}^{B_1}\big)\\  \nonumber
    &\hspace{10cm}\times\rho_{u_1(m_0,m_{11},m_{21},m_{12},k_1)}^{B_1}\bigl\}\\ \label{hn-1.1}
    &\stackrel{\text{(d)}}{\le} \frac{2}{2^{R_0 + S_{11} + S_{21}+ S_{12}}}\sum_{\substack{m_0,m_{11},\\m_{21},m_{12}}} \tr\big(I-\Pi_{0\,u_1(m_0,m_{11},m_{21},m_{12},k_1)}^{B_1}\big)\rho_{u_0(m_0,m_{11},m_{21}),u_1(m_0,m_{11},m_{21},m_{12},k_1)}^{B_1}\\ \label{hn-1.2}
    &\hspace{0.3cm}+\frac{2}{2^{R_0 + S_{11} + S_{21}+ S_{12}}}\sum_{\substack{m_0,m_{11},\\m_{21},m_{12}}} \tr\big(I-\Pi_{1\,u_1(m_0,m_{11},m_{21},m_{12},k_1)}^{B_1}\big)\rho_{u_1(m_0,m_{11},m_{21},m_{12},k_1)}^{B_1},
\end{align}
where (a) follows from the definition of the non-unique decoding operator, (b) follows by throwing out the operators
\begin{align*}
    \left\{\Pi_{u_1(m_0,m_{11},m_{21},m_{12},k_1')}^{B_1}\right\}_{k'_1\neq k_1},
\end{align*}
(c) follows from the definition and (d) follows from the operator union bound.
We next analyze the second term arose from the Hayashi-Nagaoka inequality Lemma \ref{hayashi-nagaoka}:
\begin{align}
\nonumber
&\frac{4}{2^{R_0 + S_{11} + S_{21}+ S_{12}}}\sum_{\substack{m_0,m_{11},\\m_{21},m_{12}}}
    \sum_{\substack{(m_0',m_{11}',m_{21}',m_{12}')\neq \\(m_0,m_{11},m_{21},m_{12})}}\tr \Pi_{u_1(m_0',m_{11}',m_{21}',m_{12}')}^{B_1}\rho_{u_1(m_0,m_{11},m_{21},m_{12},k_1)}^{B_1}\\ \nonumber
   &\stackrel{\text{(a)}}{=}\frac{4}{2^{R_0 + S_{11} + S_{21}+ S_{12}}}\sum_{\substack{m_0,m_{11},\\m_{21},m_{12}}}\bigg(
    \sum_{\substack{m_{12}'\neq m_{12}\\ \nonumber
    m_0'=m_0,m_{11}'=m_{11},m_{21}'=m_{21}}}\tr \Pi_{u_1(m_0,m_{11},m_{21},m_{12}')}^{B_1}\rho_{u_1(m_0,m_{11},m_{21},m_{12},k_1)}^{B_1}\\
    &+ 
    \sum_{\substack{m_{12}',\\ \nonumber
    m_0'\neq m_0,m_{11}'\neq m_{11},m_{21}'\neq m_{21}}}\tr \Pi_{u_1(m_0,m_{11},m_{21},m_{12}')}^{B_1}\rho_{u_1(m_0,m_{11},m_{21},m_{12},k_1)}^{B_1}\bigg)\\ \nonumber
    &\stackrel{\text{(b)}}{=}\frac{4}{2^{R_0 + S_{11} + S_{21}+ S_{12}}}\sum_{\substack{m_0,m_{11},\\m_{21},m_{12}}}\bigg(
    \sum_{\substack{m_{12}'\neq m_{12}\\ 
    m_0'=m_0,m_{11}'=m_{11}\\m_{21}'=m_{21}}}\tr \sum_{k_1'=1}^{2^{r_1}}\Pi_{u_1(m_0,m_{11},m_{21},m_{12}',k_1')}^{B_1}\rho_{u_1(m_0,m_{11},m_{21},m_{12},k_1)}^{B_1}\\ \nonumber
    &+ 
    \sum_{\substack{m_{12}'\\
    m_0'\neq m_0,m_{11}'\neq m_{11}\\m_{21}'\neq m_{21}}}\tr \sum_{k_1'=1}^{2^{r_1}}\Pi_{u_1(m_0,m_{11},m_{21},m_{12}',k_1')}^{B_1}\rho_{u_1(m_0,m_{11},m_{21},m_{12},k_1)}^{B_1}\bigg)\\ \nonumber
   & \stackrel{\text{(c)}}{=}\frac{4}{2^{R_0 + S_{11} + S_{21}+ S_{12}}}\sum_{\substack{m_0,m_{11},\\m_{21},m_{12}}}\\ \nonumber
   &\bigg(
    \sum_{\substack{m_{12}'\neq m_{12}\\ 
    m_0'=m_0,m_{11}'=m_{11}\\m_{21}'=m_{21}}} \sum_{k_1'=1}^{2^{r_1}}\tr\Pi_{0\, u_1(m_0,m_{11},m_{21},m_{12}',k_1')}^{B_1}\Pi_{1\, u_1(m_0,m_{11},m_{21},m_{12}',k_1')}^{B_1}\rho_{u_1(m_0,m_{11},m_{21},m_{12},k_1)}^{B_1}\\ \nonumber
    &+ 
    \sum_{\substack{m_{12}'\\
    m_0'\neq m_0,m_{11}'\neq m_{11}\\m_{21}'\neq m_{21}}} \sum_{k_1'=1}^{2^{r_1}}\tr\Pi_{0\, u_1(m_0,m_{11},m_{21},m_{12}',k_1')}^{B_1}\Pi_{1\, u_1(m_0,m_{11},m_{21},m_{12}',k_1')}^{B_1}\rho_{u_1(m_0,m_{11},m_{21},m_{12},k_1)}^{B_1}\bigg)\\ \label{hn-2.1}
    &\stackrel{\text{(d)}}{=}\frac{4}{2^{R_0 + S_{11} + S_{21}+ S_{12}}}\sum_{\substack{m_0,m_{11},\\m_{21},m_{12}}}\bigg(
    \sum_{\substack{m_{12}'\neq m_{12}\\ 
    m_0'=m_0,m_{11}'=m_{11}\\m_{21}'=m_{21}}} \sum_{k_1'=1}^{2^{r_1}}\tr\Pi_{1\, u_1(m_0,m_{11},m_{21},m_{12}',k_1')}^{B_1}\rho_{u_1(m_0,m_{11},m_{21},m_{12},k_1)}^{B_1}\\ \label{hn-2.2}
    &+
    \sum_{\substack{m_{12}'\\
    m_0'\neq m_0,m_{11}'\neq m_{11}\\m_{21}'\neq m_{21}}} \sum_{k_1'=1}^{2^{r_1}}\tr\Pi_{0\, u_1(m_0,m_{11},m_{21},m_{12}',k_1')}^{B_1}\rho_{u_1(m_0,m_{11},m_{21},m_{12},k_1)}^{B_1}\bigg)
\end{align}
Regarding the equality (a), note that since the messages $m_0$, $m_{11}$ and $m_{21}$ are encoded into the same variable $u_0$, the relevant events are limited to either all of these three messages being correct or all three being incorrect. In other words, the event where one or two of the messages $(m_0,m_{11},m_{21})$ is decoded correctly and the other(s) decoded incorrectly, does not happen by construction. The equalities (b) and (c) follow from the definitions and (d) is because of $\Pi_0^{B_1}\Pi_1^{B_1}\le \Pi_0^{B_1}\le I$ and $\Pi_0^{B_1}\Pi_1^{B_1}\le \Pi_1^{B_1}\leq I$.
So far we have evaluated the average error probability of the random code $\mathcal{C}_1$ over the uniformly distributed messages as the sum of Eqs. \eqref{hn-1.1},\,\eqref{hn-1.2},\,\eqref{hn-2.1} and \eqref{hn-2.2}, i.e. $P_{e,1}(\mathcal{C}_1)=\eqref{hn-1.1} + \eqref{hn-1.2} + \eqref{hn-2.1} + \eqref{hn-2.2}$. This quantity itself is a random variable since it is a function of the random variable $\mathcal{C}_1$. We next find $\mathbb{E}_{\mathcal{C}_1}(P_{e,1}\left(\mathcal{C}_1)\right)$, the average of the average error probability over the choice of the random code. We evaluate this average for each of the four equations above separately. Starting with \eqref{hn-1.1}, we obtain:
\begin{align*}
    &\mathbb{E}_{\mathcal{C}_1}\left(  \frac{2}{2^{R_0 + S_{11} + S_{21}+ S_{12}}}\sum_{\substack{m_0,m_{11},\\m_{21},m_{12}}} \tr\big(I-\Pi_{0\,u_1(m_0,m_{11},m_{21},m_{12},k_1)}^{B_1}\big)\rho_{u_1(m_0,m_{11},m_{21},m_{12},k_1)}^{B_1} \right)\\
    &= \frac{2}{2^{R_0 + S_{11} + S_{21}+ S_{12}}}\sum_{\substack{m_0,m_{11},\\m_{21},m_{12}}} \tr_{B_1}\mathbb{E}_{\mathcal{C}_1}\left( \big(I-\Pi_{0\,u_1(m_0,m_{11},m_{21},m_{12},k_1)}^{B_1}\big)\rho_{u_1(m_0,m_{11},m_{21},m_{12},k_1)}^{B_1}\right)\\
    &=\frac{2}{2^{R_0 + S_{11} + S_{21}+ S_{12}}}\sum_{\substack{m_0,m_{11},\\m_{21},m_{12}}} \tr_{B_1}\tr_{U_0U_1}\left( (I-\Pi_{0}^{U_0U_1B_1})\rho^{U_0U_1B_1}\right)\\
    &=\frac{2}{2^{R_0 + S_{11} + S_{21}+ S_{12}}}\sum_{\substack{m_0,m_{11},\\m_{21},m_{12}}} \tr\left( (I-\Pi_{0}^{U_0U_1B_1})\rho^{U_0U_1B_1}\right)\\
    &=2\tr\left( (I-\Pi_{0}^{U_0U_1B_1})\rho^{U_0U_1B_1}\right).
\end{align*}
We can find a similar expression for \eqref{hn-1.2} along the same lines as follows:
\begin{align*}
    \mathbb{E}_{\mathcal{C}_1}\left( \text{Eq.} \eqref{hn-1.2} \right)
=
2\tr\left( (I-\Pi_{1}^{U_0U_1B_1})\rho^{U_0U_1B_1}\right).
\end{align*}
We next evaluate the average of the Eq. \eqref{hn-2.1} as follows: 
\begin{align*}
  &\mathbb{E}_{\mathcal{C}_1}\left(  \frac{4}{2^{R_0 + S_{11} + S_{21}+ S_{12}}}\sum_{\substack{m_0,m_{11},\\m_{21},m_{12}}}
    \sum_{\substack{m_{12}'\neq m_{12}\\ 
    m_0'=m_0,m_{11}'=m_{11}\\m_{21}'=m_{21}}}\tr \sum_{k_1'=1}^{2^{r_1}}\Pi_{1\, u_1(m_0,m_{11},m_{21},m_{12}',k_1')}^{B_1}\rho_{u_0(m_0,m_{11},m_{21})}^{B_1} \right)\\
   &= \frac{4}{2^{R_0 + S_{11} + S_{21}+ S_{12}}}\sum_{\substack{m_0,m_{11},\\m_{21},m_{12}}}
    \sum_{\substack{m_{12}'\neq m_{12}\\ 
    m_0'=m_0,m_{11}'=m_{11}\\m_{21}'=m_{21}}}\sum_{k_1'=1}^{2^{r_1}}\tr \mathbb{E}_{\mathcal{C}_1}\left( \Pi_{1\,u_1(m_0,m_{11},m_{21},m_{12}',k_1')}^{B_1}\rho_{u_0(m_0,m_{11},m_{21})}^{B_1} \right)\\
   & = \frac{4}{2^{R_0 + S_{11} + S_{21}+ S_{12}}}\sum_{\substack{m_0,m_{11},\\m_{21},m_{12}}}
    \sum_{\substack{m_{12}'\neq m_{12}\\ 
    m_0'=m_0,m_{11}'=m_{11}\\m_{21}'=m_{21}}}\sum_{k_1'=1}^{2^{r_1}}\tr \Pi^{U_0U_1B_1}_{1}\rho^{U_1-U_0-B_1}\\
    &= 4\times (2^{S_{12}}-1)\times 2^{r_1}\tr\Pi^{U_0U_1B_1}_{1}\rho^{U_1-U_0-B_1}\\
    &\leq 4\times 2^{S_{12}+r_1}\tr\Pi^{U_0U_1B_1}_{1}\rho^{U_1-U_0-B_1},
\end{align*}
where the last inequality follows from adding an extra term corresponding to $m_{12}'= m_{12}$ to the summation. We obtain the average of Eq. \eqref{hn-2.2} along the same lines as follows:
\begin{align*}
    \mathbb{E}_{\mathcal{C}_1}\left( \text{Eq.} \eqref{hn-2.2} \right) \leq 
    4\times 2^{R_0 + S_{11} + S_{21}+ S_{12}+r_1}\tr\Pi^{U_0U_1B_1}_{0}(\rho^{U_0U_1}\otimes\rho^{B_1}).
\end{align*}
Bringing everything together, we obtain:
\begin{align*}
    \mathbb{E}_{\mathcal{C}_1}\left(P_{e,1}(\mathcal{C}_1)\right)&=\mathbb{E}_{\mathcal{C}_1}\left(\text{Pr}\{\hat{M}_{0b_1}\neq M_0,\hat{M}_{1b_1}\neq M_1|\mathcal{C}_1\}\right)\\
    &\le2\tr\left( (I-\Pi_{1}^{U_0U_1B_1})\rho^{U_0U_1B_1}\right)
    +
    2\tr\left( (I-\Pi_{0}^{U_0U_1B_1})\rho^{U_0U_1B_1}\right)\\
    &\hspace{0.3cm}+4\times 2^{S_{12}+r_1}\tr\Pi^{U_0U_1B_1}_{1}\rho^{U_1-U_0-B_1}
    +
    4\times 2^{R_0 + S_{11} + S_{21}+ S_{12}+r_1}\tr\Pi^{U_0U_1B_1}_{0}(\rho^{U_0U_1}\otimes\rho^{B_1})\\
    &\leq 4\tr\left( (I-\Pi_{1}^{U_0U_1B_1})\rho^{U_0U_1B_1}\right)
    +4\times 2^{S_{12}+r_1}\tr\Pi^{U_0U_1B_1}_{1}\rho^{U_1-U_0-B_1}\\
    &\hspace{1cm}+4\tr\left( (I-\Pi_{0}^{U_0U_1B_1})\rho^{U_0U_1B_1}\right)
+4\times 2^{R_0 + S_{11} + S_{21}+ S_{12}+r_1}\tr\Pi^{U_0U_1B_1}_{0}(\rho^{U_0U_1}\otimes\rho^{B_1})\\
&\stackrel{\text{(a)}}{\le} 4 \times 2^{\alpha(S_{12}+r_1)}2^{-\alpha D_{1-\alpha}\left(\mathcal{E}_1^{U_0B_1}(\rho^{U_0U_1B_1})\|\mathcal{E}^{B_1}(\rho^{U_1-U_0-B_1})\right)}\\
&+ 4\times 2^{\alpha(R_0 + S_{11} + S_{21}+ S_{12}+r_1)}2^{-\alpha D_{1-\alpha}\left(\mathcal{E}^{B_1}(\rho^{U_0U_1B_1})\|\rho^{U_0U_1}\otimes\rho^{B_1}\right)}\\
&\stackrel{\text{(b)}}{\le} 4 \times \nu_1^{\alpha} 2^{\alpha(S_{12}+r_1)} 2^{-\alpha \widetilde{D}_{1-\alpha}\left(\rho^{U_0U_1B_1}\|\mathcal{E}^{B_1}(\rho^{U_1-U_0-B_1})\right)}\\
&+ 4\times \nu^{\alpha} 2^{\alpha(R_0 + S_{11} + S_{21}+ S_{12}+r_1)} 2^{-\alpha \widetilde{D}_{1-\alpha}\left(\rho^{U_0U_1B_1}\|\rho^{U_0U_1}\otimes\rho^{B_1}\right)}\\
&\stackrel{\text{(c)}}{\le} 4 \times \nu_1^{\alpha} 2^{\alpha(S_{12}+r_1)} 2^{-\alpha \tilde{I}_{1-\alpha}^{\downarrow}\left(U_1;B_1|U_0\right)_{\rho^{U_0U_1B_1}|\rho^{U_0U_1}}}\\
&+4\times \nu^{\alpha} 2^{\alpha(R_0 + S_{11} + S_{21}+ S_{12}+r_1)} 2^{-\alpha\tilde{I}_{1-\alpha}^{\uparrow}(U_0U_1;B_1)_{\rho^{U_0U_1B_1}}},
\end{align*}
where (a) follows from Lemma \ref{hp-testing}, (b) from Lemma \ref{petz-sandwich}, with $\nu$ and $\nu_1$ being the number of distinct eigenvalues of $\mathcal{E}^{B_1}$ and $\{\mathcal{E}^{B_1}_u\}_u$, respectively, and (c) from the definitions of the R\'enyi mutual information and R\'enyi confitional mutual information.

\bigskip

\textbf{\emph{Receiver $B_2$:}}
The construction of the POVM and the analysis of error probability for the second receiver follow along the similar lines, resulting in:
\begin{align*}
 \mathbb{E}_{\mathcal{C}_1}\left(P_{e,2}(\mathcal{C}_1)\right)&=\mathbb{E}_{\mathcal{C}_1}\left(\text{Pr}\{\hat{M}_{0b_2}\neq M_0,\hat{M}_{2b_2}\neq M_2|\mathcal{C}_1\}\right)\\
 &\leq 4 \times \mu_1^\alpha 2^{\alpha(S_{22}+r_2)} 2^{-\alpha \tilde{I}_{1-\alpha}^{\downarrow}\left(U_1;B_2|U_0\right)_{\rho^{U_0U_1B_2}|\rho^{U_0U_1}}}\\
&\hspace{2cm}+4\times \mu^\alpha 2^{\alpha(R_0 + S_{11} + S_{21}+ S_{22}+r_2)} 2^{-\alpha\tilde{I}_{1-\alpha}^{\uparrow}(U_0U_1;B_2)_{\rho^{U_0U_1B_2}}}
\end{align*}

We have successfully derived our one-shot bounds. This means that there exists at least one one-shot code such that the average error probabilities of the receivers are bounded as above. Now, we proceed to recover the asymptotic i.i.d. result as a special case of these bounds.

\subsection{Asymptotic Analysis}\label{marton-asymptotic}
As mentioned before, our proof consists in creating a one-shot code and then determining the rate region through asymptotic analysis.
 The key observation is that the number of distinct eigenvalues will vanish because they are only polynomial and that R\'enyi relative entropy tends to von Nuemann relative entropy when its parameter tends to $1$.

\begin{proofof}[Theorem \ref{Marton}]
In the asymptotic setting, the channel is assumed to be memoryless, with each use being independent of the others; the codewords are also generated in an i.i.d. fashion. Leveraging Eq. \eqref{max-asymp}, we can conclude that if the condition $r_1 + r_2 \geq I(U_1;U_2\vert U_0)$ holds true, then a pair of codewords in each product bin can be identified, such that they mimic a joint distribution. This is further supported by the asymptotic mutual covering lemma \cite{1056302}, \cite[Lemma 8.1]{elgamal-kim}.

The upper bounds on the expectation of the average error probability imply the existence of at least one good code $\mathcal{C}_1$ that satisfies these bounds. Consequently, we can assess the upper bounds in the following manner:
\begin{align*}
    P_{e,1}^{(n)}(\mathcal{C}_n)&\leq 
    4  (n+1)^{\alpha d_{U_0}(d_{B_1}+2)(d_{B_1}-1)/2} 2^{\alpha n(S_{12}+r_1)} 2^{-\alpha \tilde{I}_{1-\alpha}^{\downarrow}\left(U_1^n;B_1^n|U_0^n\right)_{(\rho^{U_0U_1B_1})^{\otimes n}\vert(\rho^{U_0U_1})^{\otimes n}}}\\
&+4 (n+1)^{\alpha (d_{B_1}-1)} 2^{\alpha n(R_0 + S_{11} + S_{21}+ S_{12}+r_1)} 2^{-\alpha\tilde{I}_{1-\alpha}^{\uparrow}(U_0^nU_1^n;B_1^n)_{(\rho^{U_0U_1B_1})^{\otimes n}}},\\
  P_{e,2}^{(n)}(\mathcal{C}_n)&\leq 
    4  (n+1)^{\alpha d_{U_0}(d_{B_2}+2)(d_{B_2}-1)/2} 2^{\alpha n(S_{22}+r_2)} 2^{-\alpha \tilde{I}_{1-\alpha}^{\downarrow}\left(U_1^n;B_2^n|U_0^n\right)_{(\rho^{U_0U_1B_2})^{\otimes n}\vert(\rho^{U_0U_1})^{\otimes n}}}\\
&+4 (n+1)^{\alpha (d_{B_2}-1)} 2^{\alpha n(R_0 + S_{11} + S_{21}+ S_{22}+r_2)} 2^{-\alpha\tilde{I}_{1-\alpha}^{\uparrow}(U_0^nU_1^n;B_2^n)_{(\rho^{U_0U_1B_2})^{\otimes n}}},
\end{align*}
where we have employed the polynomial bounds on the number of distinct eigenvalues from Proposition \ref{pinching-asymptotic}. We further obtain (notice that the number of distinct eigenvalues vanish because they are only polynomial and that R\'enyi relative entropy tends to von Nuemann relative entropy when its parameter tends to $1$):
\begin{align*}
  \lim_{n\to\infty }-\frac{1}{n}\log P_{e,1}^{(n)}(\mathcal{C}_n)\ge
  \min\bigg(&\alpha \big(\tilde{I}_{1-\alpha}^{\downarrow}\left(U_1;B_1|U_0\right)_{\rho^{U_0U_1B_1}|\rho^{U_0U_1}}-(S_{12}+r_1)\big),\\
  &\alpha \big(\tilde{I}_{1-\alpha}^{\uparrow}(U_0U_1;B_1)_{\rho^{U_0U_1B_1}}-(R_0 + S_{11} + S_{21}+ S_{12}+r_1)\big)
  \bigg).
\end{align*}

\begin{align*}
  \lim_{n\to\infty }-\frac{1}{n}\log P_{e,2}^{(n)}(\mathcal{C}_n)\ge
  \min\bigg(&\alpha \big(\tilde{I}_{1-\alpha}^{\downarrow}\left(U_1;B_2|U_0\right)_{\rho^{U_0U_1B_2}|\rho^{U_0U_1}}-(S_{22}+r_2)\big),\\
  &\alpha \big(\tilde{I}_{1-\alpha}^{\uparrow}(U_0U_1;B_2)_{\rho^{U_0U_1B_2}}-(R_0 + S_{11} + S_{21}+ S_{22}+r_2)\big)
  \bigg)
\end{align*}
It now follows that as $n\to\infty$ and $\alpha\to 0$, there exists a sequence of codes $\mathcal{C}_n$ such that $P_{e,1}^{(n)}(\mathcal{C}_n)$ and $P_{e,1}^{(n)}(\mathcal{C}_n)$ vanish, if
\begin{equation}
    \begin{aligned}\label{FM-theorem 3}
           r_1+r_2&\geq I(U_1;U_2\vert U_0)_{\rho}, \\
    S_{12}+r_1&\geq I(U_1;B_1|U_0)_{\rho},\\
    R_0 + S_{11} + S_{21}+ S_{12}+r_1 &\geq I(U_0U_1;B_1)_{\rho},\\
     S_{22}+r_2&\geq I(U_2;B_2|U_0)_{\rho},\\
    R_0 + S_{11} + S_{21}+ S_{22}+r_2 &\geq I(U_0U_2;B_2)_{\rho},
    \end{aligned}
\end{equation}
where the mutual information quantities are calculated for the state in the statement of the Theorem \ref{Marton}. Using the Fourier–Motzkin elimination procedure, we remove $S_{11},S_{12},S_{21},S_{22}$ and as well as $r_1$ and $r_2$. Then it follows from the union bound $P_{e}^{(n)}(\mathcal{C}_n)\le P_{e,1}^{(n)}(\mathcal{C}_n)+P_{e,2}^{(n)}(\mathcal{C}_n)$ that the probability of decoding error tends to zero as $n\to \infty$ if the inequalities in Theorem are satisfied. This completes the proof of Theorem \ref{Marton}. 
\end{proofof}

\bigskip

\begin{remark}
   It's worth noting that the structure of the inequalities in \eqref{FM-theorem 3} resembles their classical counterpart, replacing quantum systems $B_1, B_2$ with their corresponding classical counterparts. On the other hand, using the complicated overcounting technique, papers \cite{savov-wilde} and \cite{7412732} derived a set of inequalities that, although resulting in the same final rate region after Fourier–Motzkin elimination, have a different structure. It is uncertain whether this difference in the structure of the initial inequalities will always lead to the same final inequalities.
\end{remark}

\begin{remark}
    Note that the rate region could be achieved without rate-splitting. Generally, if the rate triple $(R_0,R_1,R_2)$ is achievable, so is the triple $(R_0-\Delta_1-\Delta_2,R_1+\Delta_1,R_2+\Delta_2)$ for any $\Delta_1,\Delta_2\ge 0$. This does not depend on particular coding scheme such as superpostion coding. We can always split the the rates of the private messages and demand part of the common message to be sacrificed to accommodate for this private message. In other words, the messages deemed common message will be decoded by the receiver(s), so we choose to use part of that potential to send additional private information. Of course the other receiver will decode private information of the other receiver which is useless for it, and its common message rate is reduced because of this unintended message. This happens symmetrically for both receivers. On the other hand, it does not make sense to reduce the private rate to send common message since common message encoded into that private part will not be decoded by both receivers, so it is not common anymore. So one can derive a region without rate-splitting, then use this trick to get the rate region corresponding to rate-splitting. In summary, when the rate of the common message is reduced and some private message is sent to one of the receivers, the rate of the common message for the other receiver reduces, too. However, the other receiver recovers part of the private message intended to the first receiver. This is useless information for the second receiver and it may discard that information. It is important to understand that it is not possible to convert a common message to private messages for two receivers. Compare this to what might be considered its quantum counterpart: A GHZ state and two EPR pairs, where in certain sense they resemble a common and two private messages, respectively. It is known that a two EPR pairs can always be converted to a GHZ state, but not the other way around.
\end{remark}

\begin{remark}
    Note that the above scheme combine rate splitting with rate transfer: In rate splitting each message is represented by two or more independent sub-messages, while in rate transfer, sub-messages corresponding to different messages are encoded together into a single variable. As such, some receivers will successfully decode some sub-messages which are not intended for them. This might seem waste of rate, but in some cases that we are aware of, rate splitting with rate transfer gives equivalent rate regions whose converse is easier to establish compared to the region derived without these techniques. One example is the two-receiver broadcast channel with two-degraded message sets whose capacity is proven using these techniques \cite{korner-marton}.
\end{remark}

\section{Three-receiver multilevel quantum broadcast channel with two-degraded message set}
\label{multilevel-two-degraded}

\begin{figure}[t]
\includegraphics[width=15cm]{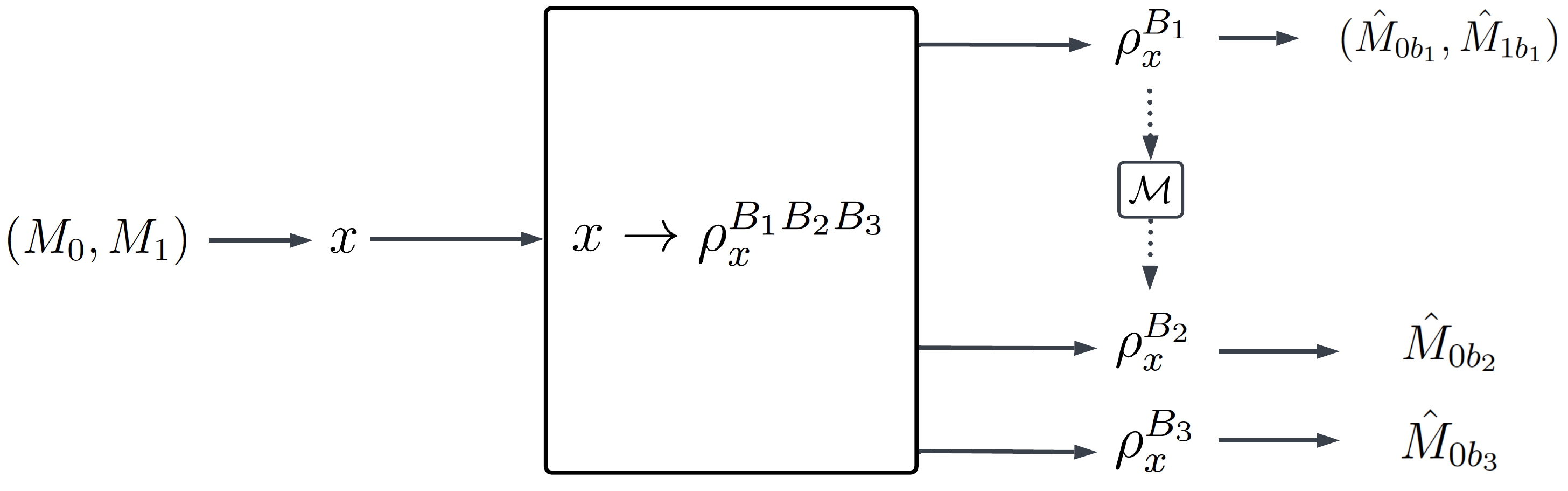}
\centering
\caption{three-receiver multilevel quantum broadcast channel with two-degraded message set. A common message $M_0$ is intended for all receivers and another message $M_1$ is intended solely for receiver $B_1$. $M_{0b_i}$ denotes the estimate made by receiver $B_i$ about the message $M_0$. Similarly, $M_{1b_1}$ denotes the estimate made by receiver $B_1$ about the the message $M_1$. Here, receiver $B_2$ is degradable with respect to the receiver $B_1$, i.e. there exists a cptp map $\mathcal{M}^{B_1\to B_2}$ such that $\rho_x^{B_2}=\mathcal{M}(\rho_x^{B_1})$ for any $x\in\mathcal{X}$.}
\label{multi-level.two-degraded}
\end{figure}

For pedagogical reasons, we start with the three-receiver multilevel quantum broadcast channel $x\to\rho_x^{B_1B_2B_3}$ with two-degraded message sets. Here, receiver $B_2$ is degradable with respect to the receiver $B_1$, i.e. there exists a cptp map $\mathcal{M}^{B_1\to B_2}$ such that $\rho_x^{B_2}=\mathcal{M}(\rho_x^{B_1})$ for any $x\in\mathcal{X}$. A transmitter wants to send a common message $m_0\in[1:2^{nR_0}]$ to all receivers, and a private (or individualized) message $m_1\in[1:2^{nR_1}]$ solely to receiver $B_1$. This scenario is depicted in Fig. \ref{multi-level.two-degraded}.
This model is particularly intriguing as it is the only model for which the capacity of its classical restriction is known \cite{Nair-Elgamal}. Additionally, in the quantum case, the converse region only encounters the common issue of conditioning on quantum systems.

We will present our results in the asymptotic setting, considering the scenario, where the channel can be utilized $n$ times, with $n\to\infty$. For this reason, we don't define a one-shot explicitly. Instead, we define an asymptotic code as follows.
\begin{definition}\label{code-def-multilevel}
A $(2^{nR_0},2^{nR_1},n,P^{(n)}_e)$ two-degraded message set code $\mathcal{C}_n$ for the three-receiver multilevel quantum broadcast channel $x\to\rho_x^{B_1B_2B_3}$ consists of the following components:
\begin{itemize}
    \item A pair of messages $(M_0,M_1)$ uniformly distributed over $[1:2^{nR_0}]\times[1:2^{nR_1}]$, where $R_0$ and $R_1$ represent the rates of the messages $M_0$ and $M_1$, respectively.
    \item An encoder that assigns a codeword $x^n(m_0,m_1)$ to each message pair $(m_0,m_1,m_2)\in[1:2^{nR_0}]\times[1:2^{nR_1}]$,
    \item Three POVMs, $\Lambda_{m_0m_1}^{B_1^{n}}$ on $B_1^{\otimes n}$, $\Lambda_{m_0}^{B_2^{n}}$ on $B_2^{\otimes n}$ and $\Lambda_{m_0^{n}}^{B_3}$ on $B_3^{\otimes n}$ (indicating which messages are intended for decoding by each receiver) which satisfy
    \begin{align}
\tr{\rho_{x^n(m_0,m_1)}^{B_1^{n}B_2^{n}B_3^{n}}\left(\Lambda_{m_0m_1}^{B_1^{n}}\otimes\Lambda_{m_0}^{B_2^{n}}\otimes\Lambda_{m_0}^{B_3^{n}}\right)}\geq 1-P^{(n)}_e(\mathcal{C}_n),
    \end{align}
    for every $(m_0,m_1)\in[1:2^{nR_0}]\times[1:2^{nR_1}]$. $P^{(n)}_e(\mathcal{C}_n)$ is the average probability of error defined as follows:
    \begin{align*}
        P^{(n)}_e(\mathcal{C}_n)= \text{Pr}\{\hat{M}_{0b_1},\hat{M}_{0b_2},\hat{M}_{0b_3}\neq M_0 \,\text{or}\,\hat{M}_{1b_1}\neq M_1\},
    \end{align*}
 where $\hat{M}_{jb_i}$ represents the estimate made by receiver $B_i$ for the message $j$.
\end{itemize}
A rate pair $(R_0,R_1)$ is said to be achievable if there exists a sequence of $(2^{nR_0},2^{nR_1},n,P^{(n)}_e)$ two-degraded message set codes $\mathcal{C}_n$ such that $P^{(n)}_e(\mathcal{C}_n)\to 0$ as $n\rightarrow \infty$. The capacity region is the closure of the set of all such achievable rate triples. 
\end{definition}

\begin{theorem}\label{two-degraded-multilevel}
    We assume a transmitter and three receivers have access to many independent uses of a three-receiver multilevel quantum broadcast channel $x\to\rho_x^{B_1B_2B_3}$, where a quantum channel $\mathcal{M}$ exists such that $\mathcal{M}(\rho_x^{B_1})=\rho_x^{B_2}$ for all $x$. The set of rate pairs $(R_0,R_1)$ is asymptotically achievable for the aforementioned two-degraded message set if the inequalities
    \begin{align*}
            R_0&\leq\min\{I(U;B_2),I(V,B_3)\},\\
            R_1&\leq I(X;B_1|U),\\
            R_0+R_1&\leq I(V;B_3) + I(X;B_1|V), 
    \end{align*}
        \label{two-degraded-1}
hold
    for some probability distribution $p(u,v,x)$ satisfying the Markov chain $U-V-X$, resulting in the cq-state:
    \begin{align*}
\rho^{UXVB_1B_2B_3}=\sum_{u,v,x}p(u)p(v|u)p(x|v)\ketbra{u}\otimes\ketbra{v}\otimes\ketbra{x}\otimes\rho^{B_1B_2B_3}_x.
    \end{align*}
\end{theorem} 
The proof of this theorem will be presented as the asymptotic analysis of the one-shot code we construct below. Therefore, the proof is deferred until Subsection \ref{asymptotic-2}.

\subsection{One-shot code construction}

The asymptotic analysis of the code developed here recovers Theorem \ref{two-degraded-multilevel}.
Fix a probability distribution $p(u)p(v|u)p(x|v)$, i.e. the random variables $(U,V,X)$ satisfy the Markov chain $U-V-X$. The goal is to transmit a common message to all receivers and a private (or individualized) message to receiver $B_1$.
We represent these messages respectively by random variables $M_0, M_1$ uniformly distributed in $[1:2^{R_0}]$ and $[1:2^{R_1}]$, respectively. 

\subsubsection{Rate splitting}
We split the private message $M_1$ into two parts, $M_{11}$ and $M_{12}$, with rates $S_1$ and $S_2$, respectively, resulting in $R_1=S_1+S_2$. These messages are uniformly distributed on their respective ranges, i.e. $s_1\in [1:2^{S_1}]$ and $s_2\in [1:2^{S_2}]$. The common message $m_0$ is represented by $U$, the pair $(m_0,s_1)$ is represented by $V$ and $(m_0,s_1,s_2)$ is represented by $X$.

\subsubsection{Codebook generation}

A code $\mathcal{C}_1$ is given as the following tuples of elements:
\begin{align*}
    \{u(m_0)\}_{m_0\in [1:2^{R_0}]} &\subset \mathcal{U},\\
    \{v(m_0,s_1)\}_{m_0\in [1:2^{R_0}],s_1\in [1:2^{S_{1}}]} &\subset \mathcal{V},\\
    \{x(m_0,s_1,s_2)\}_{m_0\in [1:2^{R_0}],s_1\in [1:2^{S_{1}}],s_2\in [1:2^{S_{2}}]} 
&\subset \mathcal{X}.
\end{align*}
This codebook is generated as follows: We generate $2^{R_0}$ sequences 
$\{u(m_0)\}_{m_0 \in [1:2^{R_0}]}$ i.i.d according to the pmf $p(u)$.
For each message $m_0$ with corresponding codeword $u(m_0)=u$, we generate $2^{S_{1}}$ elements
$\{v(m_0,s_1)\}_{s_1 \in [1:2^{S_{1}}]}$ conditionally independent subject to $p(v|U=u)$.
For each pair $(m_0,s_1)$, when $v(m_0,s_1)=v$, we generate $2^{S_{2}}$ elements
$\{x(m_0,s_1,s_2)\}_{s_2 \in [1:2^{S_{2}}]}$ conditionally independent subject to $p(x|V=v)$.
Note that the code $\mathcal{C}_1$ is a random variable.

\subsubsection{Encoding} To transmit the message $(m_0,m_1)\in [1:2^{R_0}]\times [1:2^{R_{1}}]$, the sender indicates the corresponding $(s_1,s_2)\in[1:2^{S_1}]\times [1:2^{S_{2}}]$ for the message $m_1$, and transmits $x(m_0,s_1,s_2)$ over the channel.

\subsubsection{Decoding and analysis of error probability}

Upon sending the codeword $x(m_0,s_1,s_2)$ over the channel, each receivers obtains its share of $\rho^{B_1B_2B_3}_{x(m_0,s_1,s_2)}$.
Subsequently, each receiver applies a POVM on its system to decode the intended message(s). 
 Receiver $B_1$ finds $(m_0,m_1)$ by decoding $X$, receiver $B_2$ finds $m_0$ by decoding $U$, and receiver $B_3$ finds $m_0$ by decoding $U$ non-uniquely via $V$. These POVMs are specifically designed to extract messages simultaneously. We now describe POVM construction and error probability analysis for each receiver.

\bigskip

\textbf{\emph{Receiver $B_1$:}} We begin with defining three pinching maps in a nested fashion.
We define $\mathcal{E}^{B_1}$ as the pinching on $B_1$
with respect to the spectral decomposition of $\rho^{B_1}$. For the spectral decomposition of $\cE^{B_1}(\rho^{XV- U-B_1})$, we define a pinching on $U B_1$
by $\mathcal{E}^{UB_1}_1=\sum_u\ketbra{u}\otimes\cE_{1|u}^{B_1}$, where for every $u\in\mathcal{U}$, $\cE_{1|u}^{B_1}$ is the pinching map with respect to the spectrum of the operator $\cE^{B_1}(\rho_u^{B_1})$.
Finally, we define a pinching map on $U V B_1$
for the spectral decomposition of $\cE_{1}^{UB_1}(\rho^{UX-V-B_1})$
by $\mathcal{E}^{UVB_1}_2=\sum_{u,v}\ketbra{u}\otimes\ketbra{v}\otimes\cE_{2|u,v}^{B_1}$, where for every $u\in\mathcal{U},v\in\mathcal{V}$, $\cE_{2|u,v}^{B_1}$ is the pinching map with respect to the spectrum of the operator $\cE_{1|u}^{B_1}(\rho_v^{B_1})$.
By employing these pinching maps, we define three projectors as follows:
\begin{align*}
T_{2}^{UVXB_1}&\coloneqq\{ \cE_{2}^{UVB_1}(\rho^{UVXB_1}) \ge 2^{S_{2}}\cE_{1}^{UB_1}(\rho^{UX- V-B_1})\} ,\\
T_{1}^{UVXB_1}&\coloneqq\{ \cE_{1}^{UB_1}(\rho^{UVX B_1}) \ge 2^{S_{1}+S_{2}}\cE^{B_1}(\rho^{V X- U-B_1})\}, \\
T_{0}^{UVXB_1}&\coloneqq\{ \cE^{B_1}(\rho^{UVX B_1}) \ge 2^{R_0+S_{1}+S_{2}}\rho^{UVX}\otimes \rho^{B_1}\} .
\end{align*}

We further define the projectors:
\begin{align*}
T_{2\,u,v,x}^{B_1}&\coloneqq\bra{u,v,x}T_{2}\ket{u,v,x},\\
T_{1\, u,v,x}^{B_1}&\coloneqq\bra{u,v,x}T_1\ket{u,v,x},\\
T_{0\,u,v,x}^{B_1}&\coloneqq\bra{u,v,x}T_{0}\ket{u,v,x},
\end{align*}
and the operator
\begin{align*}
T_{u,v,x}&\coloneqq T_{0\,u,v,x}^{B_1}T_{1\,u,v,x}^{B_1}T_{2\,u,v,x}^{B_1}.
\end{align*}
Note that the above-mentioned projectors and operator are not specific to any code. Nonetheless, we can utilize them to construct POVMs for any given code, particularly our random code $\mathcal{C}_1$.

\medskip
We employ these projectors to construct a pretty-good measurement POVM, which serves as the decoder of $B_1$ with respect to $\mathcal{C}_1$ as follows:
\begin{align*}
&\Delta_{m_0,s_1,s_2}(\mathcal{C}_1)\coloneqq\\
&\left(\sum_{m_0',s_1',s_2'} T_{x(m_0',s_1',s_2')}\right)^{-\frac{1}{2}}
T_{x(m_0,s_1,s_2)}
\left(\sum_{m_0',s_1',s_2'} T_{x(m_0',s_1',s_2')}\right)^{-\frac{1}{2}}
\end{align*} 

The average error probability of the code $\mathcal{C}_1$ is bounded as follows. Recall that we use the variable $\hat{M}_{tb_i}$ to denote the estimate made by receiver $B_i$ about that message $t$.
Using the above-mentioned POVM, we evaluate $B_1$'s average decoding error probability as follows:
\begin{align}
\nonumber
 P_{e,1}(\mathcal{C}_1)&=\text{Pr}\{\hat{M}_{0b_1}\neq M_0,\hat{M}_{1b_1}\neq M_1|\mathcal{C}_1\}\\ \nonumber
& = \frac{1}{2^{R_0+S_{1}+S_{2}}}\sum_{m_0,s_1,s_2} \tr (I-\Delta_{m_0,s_1,s_2}(\mathcal{C}_1))
\rho_{x(m_0,s_1,s_2)}^{B_1}\\\nonumber
\le & \frac{1}{2^{R_0+S_{1}+S_{2}}}
 \sum_{m_0,s_1,s_2} \tr \Big(
 2(I- T_{x(m_0,s_1,s_2)})\\ \nonumber
 &+4
 \sum_{(m_0',s_1',s_2')\neq (m_0,s_1,s_2)} T_{x(m_0',s_1',s_2')}
\Big) \rho_{x(m_0,s_1,s_2)}^{B_1}\\
\label{hn-1}
&=
 \frac{2}{2^{R_0+S_{1}+S_{2}}}\sum_{m_0,s_1,s_2} 
 \tr (I- T_{x(m_0,s_1,s_2)}) \rho_{x(m_0,s_1,s_2)}^{B_1} \\
 \label{hn-2}
&+\frac{4}{2^{R_0+S_{1}+S_{2}}}
\sum_{m_0,s_1,s_2} \sum_{(m_0',s_1',s_2')\neq (m_0,s_1,s_2)}
 \tr  T_{x(m_0',s_1',s_2')}
\rho_{x(m_0,s_1,s_2)}^{B_1} , 
\end{align}
where the inequality follows from Hayashi-Nagaoka operator inequality. 
We analyze each term resulting from Hayashi-Nagaoka inequality Lemma \ref{hayashi-nagaoka} separately. For the first term Eq. \eqref{hn-1}, we obtain:
\begin{align}
\nonumber
&\frac{2}{2^{R_0+S_{1}+S_{2}}}\sum_{m_0,s_1,s_2} 
 \tr (I- T_{x(m_0,s_1,s_2)}) \rho_{x(m_0,s_1,s_2)}^{B_1}\\
 \label{h-0}
&\le 
 \frac{2}{2^{R_0+S_{1}+S_{2}}}\sum_{m_0,s_1,s_2} \Big(
 \tr 
 (I-T_{0\,x(m_0,s_1,s_2)})\rho_{x(m_0,s_1,s_2)}^{B_1}\\
 \label{h-1}
& \hspace{3.5cm}+\tr(I-T_{1\,x(m_0,s_1,s_2)})\rho_{x(m_0,s_1,s_2)}^{B_1}\\
\label{h-2}
&\hspace{3.5cm}+\tr(I-T_{2\,x(m_0,s_1,s_2)})
\rho_{x(m_0,s_1,s_2)}^{B_1} \Big)
\end{align}
where the inequality follows from the operator union bound, stating that for operators $\{T_i\}_i$ where $0\le T_i\le I$ for all $i$, we have $I-\sum_{i}T_i\leq \sum_i(I-T_i)$.

Next, we find an upper bound on the second term Eq. \eqref{hn-2} as follows:
\begin{align}
\nonumber
&\frac{4}{2^{R_0+S_{1}+S_{2}}}
\sum_{m_0,s_1,s_2} \sum_{(m_0',s_1',s_2')\neq (m_0,s_1,s_2)}
 \tr  T_{x(m_0',s_1',s_2')}
\rho_{x(m_0,s_1,s_2)}^{B_1}\\ \nonumber
=&\frac{4}{2^{R_0+S_{1}+S_{2}}}
\sum_{m_0,s_1,s_2} 
\sum_{s_2'\neq s_2}
 \tr  T_{x(m_0,s_1,s_2')}
\rho_{x(m_0,s_1,s_2)}^{B_1}\\ \nonumber
&+\frac{4}{2^{R_0+S_{1}+S_{2}}}
\sum_{m_0,s_1,s_2} 
\sum_{s_1' \neq s_1 ,s_2'}
 \tr  T_{x(m_0,s_1',s_2')}
\rho_{x(m_0,s_1,s_2)}^{B_1}\\ \nonumber
&+\frac{4}{2^{R_0+S_{1}+S_{2}}}
\sum_{m_0,s_1,s_2} 
\sum_{m_0' \neq m_0 ,s_1' ,s_2'}
 \tr  T_{x(m_0',s_1',s_2')}
\rho_{x(m_0,s_1,s_2)}^{B_1}\\ 
\label{n-2}
\le &\frac{4}{2^{R_0+S_{1}+S_{2}}}
\sum_{m_0,s_1,s_2} 
\sum_{s_2'\neq s_2}
 \tr  T_{2\,x(m_0,s_1,s_2')}
\rho_{x(m_0,s_1,s_2)}^{B_1}\\
\label{n-1}
&+\frac{4}{2^{R_0+S_{1}+S_{2}}}
\sum_{m_0,s_1,s_2} 
\sum_{s_1' \neq s_1 ,s_2'}
 \tr  T_{1\,x(m_0,s_1',s_2')}
\rho_{x(m_0,s_1,s_2)}^{B_1}\\
\label{n-0}
&+\frac{4}{2^{R_0+S_{1}+S_{2}}}
\sum_{m_0,s_1,s_2} 
\sum_{m_0' \neq m_0 ,s_1' ,s_2'}
 \tr  T_{0\,x(m_0',s_1',s_2')}\rho_{x(m_0,s_1,s_2)}^{B_1},
\end{align}
where the inequality follows because $0\leq T_{u,v,x}\le T_{i\, u,v,x}\le I$ for $i=0,1,2$.
So far we have evaluated the average error probability of the random code $\mathcal{C}_1$ over the uniformly distributed messages as the sum of Eqs. \eqref{h-0},\,\eqref{h-1},\,\eqref{h-2},\, \eqref{n-2},\, \eqref{n-1} and \eqref{n-0}, i.e. $P_{e,1}(\mathcal{C}_1)=\eqref{h-0}+\eqref{h-1}+\eqref{h-2}+ \eqref{n-2}+ \eqref{n-1} + \eqref{n-0}$. This quantity itself is a random variable since it is a function of the random variable $\mathcal{C}_1$. We next find $\mathbb{E}_{\mathcal{C}_1}(P_{e,1}\left(\mathcal{C}_1)\right)$, the average of the average error probability over the choice of the random code. We evaluate this average for each of the six equations above separately. Starting with \eqref{h-0}, we obtain:
\begin{align*}
\mathbb{E}_{\mathcal{C}_1}\left( \text{Eq.} \eqref{h-0} \right)
 &=\frac{2}{2^{R_0+S_{1}+S_{2}}}\sum_{m_0,s_1,s_2} 
 \E_{\mathcal{C}_1}\tr 
 \Big((I-T_{0\,x(m_0,s_1,s_2)}
 \Big)
\rho_{x(m_0,s_1,s_2)}^{B_1} \\
&=2\tr (I-T_0^{UVXB_1}) \rho^{UVXB_1}.
\end{align*}
We find similar expressions for Eqs. \eqref{h-1} and \eqref{h-2} as follows:
\begin{align*}
    \mathbb{E}_{\mathcal{C}_1}\left( \text{Eq.} \eqref{h-1} \right)&=2\tr (I-T_1^{UVXB_1}) \rho^{UVXB_1},\\
       \mathbb{E}_{\mathcal{C}_1}\left( \text{Eq.} \eqref{h-2} \right)&=2\tr (I-T_2^{UVXB_1}) \rho^{UVXB_1}.
\end{align*}

We next evaluate the average of Eq. \eqref{n-2} over the choice of the codebook as follows:
\begin{align*}
     \mathbb{E}_{\mathcal{C}_1}\left( \text{Eq.} \eqref{n-2} \right)&=
   \mathbb{E}_{\mathcal{C}_1}\left(  \frac{4}{2^{R_0+S_{1}+S_{2}}}
\sum_{m_0,s_1,s_2} 
\sum_{s_2'\neq s_2}
 \tr  T_{2\,x(m_0,s_1,s_2')}
\rho_{v(m_0,s_1)}^{B_1}\right)\\
 &= \frac{4}{2^{R_0+S_{1}+S_{2}}}
\sum_{m_0,s_1,s_2} 
\sum_{s_2'\neq s_2} \mathbb{E}_{\mathcal{C}_1}\left(
 \tr  T_{2\,x(m_0,s_1,s_2')}
\rho_{v(m_0,s_1)}^{B_1}\right)\\
&= 4\times(2^{S_2}-1)\tr T_2^{UVXB_1} \rho^{UX-V-B_1}\\
&\le 4\times 2^{S_2}\tr T_2^{UVXB_1} \rho^{UX-V-B_1}.
\end{align*}
We find similar expressions for Eqs. \eqref{n-1} and \eqref{n-0} as follows:
\begin{align*}
    \mathbb{E}_{\mathcal{C}_1}\left( \text{Eq.} \eqref{n-1} \right)&\leq 4\times 2^{S_{1}+S_{2}}\tr T_1^{UVXB_1} \rho^{VX-U-B_1},\\
       \mathbb{E}_{\mathcal{C}_1}\left( \text{Eq.} \eqref{n-0} \right)&\le 4\times 2^{R_0+S_{1}+S_{2}}\tr T_0^{UVXB_1} (\rho^{UVX}\otimes \rho^{B_1}).
\end{align*}
We further multiply Eqs. \eqref{h-0}, \eqref{h-1} and \eqref{h-2} by $2$ and bring everything together as follows:
\begin{align*}
    \mathbb{E}_{\mathcal{C}_1}\left( P_{e,1}(\mathcal{C}_1)\right)&=\mathbb{E}_{\mathcal{C}_1}\left(\text{Pr}\{\hat{M}_{0b_1}\neq M_0,\hat{M}_{1b_1}\neq M_1|\mathcal{C}_1\}\right)\\
    &\leq 4\tr (I-T_0^{UVXB_1}) \rho^{UVXB_1} + 4\times 2^{R_0+S_{1}+S_{2}}\tr T_0^{UVXB_1} (\rho^{UVX}\otimes \rho^{B_1})\\
      & \hspace{1cm}+ 4\tr (I-T_1^{UVXB_1}) \rho^{UVXB_1} +4\times 2^{S_{1}+S_{2}}\tr T_1^{UVXB_1} \rho^{VX-U-B_1} \\
 &\hspace{1cm}+4\tr (I-T_2^{UVXB_1}) \rho^{UVXB_1}
+ 4\times 2^{S_2}\tr T_2^{UVXB_1} \rho^{UX-V-B_1} \\
&\stackrel{\text{(a)}}{\le} 
2^{\alpha (R_0+ S_{1}+ S_{2})+2}  
2^{-\alpha D_{1-\alpha} ({\cal E}^{B_1}(\rho^{UVX B_1})\| \rho^{UVX} \otimes \rho^{ B_1}) }\\
&\hspace{1cm}+2^{\alpha (S_{1}+ S_{2})+2} 
2^{-\alpha D_{1-\alpha} ({\cal E}_1^{UB_1}(\rho^{UVXB_1})\| {\cal E}^{B_1}(\rho^{VX-U- B_1})) }
\\&\hspace{1cm}+ 2^{\alpha S_{2}+2}  
2^{-\alpha D_{1-\alpha} ({\cal E}_2^{UVB_1}(\rho^{UVXB_1})\|
{\cal E}_1^{UB_1}( \rho^{UX-V- B_1})) }  \\
&\stackrel{\text{(b)}}{\le} 
\nu 2^{\alpha (R_0+ S_{1}+ S_{2})+2}  
2^{-\alpha \widetilde{D}_{1-\alpha} (\rho^{UVX B_1}\| \rho^{UVX} \otimes \rho^{ B_1}) }\\
&\hspace{1cm}+\nu_1 2^{\alpha (S_{1}+ S_{2})+2} 
2^{-\alpha \widetilde{D}_{1-\alpha} (\rho^{UVXB_1}\| {\cal E}^{B_1}(\rho^{VX-U- B_1})) }
\\&\hspace{1cm}+ \nu_2 2^{\alpha S_{2}+2}  
2^{-\alpha \widetilde{D}_{1-\alpha} (\rho^{UVXB_1}\|
{\cal E}_1^{UB_1}( \rho^{UX-V- B_1})) }  \\
&\stackrel{\text{(c)}}{\le} 
\nu 2^{\alpha (R_0+ S_{1}+ S_{2})+2}   
2^{-\alpha \tilde{I}_{1-\alpha}^\uparrow(UVX ;B_1)_{\rho^{UVX B_1} }}\\
&\hspace{1cm}+ \nu_1 2^{\alpha (S_{1}+ S_{2})+2}   
2^{-\alpha \tilde{I}_{1-\alpha}^\downarrow (VX;B_1|U)_{\rho^{UVXB_1}| \rho^{UVX}}}\\
&\hspace{1cm}+ \nu_2 2^{\alpha S_{2}+2}  
2^{-\alpha \tilde{I}_{1-\alpha}^\downarrow (UX;Y|V)_{\rho^{UVXB_1}| \rho^{UVX}}}.
\end{align*}
where 
$\nu$ is the number of distinct eigenvalues of $\rho^{B_1}$,
$\nu_1$ is the maximum number of distinct eigenvalues of the operators $\{\mathcal{E}^{B_1}(\rho^{B_1}_{u})\}_u$,
$\nu_2$ is the maximum number of eigenvalues of $\{\mathcal{E}_{1|u}^{B_1}(\rho^{B_1}_{v})\}_{u,v}$. Here, (a) follows from Lemma \ref{hp-testing}, (b) from Lemma \ref{petz-sandwich} and (c) from the definitions of the R\'enyi quantum mutual information quantities.

\bigskip

\textbf{\emph{Receiver $B_2$:}} 
Analuzing the second receiver is relatively simple because it only decodes the common message $M_0$ and it is degradable with respect to receiver $B_1$. Therefore, we only provide a high-level discussion for this case.
Let the pinching on $B_2$ for the spectral decomposition of $\rho^{B_2}$ be denoted by $\mathcal{E}^{B_2}$. 
We define the following projector:
\begin{align*}
    O^{UB_2}\coloneqq\{\mathcal{E}^{B_2}(\rho^{UB_2})\geq 2^{R_0}\rho^{U}\otimes\rho^{B_2}\}.
\end{align*}
We further define the following projector:
\begin{align*}
    O_{u}^{B_2}=\bra{u}O\ket{u}.
\end{align*}
The decoder utilized by receiver $B_2$ encompasses a square-root measurement implemented with the following POVM:
\begin{align*}
    \Delta_{m_0}^{B_2}(\mathcal{C}_1)\coloneqq\left(\sum_{m_0'}O_{u(m_o')}^{B_2}\right)^{-\frac{1}{2}}O_{u(m_0)}^{B_2}\left(\sum_{m_0'}O_{u(m_o')}^{B_2}\right)^{-\frac{1}{2}}.
\end{align*}
Note that, unlike $O^{UB_2}$ or $O_u^{B_2}$ which do not deponent on any specific code, we constructed the above-mentioned POVM based on the these operators for our code $\mathcal{C}_1$.

Let $\hat{M}_{0b_2}$ denote receiver $B_2$'s estimate of the common message. Utilizing this POVM, the average error probability of the code $\mathcal{C}_1$ relevant to receiver $B_2$ is calculated as follows:
\begin{align*}
    P_{e,2}(\mathcal{C}_1)&=\text{Pr}\{\hat{M}_{0b_2}\neq M_0|\mathcal{C}_1\}\\
    &=\frac{1}{2^{R_0}}\sum_{m_0}\tr(I-\Delta_{m_0}^{B_2})\rho^{B_2}_{u(m_0)}\\
    &\le \frac{1}{2^{R_0}}\sum_{m_0}\tr\big(2(I-O_{u(m_0)}^{B_2})+4\sum_{m_0'\neq m_0}O_{u(m_0')}^{B_2}\big)\rho^{B_2}_{u(m_0)},
\end{align*}
where the inequality corresponds to the Hayashi-Nagaoka inequality Lemma \ref{hayashi-nagaoka}.
Evaluating the average over the random choice of the code, we find:
\begin{align*}
\E_{\mathcal{C}_1}(P_{e,2}(\mathcal{C}_1))&=\E_{\mathcal{C}_1}\left(\text{Pr}\{\hat{M}_{0b_2}\neq M_0|\mathcal{C}_1\}\right)\\
&\le
\E_{\mathcal{C}_1}\left(\frac{1}{2^{R_0}}\sum_{m_0}\tr\big(2(I-O_{u(m_0)}^{B_2})+4\sum_{m_0'\neq m_0}O_{u(m_0')}^{B_2}\big)\rho^{B_2}_{u(m_0)}\right)\\
   &=2\tr (I-O^{UB_2})\rho^{UB_2} + 4(2^{R_0}-1)\tr O^{UB_2}(\rho^{U}\otimes\rho^{B_2})\\
   &\le 4\tr (I-O^{UB_2})\rho^{UB_2} + 4\times 2^{R_0}\tr O^{UB_2}(\rho^{U}\otimes\rho^{B_2})\\
   &\stackrel{\text{(a)}}{\le} 2^{\alpha R_0+2}2^{-\alpha D_{1-\alpha}(\mathcal{E}^{B_2}(\rho^{UB_2})\|\rho^{U}\otimes \rho^{B_2})}\\ 
    &\stackrel{\text{(b)}}{\le} \mu 2^{\alpha R_0+2}2^{-\alpha \widetilde{D}_{1-\alpha}(\rho^{UB_2}\|\rho^{U}\otimes \rho^{B_2})}\\
    &\stackrel{\text{(c)}}{=} 2^{\alpha R_0+2}2^{-\alpha \tilde{I}_{1-\alpha}^{\uparrow}(U;B_2)},
\end{align*}
where (a) follows from Lemma \ref{hp-testing}, (b) from Lemma \ref{petz-sandwich} and $\mu$ is the number of distinct eigenvalues of $\rho^{B_2}$, and (c) from the definition of the R\'enyi quantum mutual information.

\bigskip

\textbf{\emph{Receiver $B_3$:}}
Next, we discuss about the third receiver. This receiver uses non-unique decoding to retrieve the common message. Specifically, it also attempts to detect message $M_{11}$, but if it fails to identify this message, it is not counted as a fault. Let $\mathcal{E}^{B_3}$ be the pinching map on $B_3$ with respect to the spectral decompsotion of the operator $\rho^{B_3}$.
We define the following projector:
\begin{align}
\Upsilon^{V B_3}\coloneqq \{ \mathcal{E}_{3}(\rho^{V B_3}) \ge 2^{R_0+S_{1}}
\rho^{V}\otimes \rho^{B_3}\} .
\end{align}
We also define
\begin{align}
\Upsilon^{B_3}_{v}\coloneqq\bra{v} \Upsilon^{V B_3}\ket{v}.
\end{align}
By employing the latter operator, which is not based on any particular code, we constrcut the following pretty-good POVM as the decoder for $B_3$, tailored specifically to the one-shot code $\mathcal{C}_1$:
\begin{align}
\Delta_{m_0}^{B_3}(\mathcal{C}_1)\coloneqq
\left(\sum_{m_0'} \Upsilon^{B_3}_{v(m_0')}\right)^{-\frac{1}{2}}
\Upsilon^{B_3}_{v(m_0)}
\left(\sum_{m_0'} \Upsilon^{B_3}_{v(m_0')}\right)^{-\frac{1}{2}},
\end{align} 
where 
\begin{align} 
\label{non-unique-povm-2}\Upsilon^{B_3}_{v(m_0)}=\sum_{s_1=1}^{2^{S_1}}\Upsilon^{B_3}_{v(m_0,s_1)}.
\end{align}
The non-unique decoding is implemented by ``collecting POVMs'' in the above manner. It is important to note that in order to identify a codeword $v\in\mathcal{V}$, we need two indices, $m_0$ and $s_1$. However, by summing up the operators over the variable $s_1$, the decoder does not uniquely determine $s_1$; instead, it finds $m_0$ up to $s_1$. The subtle difference between the POVM in Eq. \eqref{non-unique-povm-2} and the similar POVM we introduced in the context of Marton code, i.e. Eq. \eqref{non-unique-povm-1}, lies in the fact that in the latter, the non-unique variable $k_1$ is generated thorugh binning, while in the former, it is created by superposition of codewords. 

Let $\hat{M}_{0b_3}$ denote receiver $B_3$'s estimate of the common message. Utilizing the above POVM, the average error probability of the code $\mathcal{C}_1$ relevant to receiver $B_3$ is calculated as follows:
 
\begin{align*}
  P_{e,3}(\mathcal{C}_1)&=\text{Pr}\{\hat{M}_{0b_3}\neq M_0|\mathcal{C}_1\}\\
 &=\frac{1}{2^{R_0}}\sum_{m_0=1}^{2^{R_0}} \tr (I-\Delta_{m_0}^{B_3}(\mathcal{C}_1))
\rho^{B_3}_{v(m_0,s_1)}\\
&\stackrel{\text{(a)}}{\leq}  
\frac{1}{2^{R_0}}
 \sum_{m_0} \tr \Big(
 2(I-\Upsilon^{B_3}_{v(m_0)})+
 4\sum_{ m_0'\neq m_0} \Upsilon^{B_3}_{v(m_0')}
\Big) \rho^{B_3}_{v(m_0,s_1)} \\
 &\stackrel{\text{(b)}}{\leq}  
\frac{1}{2^{R_0}}
 \sum_{m_0} \tr \Big(
 2(I- \sum_{s_1'}\Upsilon^{B_3}_{v(m_0,s_1')})+
 4\sum_{ m_0'\neq m_0, s_1'} \Upsilon^{B_3}_{v(m_0',s_1')}
\Big) \rho^{B_3}_{v(m_0,s_1)} \\
&=\frac{1}{2^{R_0}}
 \sum_{m_0} 2\tr
 (I- \Upsilon^{B_3}_{v(m_0,s_1)} -\sum_{s_1'\neq s_1}\Upsilon^{B_3}_{v(m_0,s_1')})\rho^{B_3}_{v(m_0,s_1)}\\
 &\hspace{3cm}+\frac{1}{2^{R_0}}
 \sum_{m_0}4\sum_{ m_0'\neq m_0, s_1'}\tr \Upsilon^{B_3}_{v(m_0',s_1')}\rho^{B_3}_{v(m_0,s_1)}\\
 &\stackrel{\text{(c)}}{\leq}\frac{1}{2^{R_0}}
 \sum_{m_0} 2\tr
 (I- \Upsilon^{B_3}_{v(m_0,s_1)} )\rho^{B_3}_{v(m_0,s_1)}\\
 &\hspace{3cm}+\frac{1}{2^{R_0}}
 \sum_{m_0}4\sum_{ m_0'\neq m_0, s_1'}\tr \Upsilon^{B_3}_{v(m_0',s_1')}\rho^{B_3}_{v(m_0,s_1)},
\end{align*}
where (a) folows from Hayashi-Nagaoka inequality Lemma \ref{hayashi-nagaoka}, (b) uses the definition of the non-unique POVM from Eq. \eqref{non-unique-povm-2}, and (c) comes about by throwing out the $(S_1-1)$ negative terms from the non-unique POVM.

We next find $\mathbb{E}_{\mathcal{C}_1}\left(P_{e,3}(\mathcal{C}_1)\right)$, the average of the average error probability over the choice of the random code (because the average error probability $P_{e,3}(\mathcal{C}_1)$, as a function of the random code $\mathcal{C}_1$, is a random variable itself). We obtain:
\begin{align*}
    \mathbb{E}_{\mathcal{C}_1}\left(P_{e,3}(\mathcal{C}_1)\right)
    &\leq\mathbb{E}_{\mathcal{C}_1}\Bigg(\frac{1}{2^{R_0}}
 \sum_{m_0} 2\tr
 (I- \Upsilon^{B_3}_{v(m_0,s_1)} )\rho^{B_3}_{v(m_0,s_1)}\\
 &\hspace{3cm}+\frac{1}{2^{R_0}}
 \sum_{m_0}4\sum_{ m_0'\neq m_0, s_1'}\tr \Upsilon^{B_3}_{v(m_0',s_1')}\rho^{B_3}_{v(m_0,s_1)}\Bigg)\\
 &=\frac{1}{2^{R_0}}
 \sum_{m_0} 2\mathbb{E}_{\mathcal{C}_1}\left(\tr
 (I- \Upsilon^{B_3}_{v(m_0,s_1)} )\rho^{B_3}_{v(m_0,s_1)}\right)\\
 &\hspace{3cm}+\frac{1}{2^{R_0}}
 \sum_{m_0}4\sum_{ m_0'\neq m_0, s_1'}\mathbb{E}_{\mathcal{C}_1}\left(\tr \Upsilon^{B_3}_{v(m_0',s_1')}\rho^{B_3}\right)\\
&=2\tr
 (I- \Upsilon^{VB_3})\rho^{VB_3} +
 4 \times (2^{R_0+S_1}-1) \tr \Upsilon^{VB_3} (\rho^{V}\otimes \rho^{B_3})\\
 &\le 4\tr
 (I- \Upsilon^{VB_3})\rho^{VB_3} +
 4 \times 2^{R_0+S_1} \tr \Upsilon^{VB_3} (\rho^{V}\otimes \rho^{B_3})\\
 &\stackrel{\text{(a)}}{\le}2^{\alpha(R_0+S_{1})+2} 
2^{-\alpha D_{1-\alpha}\left(\mathcal{E}^{B_3}(\rho^{VB_3})\|\rho^{V}\otimes \rho^{B_3}\right)}\\
&\stackrel{\text{(b)}}{\le}\eta 2^{\alpha(R_0+S_{1})+2} 
2^{-\alpha \widetilde{D}_{1-\alpha}\left(\rho^{VB_3}\|\rho^{V}\otimes \rho^{B_3}\right)}\\
&\stackrel{\text{(c)}}{=}2^{\alpha(R_0+S_{1})+2} 
2^{-\alpha \tilde{I}_{1-\alpha}^{\uparrow}(V;B_3)_{\rho^{VB_3}}},
\end{align*}
where (a) follows from Lemma \ref{hp-testing}, (b) from Lemma \ref{petz-sandwich} and $eta$ is the number of distinct eigenvalues of $\rho^{B_3}$, and (c) from the definition of the R\'enyi quantum mutual information.

\bigskip

\subsection{Asymptotic analysis}\label{asymptotic-2}
As mentioned before, our proof consists in creating a one-shot code and then determining the rate region through asymptotic analysis.
 The key observation is that the number of distinct eigenvalues will vanish because they are only polynomial and that R\'enyi relative entropy tends to von Nuemann relative entropy when its parameter tends to $1$.

\begin{proofof}[Theorem \ref{two-degraded-multilevel}]


The upper bounds on the expectation of the average error probability imply the existence of at least one good code $\mathcal{C}_1$ that satisfies these bounds. Consequently, we can assess the upper bounds in the following manner:
\begin{align*}
    P_{e,1}^{(n)}(\mathcal{C}_n)&\leq 
    4 (n+1)^{\alpha (d_{B_1}-1)} 2^{\alpha n(R_0 + S_{1} + S_{2})} 2^{-\alpha\tilde{I}_{1-\alpha}^{\uparrow}(U^n,V^n,X^n;B_1^n)_{(\rho^{UVXB_1})^{\otimes n}}}\\
    &+4  (n+1)^{\alpha d_{U}(d_{B_1}+2)(d_{B_1}-1)/2} 2^{\alpha n(S_{1}+S_2)} 2^{-\alpha \tilde{I}_{1-\alpha}^{\downarrow}\left(V^n,X^n;B_1^n|U^n\right)_{(\rho^{UVXB_1})^{\otimes n}\vert(\rho^{UVX})^{\otimes n}}}\\
    &+4  (n+1)^{\alpha d_{V}d_U(d_{B_1}+2)(d_{B_1}-1)/2} 2^{\alpha nS_2} 2^{-\alpha \tilde{I}_{1-\alpha}^{\downarrow}\left(U^n,X^n;B_1^n|V^n\right)_{(\rho^{UVXB_1})^{\otimes n}\vert(\rho^{UVX})^{\otimes n}}}\\
  P_{e,2}^{(n)}(\mathcal{C}_n)&\leq 
4 (n+1)^{\alpha (d_{B_2}-1)} 2^{\alpha nR_0} 2^{-\alpha\tilde{I}_{1-\alpha}^{\uparrow}(U^n;B_2^n)_{(\rho^{UB_2})^{\otimes n}}},\\
  P_{e,3}^{(n)}(\mathcal{C}_n)&\leq 
4 (n+1)^{\alpha (d_{B_3}-1)} 2^{\alpha n(R_0+S_1)} 2^{-\alpha\tilde{I}_{1-\alpha}^{\uparrow}(V^n;B_3^n)_{(\rho^{VB_3})^{\otimes n}}},
\end{align*}
where we have employed the polynomial bounds on the number of distinct eigenvalues from Proposition \ref{pinching-asymptotic}. We further obtain:
\begin{align*}
  \lim_{n\to\infty }-\frac{1}{n}\log P_{e,1}^{(n)}(\mathcal{C}_n)\ge
  \min\bigg(&\alpha \big(\tilde{I}_{1-\alpha}^{\uparrow}(U,V,X;B_1)_{\rho^{UVXB_1}}-(R_0 + S_{1} + S_{2})\big),\\
  &\alpha \big(\tilde{I}_{1-\alpha}^{\downarrow}\left(V,X;B_1|U\right)_{\rho^{UVXB_1}|\rho^{UVX}}-(S_{1}+S_2)\big),\\
  &\alpha \big(\tilde{I}_{1-\alpha}^{\downarrow}\left(U,X;B_1|V\right)_{\rho^{UVXB_1}|\rho^{UVX}}-S_2\big)
  \bigg),
\end{align*}
\begin{align*}
\lim_{n\to\infty }-\frac{1}{n}\log P_{e,2}^{(n)}(\mathcal{C}_n)&\ge
  \alpha \big(\tilde{I}_{1-\alpha}^{\downarrow}(U;B_2)_{\rho^{UB_2}}-R_0\big),\\
\lim_{n\to\infty }-\frac{1}{n}\log P_{e,3}^{(n)}(\mathcal{C}_n)&\ge
  \alpha \big(\tilde{I}_{1-\alpha}^{\downarrow}(V;B_3)_{\rho^{VB_3}}-(R_0+S_1)\big).
\end{align*}
It now follows that as $n\to\infty$ and $\alpha\to 0$, there exists a sequence of codes $\mathcal{C}_n$ such that $P_{e,1}^{(n)}(\mathcal{C}_n),P_{e,2}^{(n)}(\mathcal{C}_n)$ and $P_{e,3}^{(n)}(\mathcal{C}_n)$ vanish, if
\begin{align*}
R_0 + S_{1} + S_{2} &\leq I(X;B_1)_{\rho},\\
    S_{1}+S_2&\leq I(X;B_1|U)_{\rho},\\
     S_{2}&\leq I(X;B_1|V)_{\rho},\\
     R_0&\leq I(U;B_2)_{\rho}, \\
    R_0 + S_{1} &\leq I(V;B_3)_{\rho},
\end{align*}
where the mutual information quantities are calculated for the state in the statement of the Theorem \ref{two-degraded-multilevel}. Substituting $R_1=S_1+S_2$ and using the Fourier–Motzkin elimination procedure, we remove $S_{1},S_{2}$, it follows from the union bound $P_{e}^{(n)}(\mathcal{C}_n)\le P_{e,1}^{(n)}(\mathcal{C}_n)+P_{e,2}^{(n)}(\mathcal{C}_n)+P_{e,3}^{(n)}(\mathcal{C}_n)$ that the probability of decoding error tends to zero as $n\to \infty$ if the inequalities in Theorem are satisfied. This completes the proof of Theorem \ref{two-degraded-multilevel}. 
\end{proofof}
\bigskip

It is also possible to find a simple achievability region for this problem. For this problem, we now provide an inner bound which can be considered a straightforward extension of superposition scheme \cite{savov-wilde,q-yard} to three-receiver multilevel channel, generalizing the seminal work of K\"orner and Marton \cite{korner-marton}.

\begin{theorem}\label{straight-superposition}
    Consider a three-receiver multilevel broadcast channel with density matrices $\{\rho^{B_1B_2B_3}_x\}_x$. Further assume that there exists a quantum channel $\mathcal{M}$ such that $\mathcal{M}(\rho_x^{B_1})=\rho_x^{B_2}$ for all $x$. Then,
    the set of rate pairs $(R_0,R_1)$ is achievable if
    \begin{align*}
            R_0&\leq\min\{I(U;B_2),I(U,B_3)\},\\
            R_1&\leq I(X;B_1|U), 
    \end{align*}
     for some probability distribution $p(u,x)$ giving rise to the cq-state
    \begin{align*}
        \rho^{UXB_1B_2B_3}=\sum_{u,x}p(u,x)\ketbra{u}\otimes\ketbra{x}\otimes\rho^{B_1B_2B_3}_x.
    \end{align*}
\end{theorem}
\begin{proof}
We can easily see that setting $U=V$ in the above theorem gives the straightforward superposition region in Theorem \ref{straight-superposition}
\end{proof}

\bigskip

\subsection{No-go region}
The following theorem establishes a weak converse for the setting of Theroem \ref{two-degraded-multilevel}. We will argue that it is tight up to the common problem of conditioning on quantum systems, meaning that its acheivability is known only in the essentially classical case.
\begin{theorem}\label{converse-1}
Assuming a transmitter and three receivers have access to many independent uses of a three-receiver multilevel quantum broadcast channel $x\to\rho_x^{B_1B_2B_3}$ with two-degraded message set, where a quantum channel $\mathcal{M}$ exists such that $\mathcal{M}(\rho_x^{B_1})=\rho_x^{B_2}$ for all $x$, if the set of rate pairs $(R_0,R_1)$ are asymptotically achievable, then they are contained in the following region:
        \begin{align*}
            R_0&\leq\min\{I(U;B_2),I(V,B_3)\},\\
            R_1&\leq I(X;B_1|U),\\
            R_0+R_1&\leq I(V;B_3) + I(X;B_1|V),
        \end{align*}
    for some state 
     \begin{align*}
        \rho^{UXVB_1B_2B_3}=\sum_{x}p(x)\ketbra{x}\otimes\rho_x^{UV}\otimes\rho^{B_1B_2B_3}_x,
    \end{align*}
    such that there for all $x\in\mathcal{X}$, there exist cptp maps $\{\mathcal{M}_x\}$ satisfying $\mathcal{M}_x(\rho_x^{U})=\rho_x^{V}$. Note that systems $U$ and $V$ are generically quantum systems. 
\end{theorem}
\begin{proof}
The proof follows from combining two no-go results for two-receiver degraded quantum broadcast channel and two-receiver quantum broadcast channel with degraded message set \cite[Sec. II B-C]{q-yard}. Here, we only provide a high-level outline of the process.
Consider two quantum broadcast channels:
One is $x\to\rho_x^{B_1B_2}$, which is a degraded channel by definition, and the second $x\to\rho_x^{B_1B_3}$, which is a general two-receiver quantum broadcast channel with two-degraded message set. The weak converse for both of these channels are established as follows \cite[Sec. II B-C]{q-yard}: For the first channel, we can identify an auxiliary random variable $U_i=M_0B_1^{i-1}$ and show that
    \begin{align*}
        R_0&\leq I(U;B_2),\\
        R_1&\leq I(X;B_1|U).
    \end{align*}
For the second channel, we can choose $V_i=(M_0,B_1^{i-1},B_{3,i+1}^{n})$ and show that 
    \begin{align*}
        R_0&\leq\min\{I(V;B_3),I(V;B_1)\},\\
        R_0+R_1&\leq I(V;B_3)+I(X;B_1|V).
    \end{align*}
We observe that, with abuse of notation, the required Markov chain holds $U-V-X$. Noting that $I(V;B_1)\geq I(V;B_2)\geq I(U;B_2)$, the combination of these regions gives us the required no-go result.
\end{proof}
\begin{remark}
Note that in the aforementioned Markov chain $U-V-X$, the systems $U$ and $V$ is generally a quantum systems, and they form a Markov chain in the sense that $V$ makes $X$ conditionally independent of $U$. Moreover, notice that we only could find single-letter converses by choosing quantum auxiliary variables $(U,V)$. This means that we don't know if our achievability region could match this outer bound except when the quantum systems mutually commute, i.e. they are essentially classical. It is also not known whether the no-go region involving quantum systems is bigger than the achievability region. These are important questions in quantum information theory \cite{6634255,8370123}.
\end{remark}

\bigskip

\section{General three-receiver quantum broadcast channel with two-degraded message set}
\label{general-two-degraded}
In this section, we continue with transmission of a two-degraded message set over the three-receiver quantum broadcast channel $x\to\rho_x^{B_1B_2B_3}$, where the transmitter wants to send a common message $m_0\in[1:2^{nR_0}]$ to all receivers, and a private (or individualized) message $m_1\in[1:2^{nR_1}]$ solely to receiver $B_1$. However, unlike the previous section, the three-receiver quantum broadcast channel is not assumed to be multilevel. This corresponds to the scenario depicted in Fig. \ref{multi-level.two-degraded} except that there exist no degrading channel $\mathcal{M}^{B_1\to B_2}$.
Definitions of code, achievability, and capacity region are straightforward from Definition \ref{code-def-multilevel}.

\begin{theorem}\label{general-two-degraded.t}
     We assume a transmitter and three receivers have access to many independent uses of a three-receiver quantum broadcast channel $x\to\rho_x^{B_1B_2B_3}$. The set of rate pairs $(R_0,R_1)$ is asymptotically achievable for the aforementioned two-degraded message set if
    \begin{align*}
        R_0 & \leq \min\{I(V_2;B_2), I(V_3;B_3)\},\\
        2R_0 & \leq I(V_2;B_2) + I(V_3; B_3) -I(V_2;V_3|U),\\
        R_0 + R_1 &\leq \min\{I(X;B_1),I(V_2;B_2) + I(X;B_1|V_2), I(V_3,B_3)+I(X;B_1|V_3)\},\\
        2R_0 + R_1 & \leq I(V_2;B_2) I(V_3;B_3)+ I(X;B_1|V_2,V_3)-I(V_2,V_3|U),\\
        2R_0 + 2R_1 &\leq I(V_2;B_2) + I(X;B_1|V_2)+I(V_3;B_3)+I(X;B_1|V_3)-I(V_2;V_3|U),\\
        2R_0 + 2R_1 & \leq I(V_2;B_2) + I(V_3;B_3) + I(X;B_1|U) + I(X;B_1|V_2,V_3) - I(V_2;V_3|U),
\end{align*}
hold for some probability distribution
\begin{align*}
    p(u,v_2,v_3,x) =& p(u)(v_2|u)p(x,v_3|v_2)\\
    =& p(u)(v_3|u)p(x,v_2|v_3),
\end{align*}
i.e. the Markov chains $U-V_2 - (V_3,X)$ and $U-V_3 - (V_2,X)$ are simultanoeusly satisfied,
 resulting in the cq-state $\rho^{UXV_2V_3B_1B_2B_3}$.
\end{theorem}

\bigskip

\subsection{One-shot code construction}
Consider a three-receiver  cq-broadcast channel $x\to\rho_x^{B_1B_2B_3}$, with corresponding marginal channels 
$x\to \rho^{B_1}_x$, $x\to \rho^{B_2}_x$, and $x\to \rho^{B_3}_x$. For codebook construction, we fix a probability distribution on four random variables satisfying the Markov chains $U-V_2-(X,V_3)$ and $U-V_3-(X,V_2)$.
We use pinching technique to construct POVMs and use Sen's mutual covering lemma for Marton coding \cite[Fact 2]{Sen2021}. 
 The asymtotic analysis of the code developped here recovers Theorem \ref{general-two-degraded.t}
\subsubsection{Rate splitting}
We split the message $M_1$ (which is intended solely to the first receiver) into four independent messages $\{M_{1i}\}_{i=0,1,2,3}$, of respective rates $\{S_{i}\}_{i=0,1,2,3}$,
i.e. $R_1=S_0 + S_1 + S_2 + S_3$. For notational convenience, we denote the relazitations of the uniform random variables $M_{10},M_{11},M_{12}$ and $M_{13}$ by $s_0,s_1,s_2$ and $s_3$, respectively. The message pair $(M_0,M_{10})$ is represented by $U$. Using superposition and Marton coding, the message triple $(M_0,M_{10},M_{12})$ is represented by $V_2$ and the message triple $(M_0,M_{10},M_{13})$ is represented by
$V_3$. Finally, using another superposition coding, the message pair $(M_0,M_{1})$ is represented by $X$. Receiver $B_1$ find the pair $(M_0,M_{1})$ by decoding $U,V_2,V_3,X$, and receivers $B_2$ and $B_3$ find $M_0$ by decoding $U$ indirectly through $V_2$ and $V_3$, respectively.  
\subsubsection{Codebook generation} Fix a probability distribution $p(u,v_2,v_3,x)$ such that the Markov chains $U-V_2-(X,V_3)$ and $U-V_3-(X,V_2)$ hold simultaneously, i.e. $p(u,v_2,v_3,x) = p(u)(v_2|u)p(x,v_3|v_2)
    = p(u)(v_3|u)p(x,v_2|v_3)$. Furthermore, let $T_2\geq S_2$ and $T_3\geq S_3$.
Our random code $\mathcal{C}_1$ is composed of four tuples, each containing subsets of elements from $\mathcal{U},\mathcal{V}_2,\mathcal{V}_3$, and $\mathcal{X}$. These tuples are generated as follows:
Randomly and independently generate $2^{R_0+S_0}$ elements 
$\{u(m_0,s_0)\}$, $m_0 \in [1:2^{R_0}],s_0 \in [1:2^{S_0}]$ according to the pmf $p(u)$.
For each $u(m_0,s_0)$, generate (a) $2^{S_2+r_1}$ elements
$\{v_2(m_0,s_0,s_2,k_1)\}_{(s_2,k_1)\in [1:2^{S_2}]\times[1:2^{r_1}]}$, randomly and conditionally independent from the pmf $p(v_2|u)$, (b)
$2^{S_3+r_2}$ elements
$\{v_3(m_0,s_0,s_3,k_2)\}_{(s_3,k_2)\in [1:2^{S_3}]\times[1:2^{r_2}]}$, randomly and conditionally independent according to $p(v_3|u)$. The elements 
\begin{align*}
    \{v_2(m_0,s_0,s_2,k_1)\}_{(s_2,k_1)\in [1:2^{S_2}]\times[1:2^{r_1}]}
\end{align*}
are randomly partitioned into $2^{S_2}$ bins of equal size, so that each bin contains the same number of elements $2^{r_1}$. Similarly, the $2^{S_3+r_2}$ elements $\{v_3(m_0,s_0,s_3,k_2)\}_{(s_3,k_2)\in [1:2^{S_3}]\times[1:2^{r_2}]}$ are partitioned into $2^{S_3}$ equal size bins, so that each bin contains $2^{r_2}$ elements.
To ensure that each product bin $(s_2,s_3)$ contains a suitable pair $(v_2(m_0,s_0,s_2,k_1),v_2(m_0,s_0,s_3,k_2))$ for encoding with high probability, we require that
\begin{align}\label{covering-condition}
    r_1 + r_2\geq I_{\text{max}}^{\varepsilon}(V_2;V_3|U)+2\log\frac{1}{\varepsilon}.
\end{align}
This follows from mutual covering lemma \cite[Fact 2]{Sen2021}.
Roughly speaking, this requirement means that there exists at least one pair $(v_2(m_0,s_0,s_2,k_1),v_2(m_0,s_0,s_3,k_2))$ in each product bin $(s_2,s_3)$ whose joint distribution is close to $p(v_2,v_3|u)$ although they are generated from $p(v_2|u)p(v_3|u)$.
Finally, for each pair $(v_2(m_0,s_0,s_2,k_1),v_3(m_0,s_0,s_3,k_2))$ in each product bin $(s_2,s_3)$ that satisfies the condition of mutual covering lemma, randomly and conditionally independently generate $2^{S_1}$ elements $x(m_0,s_0,s_2,s_3,s_1)$, $s_1 \in [1:2^{S_1}]$, each according to $P(x|v_2,v_3)$. This completes the codebook.
\subsubsection{Encoding} To send message pair $(m_0,m_1)$, we express $m_1$ by the quadruple $(s_0,s_1,s_2,s_3)$. The sender looks up a suitable pair $(v_2(m_0,s_0,s_2,k_1),v_3(m_0,s_0,s_3,k_2))$ in the product bin $(s_2,s_3)$ and eventually sends the codeword $x(m_0,s_0,s_2,k_1,s_3,k_2)$ over the three-receiver quantum broadcast channel. 

\begin{remark}
    The codebook contains Marton's code and superposition coding. Specifically, to transmit the two parts of message $M_1$, namely $M_{12}$ with rate $S_2$ and $M_{13}$ with rate $S_3$, we use Marton's code. However, since both encoded messages in Marton's code are for the same receiver $B_1$, it differs from Marton's original code, where the messages are for two different receivers. This difference leads to new rates for the messages. Roughly speaking, if the codewords $(v_2(m_0,s_0,s_2,k_1),v_3(m_0,s_0,s_3,k_2))$ were for distinct receivers, we'd need restrictions on the total number of sequences of variables $v_2$ and $v_3$. However, when considering both sequences together, the decoder can eliminate those that don't meet the mutual conditional lemma, reducing its list-size accordingly. This means we only need to impose restrictions on the rates of the messages $S2$ and $S_3$. This is also evident in the classical counterpart of this problem \cite{Nair-Elgamal}, where in Eqs. (20-23), the restrictions are placed on $S_2$ and $S_3$ rather than $T_2$ and $T_3$.
    \end{remark}

\subsubsection{Decoding and the analysis of error probability} 
\label{general-two-degraded-error}
The failure of encoder contributes a constant additive term $f(\varepsilon)$ to the ultimate error probability. We go through the details of the decoding errors below assuming that the encoding was successful.. 
Upon sending the codeword $x(m_0,s_0,s_1,s_2,s_3)$, each receivers obtains its share of $\rho^{B_1B_2B_3}_{x(m_0,s_0,s_2,s_3,s_1)}$.
Subsequently, each receiver applies a POVM on its system to decode the intended messages. These POVMs are specifically designed to extract messages simultaneously. The details are outlined as follows:
\hfill\\
\vspace{-0.4cm}

\textbf{\emph{Receiver $B_1$:}} Let $(\hat{m}_0,\hat{s}_0,\hat{s}_2,\hat{s}_3,\hat{s}_1)$ denote the estimates made by receiver $B_1$ for the messages.
Receiver $B_1$ declares that $(\hat{m}_0,\hat{s}_0,\hat{s}_2,\hat{s}_3,\hat{s}_1)=(m_0,s_0,s_1,s_2,s_3)$ is sent if it uniquely corresponds to the tuple $\left(u(m_0,s_0),v_2(m_0,s_0,s_2,k_1),v_3(m_0,s_0,s_3,k_2),x(m_0,s_0,s_2,s_3,s_1)\right)$ identified by its POVM for some $k_1$ and $k_2$. Notice that $s_2$ and $s_3$ are the product bin indices of $v_2(m_0,s_0,s_2,k_1)$ and $v_3(m_0,s_0,s_3,k_2)$, respectively, which are supposed be decoded correctly, but only up to $k_1$ and $k_2$. That is, the pair $(k_1,k_2)\in[1:2^{r_1}]\times[1:2^{r_2}]$ may not correspond to the actual pair chosen by the encoder, but there is no error as long as the pair $(s_1,s_2)$ is exactly the one chosen by the encoder. Now, we proceed to define the pinching maps as components of the decoding POVM for this receiver. The pinching on $B_1$ with respect to the spectral decomposition of $\rho^{B_1}$
is denoted by $\mathcal{E}^{B_1}$.
We denote the pinching on $U B_1$
for the spectral decomposition of $\mathcal{E}^{B_1}(\rho^{XV_2V_3-U-B_1})$
by $\mathcal{E}^{UB_1}_{1}=\sum_u\ketbra{u}\otimes \mathcal{E}^{B_1}_{1|u}$, where $\mathcal{E}^{B_1}_{1|u}$ is the pinching on $B_1$ for the spectral decomposition of $\mathcal{E}^{B_1}(\rho_{u}^{B_1})$. Note that in fact this pinching consists of a family of pinching maps $\{\mathcal{E}^{B_1}_{1|u}\}_{u}$.
We denote the pinching on $U V_2 B_1$
for the spectral decomposition of $\mathcal{E}^{UB_1}_1(\rho_{UV_3X-V_2-B_1})$
by $\mathcal{E}^{UV_2B_1}_2=\sum_{u,v_2}\ketbra{u}\otimes\ketbra{v_2}\otimes \mathcal{E}^{B_1}_{2|u,v_2}$, where $\mathcal{E}^{B_1}_{2|u,v_2}$ is a pinching maps on $B_1$ with respect to the spectral decomposition of $\mathcal{E}^{B_1}_{1|u}(\rho^{B_1}_{v_2})$.
We denote the pinching on $U V_3 B_1$
for the spectral decomposition of $\mathcal{E}^{UB_1}(\rho_{UV_2X-V_3-B_1})$
by $\mathcal{E}^{UV_3B_1}_{3}=\sum_{u,v_3}\ketbra{u}\otimes\ketbra{v_3}\otimes \mathcal{E}^{B_1}_{3|u,v_3}$, where $\mathcal{E}^{B_1}_{3|u,v_3}$ is a pinching map on $B_1$ for the spectral decomposition of $\mathcal{E}^{B_1}_u(\rho^{B_1}_{v_3})$.
We denote the pinching on $U V_2 V_3 B_1$
for the spectral decomposition of $\mathcal{E}^{UB_1}(\rho_{UX-V_2V_3-B_1})$
by $\mathcal{E}^{UV_2V_3B_1}_4=\sum_{u,v_2,v_3}\ketbra{u}\otimes\ketbra{v_2}\otimes\ketbra{v_3}\otimes \mathcal{E}^{B_1}_{4|u,v_2,v_3}$. This pinching consists of a family of pinching maps $\{\mathcal{E}^{B_1}_{4|u,v_2,v_3}\}_{u,v_2,v_3}$ on $B_1$ for the spectral decomposition of $\mathcal{E}^{B_1}_{1|u}(\rho^{B_1}_{v_2,v_3})$.
We employ these pinching maps to define five projectors as follows:  
\begin{align*}
\Theta_{4}^{UV_2V_3XB_1}&\coloneqq\Big\{ \mathcal{E}^{UV_2V_3B_1}_4(\rho^{UV_2V_3XB_1}) \ge 2^{S_1}\mathcal{E}^{UB_1}(\rho^{UX- V_2V_3-B_1})\Big\},\\
\Theta_{3}^{UV_2V_3XB_1}&\coloneqq\Big\{ \mathcal{E}^{UV_3B_1}_3(\rho^{UV_2V_3XB_1}) \ge 2^{S_2+S_1}\mathcal{E}^{UB_1}(\rho^{UXV_2- V_3-B_1})\Big\},\\
\Theta_{2}^{UV_2V_3XB_1}&\coloneqq\Big\{ \mathcal{E}^{UV_2B_1}_2(\rho^{UV_2B_1}) \ge 2^{S_3+S_1}\mathcal{E}^{UB_1}(\rho^{UXV_3- V_2-B_1})\Big\},\\
\Theta_{1}^{UV_2V_3XB_1}&\coloneqq\Big\{\mathcal{E}^{UB_1}(\rho^{UV_2V_3X B_1}) \ge 2^{S_1+S_2+S_3}\mathcal{E}^{B_1}(\rho^{V_2 V_3 X- U-B_1})\Big\}, \\
\Theta_{0}^{UV_2V_3XB_1}&\coloneqq\Big\{ \mathcal{E}^{B_1}(\rho^{U V_2 V_3 XB_1}) \ge 2^{R_0+S_0+S_1+S_2+S_3}\rho^{UV_2V_3X}\otimes \rho^{B_1}\Big\}.
\end{align*}

We further define the projectors
\begin{align*}
\Theta_{4\,u,v_2,v_3,x}^{B_1}&\coloneqq\bra{u,v_2,v_3,x}\Theta_{4}\ket{u,v_2,v_3,x},\\
\Theta_{3\,u,v_2,v_3,x}^{B_1}&\coloneqq\bra{u,v_2,v_3,x}\Theta_{3}\ket{u,v_2,v_3,x},\\
\Theta_{2\,u,v_2,v_3,x}^{B_1}&\coloneqq\bra{u,v_2,v_3,x}\Theta_{2}\ket{u,v_2,v_3,x},\\
\Theta_{1\,u,v_2,v_3,x}^{B_1}&\coloneqq\bra{u,v_2,v_3,x}\Theta_{1}\ket{u,v_2,v_3,x},\\
\Theta_{0\,u,v_2,v_3,x}^{B_1}&\coloneqq\bra{u,v_2,v_3,x}\Theta_{0}\ket{u,v_2,v_3,x}.
\end{align*}
We can now build the main ingredient of the square-root measurement POVM as:
\begin{align*}
\Theta_{u,v_2,v_3,x}^{B_1}&\coloneqq
\Theta_{0\,u,v_2,v_3,x}^{B_1}
\Theta_{1\,u,v_2,v_3,x}^{B_1}
\Theta_{2\,u,v_2,v_3,x}^{B_1}
\Theta_{3\,u,v_2,v_3,x}^{B_1}
\Theta_{4\,u,v_2,v_3,x}^{B_1}.
\end{align*}
So far the projectors do not depend on any specific code. As a matter of fact, the operator $\Theta^{B_1}$ depends on all four classical variables $u,v_2,v_3$ and $x$ despite our code involves certain conditional independence among these variables. We can however use this operator to build POVMs for any code. For the random code we have created, the POVM of receiver $B_1$ is defined as follows:
\begin{equation}
    \begin{aligned}
    \label{POVM-general-two}
        \Lambda_{m_0,s_0,s_2,s_3,s_1}^{B_1}(\mathcal{C}_1)\coloneqq&
\left(\sum_{m_0',s_0',s_2',s_3',s_1'} \Theta_{x(m_0',s_0',s_2',s_3',s_1')}(\mathcal{C}_1)\right)^{-\frac{1}{2}} \\
&\hspace{3cm}\times \Theta_{x(m_0,s_0,s_2,s_3,s_1)}(\mathcal{C}_1) \times\\
&\hspace{3.5cm}\left(\sum_{m_0',s_0',s_2',s_3',s_1'} \Theta_{x(m_0',s_0',s_2',s_3',s_1')}(\mathcal{C}_1)\right)^{-\frac{1}{2}}.
    \end{aligned}
\end{equation}
The decoder finds all the messages encoded into the variables $U,V_2,V_3$ and $X$ uniquely. Note that the indices $k_1$ and $k_2$ respectively in variables $V_2$ and $V_3$ are discarded in generation of $x$. In other words, Note that $B_1$ is interested in decoding both $(s_{2},s_3)$ correctly, so both bin indices are relevant here. In other words, the decoded codeword $x(m_0,s_{0},s_{2},s_{3},s_1)$ does not depend on indices $k_1$ and $k_2$.
(note that unlike the usual scenario for Marton code where messages encoded in different bin indices are intended for different receivers, here both bin indices are useful messages intended for receiver $B_1$). 

The average error probability of the code $\mathcal{C}_1$ for receiver $B_1$ is bounded as follows. Recall that we use the variable $\hat{M}_{tb_i}$ to denote the estimate made by receiver $B_i$ about that message $t$. We obtain:

\begin{align}
\nonumber
 P_{e,1}(\mathcal{C}_1)&=\text{Pr}\{(\hat{m}_0,\hat{s}_0,\hat{s}_2,\hat{s}_3,\hat{s}_1)\neq(m_0,s_0,s_2,s_3,s_1)\}\\
 \nonumber
 &=\text{Pr}\{\hat{M}_{0b_1}\neq M_0,\hat{M}_{1b_1}\neq M_1|\mathcal{C}_1\}\\
 \nonumber
     &=\text{Pr}\{\hat{M}_{0b_1}\neq M_0,\hat{M}_{10b_1}\neq M_{10},\hat{M}_{12b_1}\neq M_{12},\hat{M}_{13b_1}\neq M_{13},\hat{M}_{11b_1}\neq M_{11}|\mathcal{C}_1\}\\ \nonumber
&\coloneqq
 \frac{1}{2^{R_0+S_0+S_2+S_3+S_1}}\sum_{m_0,s_0,s_2,s_3,s_1} \tr (I-\Lambda_{m_0,s_0,s_2,s_3,s_1}^{B_1}(\mathcal{C}_1))
\rho_{x(m_0,s_0,s_2,s_3,s_1)}^{B_1}\\\nonumber
&\le  \frac{1}{2^{R_0+S_0+S_2+S_3+S_1}}
 \sum_{m_0,s_0,s_2,s_3,s_1} \tr \Big(
 2(I- \Theta_{x(m_0,s_0,s_2,s_3,s_1)}^{B_1}(\mathcal{C}_1))\\ \nonumber
 &\hspace{1.5cm}+4
 \sum_{(m_0',s_0',s_2',s_3',s_1')\neq (m_0,s_0,s_2,s_3,s_1)} \Theta_{x(m_0',s_0',s_2',s_3',s_1')}^{B_1}(\mathcal{C}_1)
\Big) \rho_{x(m_0,s_0,s_2,s_3,s_1)}^{B_1}\\
\label{hn-h1}
&=\frac{2}{2^{R_0+S_0+S_2+S_3+S_1}}
 \sum_{m_0,s_0,s_2,s_3,s_1} \tr 
 (I- \Theta_{x(m_0,s_0,s_2,s_3,s_1)}^{B_1}(\mathcal{C}_1))\rho_{x(m_0,s_0,s_2,s_3,s_1)}^{B_1}\\
 \label{hn-h2}
 &+\frac{4}{2^{R_0+S_0+S_2+S_3+S_1}}\sum_{m_0,s_0,s_2,s_3,s_1}
\sum_{\substack{(m_0',s_0',s_2',s_3',s_1')\neq\\(m_0,s_0,s_2,s_3,s_1)}} \tr \Theta_{x(m_0',s_0',s_2',s_3',s_1')}^{B_1}(\mathcal{C}_1)
 \rho_{x(m_0,s_0,s_2,s_3,s_1)}^{B_1},
\end{align}
where the inequality corresponds to the Hayashi-Nagaoka inequality Lemma \ref{hayashi-nagaoka}.
We analyze each term resulting from Hayashi-Nagaoka inequality Eqs. \eqref{hn-h1} and \eqref{hn-h2}, separately.
We have:

\begin{align}
\nonumber
    \text{Eq.}\,\,\eqref{hn-h1}&=\frac{2}{2^{R_0+S_0+S_2+S_3+S_1}}
 \sum_{m_0,s_0,s_2,s_3,s_1} \tr 
 (I- \Theta_{x(m_0,s_0,s_2,s_3,s_1)}^{B_1}(\mathcal{C}_1))\rho_{x(m_0,s_0,s_2,s_3,s_1)}^{B_1}\\
 \label{hn-hn0}
  &\le\frac{2}{2^{R_0+S_0+S_2+S_3+S_1}}
 \sum_{m_0,s_0,s_2,s_3,s_1} \tr 
 \left(I- \Theta_{0\,x(m_0,s_0,s_2,s_3,s_1)}^{B_1}(\mathcal{C}_1)\right)\rho_{x(m_0,s_0,s_2,s_3,s_1)}^{B_1}\\
  \label{hn-hn1}
  &+\frac{2}{2^{R_0+S_0+S_2+S_3+S_1}}
 \sum_{m_0,s_0,s_2,s_3,s_1} \tr 
 \left(I- \Theta_{1\,x(m_0,s_0,s_2,s_3,s_1)}^{B_1}(\mathcal{C}_1)\right)\rho_{x(m_0,s_0,s_2,s_3,s_1)}^{B_1}\\
  \label{hn-hn2}
   &+\frac{2}{2^{R_0+S_0+S_2+S_3+S_1}}
 \sum_{m_0,s_0,s_2,s_3,s_1} \tr 
 \left(I- \Theta_{2\,x(m_0,s_0,s_2,s_3,s_1)}^{B_1}(\mathcal{C}_1)\right)\rho_{x(m_0,s_0,s_2,s_3,s_1)}^{B_1}\\
  \label{hn-hn3}
   &+\frac{2}{2^{R_0+S_0+S_2+S_3+S_1}}
 \sum_{m_0,s_0,s_2,s_3,s_1} \tr 
 \left(I- \Theta_{3\,x(m_0,s_0,s_2,s_3,s_1)}^{B_1}(\mathcal{C}_1)\right)\rho_{x(m_0,s_0,s_2,s_3,s_1)}^{B_1}\\
  \label{hn-hn4}
   &+\frac{2}{2^{R_0+S_0+S_2+S_3+S_1}}
 \sum_{m_0,s_0,s_2,s_3,s_1} \tr 
 \left(I- \Theta_{4\,x(m_0,s_0,s_2,s_3,s_1)}^{B_1}(\mathcal{C}_1)\right)\rho_{x(m_0,s_0,s_2,s_3,s_1)}^{B_1},
\end{align}
where the inequality follows because $ \Theta_{x(m_0,s_0,s_2,k_1,s_3,k_2,s_1)}^{B_1}(\mathcal{C}_1)\le \Theta_{j\,x(m_0,s_0,s_2,k_1,s_3,k_2,s_1)}^{B_1}(\mathcal{C}_1)\leq I^{B_1}$ for $j=0,1,2,3$. We now turn into the second term arose from Hayashi-Nagaoka inequality Eq. \eqref{hn-h2}. We have
\begin{align}
\nonumber
    \text{Eq.}\,&\,\eqref{hn-h2}=\frac{4}{2^{R_0+S_0+S_2+S_3+S_1}}\sum_{m_0,s_0,s_2,s_3,s_1}
\sum_{\substack{(m_0',s_0',s_2',s_3',s_1')\neq\\
(m_0,s_0,s_2,s_3,s_1)}} \tr \Theta_{x(m_0',s_0',s_2',s_3',s_1')}^{B_1}(\mathcal{C}_1)
 \rho_{x(m_0,s_0,s_2,s_3,s_1)}^{B_1}\\
 \label{hnn-4}
 &\le\frac{4}{2^{R_0+S_0+S_2+S_3+S_1}}\sum_{m_0,s_0,s_2,s_3,s_1}
 \sum_{s_1'\neq s_1} \tr \Theta_{4\,x(m_0,s_0,s_2,s_3,s_1')}^{B_1}(\mathcal{C}_1)
 \rho_{x(m_0,s_0,s_2,s_3,s_1)}^{B_1} \\
 \label{hnn-3}
  &+\frac{4}{2^{R_0+S_0+S_2+S_3+S_1}}\sum_{m_0,s_0,s_2,s_3,s_1}
 \sum_{s_2'\neq s_2,s_1'} \tr\Theta_{3\,x(m_0,s_0,s_2',s_3,s_1')}^{B_1}(\mathcal{C}_1)
 \rho_{x(m_0,s_0,s_2,s_3,s_1)}^{B_1} \\
 \label{hnn-2}
   &+\frac{4}{2^{R_0+S_0+S_2+S_3+S_1}}\sum_{m_0,s_0,s_2,s_3,s_1}
 \sum_{s_3'\neq s_3,s_1'} \tr \Theta_{2\,x(m_0,s_0,s_2,s_3',s_1')}^{B_1}(\mathcal{C}_1)
 \rho_{x(m_0,s_0,s_2,s_3,s_1)}^{B_1} \\
 \label{hnn-1}
    &+\frac{4}{2^{R_0+S_0+S_2+S_3+S_1}}\sum_{m_0,s_0,s_2,s_3,s_1}
 \sum_{s_2'\neq s_2,s_3'\neq s_3,s_1'} \tr \Theta_{1\,x(m_0,s_0,s_2',s_3',s_1')}^{B_1}(\mathcal{C}_1)
 \rho_{x(m_0,s_0,s_2,s_3,s_1)}^{B_1} \\
 \label{hnn-0}
 &+\frac{4}{2^{R_0+S_0+S_2+S_3+S_1}}\sum_{m_0,s_0,s_2,s_3,s_1}
\sum_{\substack{(m_0',s_0')\neq(m_0,s_0),\\s_2',s_3',s_1'}} \tr \Theta_{0\,x(m_0',s_0',s_2',s_3',s_1')}^{B_1}(\mathcal{C}_1)
 \rho_{x(m_0,s_0,s_2,s_3,s_1)}^{B_1},
\end{align}
where the inequality follows from $ \Theta_{x(m_0,s_0,s_2,k_1,s_3,k_2,s_1)}^{B_1}(\mathcal{C}_1)\le \Theta_{j\,x(m_0,s_0,s_2,k_1,s_3,k_2,s_1)}^{B_1}(\mathcal{C}_1)\leq I^{B_1}$ for $j=0,1,2,3$. We implement $(m_0',s_0',s_2',s_3',s_1') \neq (m_0,s_0,s_2,s_3,s_1)$ in four distinct steps, each corresponding to different layers of superposition coding. In each step, the corresponding projector based on the incorrect message is retained while others are discarded. Note also that the message pair $(m_0',s_0')\neq(m_0,s_0)$ cannot be further divided into smaller indices because both of them will either be decoded correctly or both incorrectly.

We have now obtained the following upper bound on the average error probability of the random code we constructed, given as the sum of the following equation:
\begin{align*}
     P_{e,1}(\mathcal{C}_1) = \eqref{hn-hn0}+\eqref{hn-hn1}+\eqref{hn-hn2}+\eqref{hn-hn3} +\eqref{hn-hn4} + \eqref{hnn-4} + \eqref{hnn-3} + \eqref{hnn-2} + \eqref{hnn-1} + \eqref{hnn-0}
\end{align*}

We now find the expectation of the average error probability over the choice of the random code $\mathcal{C}_1$. For the first five of the above-mentioned equations, it is not difficult to see that
\begin{align}
\nonumber
\mathbb{E}_{\mathcal{C}_1}\big( \eqref{hn-hn0}+\eqref{hn-hn1}+\eqref{hn-hn2}+\eqref{hn-hn3} +\eqref{hn-hn4}\big)&= 2\sum_{j=0}^{4}\tr (I- \Theta_{j}^{UV_2V_3XB_1})\rho^{UV_2V_3XB_1}\\
\label{hay-1}
&\le 4\sum_{j=0}^{4}\tr (I- \Theta_{j}^{UV_2V_3XB_1})\rho^{UV_2V_3XB_1}
\end{align}
For the subsequent five equations, we evaluate the expectation as follows:
\begin{align*}
    &\mathbb{E}_{\mathcal{C}_1}\big( \eqref{hnn-4}+\eqref{hnn-3}+\eqref{hnn-2}+\eqref{hnn-1} +\eqref{hnn-0}\big)=\\
 &\frac{4}{2^{R_0+S_0+S_2+S_3+S_1}}\sum_{m_0,s_0,s_2,s_3,s_1}
 \sum_{s_1'\neq s_1} \mathbb{E}_{\mathcal{C}_1}\left(\tr \Theta_{4\,x(m_0,s_0,s_2,s_3,s_1')}^{B_1}(\mathcal{C}_1)
 \rho_{v_2(m_0,s_0,s_2,k_1)v_3(m_0,s_0,s_3,k_2)}^{B_1}\right) \\
  &+\frac{4}{2^{R_0+S_0+S_2+S_3+S_1}}\sum_{m_0,s_0,s_2,s_3,s_1}
 \sum_{s_2'\neq s_2,s_1'}\mathbb{E}_{\mathcal{C}_1}\left( \tr \Theta_{3\,x(m_0,s_0,s_2',s_3,s_1')}^{B_1}(\mathcal{C}_1)
 \rho_{v_3(m_0,s_0,s_3,k_2)}^{B_1}\right) \\
   &+\frac{4}{2^{R_0+S_0+S_2+S_3+S_1}}\sum_{m_0,s_0,s_2,s_3,s_1}
 \sum_{s_3'\neq s_3,s_1'} \mathbb{E}_{\mathcal{C}_1}\left(\tr \Theta_{2\,x(m_0,s_0,s_2,s_3',s_1')}^{B_1}(\mathcal{C}_1)
 \rho_{v_2(m_0,s_0,s_2,k_1)}^{B_1}\right) \\
    &+\frac{4}{2^{R_0+S_0+S_2+S_3+S_1}}\sum_{m_0,s_0,s_2,s_3,s_1}
 \sum_{s_2'\neq s_2,s_3'\neq s_3,s_1'} \mathbb{E}_{\mathcal{C}_1}\left(\tr \Theta_{1\,x(m_0,s_0,s_2',s_3',s_1')}^{B_1}(\mathcal{C}_1)
 \rho_{u(m_0,s_0)}^{B_1} \right)\\
 &+\frac{4}{2^{R_0+S_0+S_2+S_3+S_1}}\sum_{m_0,s_0,s_2,s_3,s_1}
\sum_{\substack{(m_0',s_0')\neq(m_0,s_0),\\s_2',s_3',s_1'}} \mathbb{E}_{\mathcal{C}_1}\left(\tr \Theta_{0\,x(m_0',s_0',s_2',s_3',s_1')}^{B_1}(\mathcal{C}_1)
 \rho^{B_1}\right),\\
 &= 4 (2^{S_1}-1) \tr \Theta_{4}^{UV_2V_3XB_1}
 \rho^{UX-V_2V_3-B_1} \\
  & \hspace{1cm}+4 (2^{S_2}-1)2^{S_1} \tr \Theta_{3}^{UV_2V_3XB_1}
 \rho^{UV_2X-V_3-B_1} \\
  & \hspace{1cm}+4 (2^{S_3}-1)2^{S_1} \tr \Theta_{2}^{UV_2V_3XB_1}
 \rho^{UV_3X-V_2-B_1}  \\
  & \hspace{1cm}+ 4(2^{S_2+S_3}-1)2^{S_1} \tr \Theta_{1}^{UV_2V_3XB_1}
 \rho^{XV_2V_3-U-B_1}  \\
   & \hspace{1cm}+4 (2^{R_0+S_0}-3)2^{S_3}2^{S_2}2^{S_1} \tr \Theta_{0}^{UV_2V_3XB_1}
 \rho^{UV_2V_3XB_1}\\
  &\le 4\times2^{S_1} \tr \Theta_{4}^{UV_2V_3XB_1}
 \rho^{UX-V_2V_3-B_1} \\
  & \hspace{1cm}+4 \times2^{S_2}2^{S_1} \tr \Theta_{3}^{UV_2V_3XB_1}
 \rho^{UV_2X-V_3-B_1} \\
  & \hspace{1cm}+4 \times2^{S_3}2^{S_1} \tr \Theta_{2}^{UV_2V_3XB_1}
 \rho^{UV_3X-V_2-B_1}  \\
  & \hspace{1cm}+ 4\times2^{S_2+S_3}2^{S_1} \tr \Theta_{1}^{UV_2V_3XB_1}
 \rho^{XV_2V_3-U-B_1}  \\
   & \hspace{1cm}+4 \times2^{R_0+S_0}2^{S_3}2^{S_2}2^{S_1} \tr \Theta_{0}^{UV_2V_3XB_1}
 \rho^{UV_2V_3XB_1}
\end{align*}
Combining the last five lines of the previous evaluation with Eq. \eqref{hay-1} yields the expectation of the average error probability. To align with the formulation of the hypothesis testing lemma (Lemma \ref{hp-testing}), we organize these 10 equations in the following order:
\begin{align*}
\mathbb{E}_{\mathcal{C}_1}\big(P_{e,1}(\mathcal{C}_1)\big)&=
    4\tr (I- \Theta_{0}^{UV_2V_3XB_1})\rho^{UV_2V_3XB_1} + 4 \times2^{R_0+S_0+S_3+S_2+S_1} \tr \Theta_{0}^{UV_2V_3XB_1} \rho^{UV_2V_3XB_1}\\
   & \hspace{1cm}+ 4\tr (I- \Theta_{1}^{UV_2V_3XB_1})\rho^{UV_2V_3XB_1} + 4\times2^{S_2+S_3+S_1} \tr \Theta_{1}^{UV_2V_3XB_1}
 \rho^{XV_2V_3-U-B_1}\\
& \hspace{1cm}+ 4\tr (I- \Theta_{2}^{UV_2V_3XB_1})\rho^{UV_2V_3XB_1}
  +4 \times2^{S_3+S_1} \tr \Theta_{2}^{UV_2V_3XB_1}
 \rho^{UV_3X-V_2-B_1} \\
 & \hspace{1cm}+ 4\tr (I- \Theta_{3}^{UV_2V_3XB_1})\rho^{UV_2V_3XB_1}
  +4 \times2^{S_2+S_1} \tr \Theta_{3}^{UV_2V_3XB_1}
 \rho^{UV_2X-V_3-B_1}\\
  &\hspace{1cm} + 4\tr (I- \Theta_{4}^{UV_2V_3XB_1})\rho^{UV_2V_3XB_1}
  +4\times2^{S_1} \tr \Theta_{4}^{UV_2V_3XB_1}
 \rho^{UX-V_2V_3-B_1}\\
 &\stackrel{\text{(a)}}{\le}2^{\alpha(R_0+S_0+S_2+S_3+S_1)+2}2^{-\alpha D_{1-\alpha}\left(\mathcal{E}^{B_1}(\rho^{U V_2 V_3 XB_1}) \| \rho^{UV_2V_3X}\otimes \rho^{B_1}\right)}\\
 &\hspace{1cm}+2^{\alpha(S_2+S_3+S_1)+2}2^{-\alpha D_{1-\alpha}\left(\mathcal{E}^{UB_1}_{1}(\rho^{UV_2V_3X B_1}) \| \mathcal{E}^{B_1}(\rho^{V_2 V_3 X- U-B_1})\right)}\\
 &\hspace{1cm}+2^{\alpha(S_3+S_1)+2}2^{-\alpha D_{1-\alpha}\left( \mathcal{E}^{UV_2B_1}_{2}(\rho^{UV_2v_3XB_1}) \| \mathcal{E}^{UB_1}_{1}(\rho^{UXV_3- V_2-B_1})\right)}\\
 &\hspace{1cm}+2^{\alpha(S_2+S_1)+2}2^{-\alpha D_{1-\alpha}\left(\mathcal{E}^{UV_3B_1}_{3}(\rho^{UV_2V_3XB_1}) \| \mathcal{E}^{UB_1}_{1}(\rho^{UXV_2- V_3-B_1})\right)}\\
 &\hspace{1cm}+2^{\alpha S_1+2}2^{-\alpha D_{1-\alpha}\left(\mathcal{E}^{UV_2V_3B_1}_{4}(\rho^{UV_2V_3XB_1}) \| \mathcal{E}^{UB_1}_{1}(\rho^{UX- V_2V_3-B_1})\right)}\\
 &\stackrel{\text{(b)}}{\le} \nu 2^{\alpha(R_0+S_0+S_2+S_3+S_1)+2}2^{-\alpha\widetilde{D}_{1-\alpha}\left(\rho^{U V_2 V_3 XB_1} \| \rho^{UV_2V_3X}\otimes \rho^{B_1}\right)}\\
 &\hspace{1cm}+\nu_1 2^{\alpha(S_2+S_3+S_1)+2}2^{-\alpha\widetilde{D}_{1-\alpha}\left(\rho^{U V_2 V_3 XB_1} \| \mathcal{E}^{B_1}(\rho^{V_2 V_3 X- U-B_1})\right)}\\
 &\hspace{1cm}+\nu_2 2^{\alpha(S_3+S_1)+2}2^{-\alpha\widetilde{D}_{1-\alpha}\left(\rho^{U V_2 V_3 XB_1} \| \mathcal{E}^{UB_1}_{1}(\rho^{UXV_3- V_2-B_1})\right)}\\
 &\hspace{1cm}+\nu_3 2^{\alpha(S_2+S_1)+2}2^{-\alpha\widetilde{D}_{1-\alpha}\left(\rho^{U V_2 V_3 XB_1} \| \mathcal{E}^{UB_1}_{1}(\rho^{UXV_2- V_3-B_1})\right)}\\
 &\hspace{1cm}+\nu_4 2^{\alpha S_1+2}2^{-\alpha\widetilde{D}_{1-\alpha}\left(\rho^{U V_2 V_3 XB_1} \| \mathcal{E}^{UB_1}_{1}(\rho^{UX- V_2V_3-B_1})\right)}\\
 &\stackrel{\text{(c)}}{\le} \nu 2^{\alpha(R_0+S_0+S_2+S_3+S_1)+2-\alpha\tilde{I}_{\alpha}^{\uparrow}\left(U V_2 V_3 X; B_1\right)_{\rho^{U V_2 V_3 XB_1}}}\\
 &\hspace{1cm}+\nu_1 2^{\alpha(S_2+S_3+S_1)+2-\alpha\tilde{I}_{\alpha}^{\downarrow}\left(V_2 V_3 X; B_1|U\right)_{\rho^{U V_2 V_3 XB_1}|\rho^{U V_2 V_3 X}}}\\
 &\hspace{1cm}+\nu_2 2^{\alpha(S_3+S_1)+2-\alpha\tilde{I}_{\alpha}^{\downarrow}\left(U V_3 X; B_1|V_2\right)_{\rho^{U V_2 V_3 XB_1}|\rho^{U V_2 V_3 X}}}\\
 &\hspace{1cm}+\nu_3 2^{\alpha(S_2+S_1)+2-\alpha\tilde{I}_{\alpha}^{\downarrow}\left(U V_2 X; B_1|V_3\right)_{\rho^{U V_2 V_3 XB_1}|\rho^{U V_2 V_3 X}}}\\
 &\hspace{1cm}+\nu_4 2^{\alpha S_1+2-\alpha\tilde{I}_{\alpha}^{\downarrow}\left(U X; B_1|V_2,V_3\right)_{\rho^{U V_2 V_3 XB_1}|\rho^{U V_2 V_3 X}}}.
\end{align*}
where $\nu,\nu_1,\nu_2,\nu_3,\nu_4$ are the maximum number of distinct eigenvalues of the operators $\rho^{B_1},\{\mathcal{E}^{B_1}(\rho^{B_1}_{1|u})\}_u$, $\{\mathcal{E}^{B_1}_{1|u}(\rho^{B_1}_{v_2})\}_{u,v_2},\{\mathcal{E}^{B_1}_{1|u}(\rho^{B_1}_{v_3})\}_{u,v_3}$, and $\{\mathcal{E}^{B_1}_{1|u}(\rho^{B_1}_{v_2,v_3})\}_{u,v_2,v_3}$, respectively. Here, where (a) follows from Lemma \ref{hp-testing}, (b) from Lemma \ref{petz-sandwich}, and (c) from the definition of the R\'enyi quantum mutual information quantities.
\medskip

\begin{remark}
      In order to transmit the two parts of message $M_1$, namely $M_{12}$ with rate $S_2$ and $M_{13}$ with rate $S_3$, we use Marton's code. However, since both encoded messages in Marton's code are for the same receiver $B_1$, its analysis differs from Marton's original code, where the messages are for two different receivers. This difference leads to new rates for the messages. Roughly speaking, if the codewords $(v_2(m_0,s_0,s_2,k_1),v_3(m_0,s_0,s_3,k_2))$ were for distinct receivers, we would need restrictions on the total number of sequences of variables $v_2$ and $v_3$. However, when considering both sequences together, the decoder can eliminate those that do not meet the condition of mutual conditional lemma, reducing its list-size accordingly. This means we only need to impose restrictions on the rates of the messages $S2$ and $S_3$. This is also evident in the classical counterpart of this problem \cite{Nair-Elgamal}, where in Eqs. (20-23), the restrictions are placed on $S_2$ and $S_3$ rather than $T_2$ and $T_3$.
    \end{remark}

\bigskip

\textbf{\emph{Receiver $B_2$:}} Receiver $B_2$ decodes the message pair $(m_0,s_0)$ encoded into the variable $U$ via non-unique decoding using $v_2(m_0,s_0,s_2,k_1)$. Specifically, the decoder searches for both variables $(s_2, k_1)$, with no error occurring if it makes an incorrect estimate for either. It is worth noting that receiver $B_1$ also employs a non-unique, aiming to decode the message $s_2$ via non-unique decoding through $k_1$. Receiver $B_2$ successfully decodes the message $m_{10}$ since it is encoded toghether with the common message into variable $U$. However, the information contained in message $m_{10}$ is not intended to $B_1$ and it may be discarded after being decoded. This differs from its decoding of message pair $(s_{2},k_1)$, where successful decoding may or may not occur -- either case is acceptable. The decoding measurement is based on the following projector:
\begin{align*}
   Q^{V_2B_2}= \{\mathcal{E}^{B_2}\left(\rho^{V_2B_2}\right)\ge 2^{R_0+S_0+S_2+r_1} \rho^{V_2}\otimes \rho^{B_2}\}.
\end{align*}
We further define the projector
\begin{align*}
Q_{v_2}^{B_2}\coloneqq\bra{v_2} Q^{B_2}\ket{v_2}.
\end{align*}
The above projector is not related to any code. However, as the decoder of $B_2$, we can build the following POVM for the constructed $\mathcal{C}_1$: 
\begin{align*}
\Lambda_{m_0,s_0}^{B_2}(\mathcal{C}_1)\coloneqq
\left(\sum_{m_0',s_0'} Q^{B_2}_{v_2(m_0',s_0')}\right)^{-\frac{1}{2}} Q^{B_2}_{v_2(m_0,s_0)}
\left(\sum_{m_0',s_0'} Q^{B_2}_{v_2(m_0',s_0')}\right)^{-\frac{1}{2}},
\end{align*} 
where 
\begin{align*}
Q^{B_2}_{v_2(m_0,s_0)}=\sum_{s_2'=1}^{2^{S_2}}\sum_{k_1'=1}^{2^{r_1}}Q^{B_2}_{v_2(m_0,s_0,s_2',k_1')}.
\end{align*}

The average error probability of the code $\mathcal{C}_1$ for receiver $B_2$ using the aforementioned POVM is bounded as follows. Recall that we use the variable $\hat{M}_{tb_i}$ to denote the estimate made by receiver $B_i$ about that message $t$. We obtain:

\begin{align*}
P_{e,2}(\mathcal{C}_1)&=\text{Pr}\{\hat{M}_{0b_2}\neq M_0\}\\
&\stackrel{\text{(a)}}{\le}
\text{Pr}\{(\hat{M}_{0b_2},\hat{M_{10b_2}})\neq(M_0,M_{10})\}\\
&\coloneqq
 \frac{1}{2^{R_0+S_0}}\sum_{m_0,s_0} \tr (I-\Lambda_{m_0,s_0}^{B_2}(\mathcal{C}_1))
\rho^{B_2}_{v_2(m_0,s_0,s_2,k_1)}\\
&\stackrel{\text{(b)}}{\le}
\frac{1}{2^{R_0+S_0}}
 \sum_{m_0,s_0} \tr \Big(
 2(I- Q_{v_2(m_0,s_0)}^{B_2})+
 4\sum_{ (m_0',s_0')\neq (m_0,s_0)} Q_{v_2(m_0',s_0')}
\Big) \rho^{B_2}_{v_2(m_0,s_0,s_2,k_1)} \\
&\stackrel{\text{(c)}}{\le}
\frac{1}{2^{R_0+S_0}}
 \sum_{m_0,s_0} \tr \Big(
 2(I- \sum_{s_2'=1}^{2^{S_2}}\sum_{k_1'=1}^{2^{r_1}}Q^{B_2}_{v_2(m_0,s_0,s_2',k_1')})\\
 &\hspace{4cm}+
 4\sum_{ (m_0',s_0')\neq (m_0,s_0)} \sum_{s_2'=1}^{2^{S_2}}\sum_{k_1'=1}^{2^{r_1}}Q^{B_2}_{v_2(m_0,s_0,s_2',k_1')}
\Big) \rho^{B_2}_{v_2(m_0,s_0,s_2,k_1)}\\
&=
\frac{2}{2^{R_0+S_0}}
 \sum_{m_0,s_0} \tr 
 (I- Q^{B_2}_{v_2(m_0,s_0,s_2,k_1)}-\sum_{s_2'\neq s_2}\sum_{k_1'\neq k_1}Q^{B_2}_{v_2(m_0,s_0,s_2',k_1')})\rho^{B_2}_{v_2(m_0,s_0,s_2,k_1)}\\
 &\hspace{4cm}+
 \frac{4}{2^{R_0+S_0}}\sum_{m_0,s_0}\sum_{ (m_0',s_0')\neq (m_0,s_0)} \sum_{s_2'=1}^{2^{S_2}}\sum_{k_1'=1}^{2^{r_1}}Q^{B_2}_{v_2(m_0,s_0,s_2',k_1')} \rho^{B_2}_{v_2(m_0,s_0,s_2,k_1)}\\
 &\stackrel{\text{(d)}}{\le}
\frac{2}{2^{R_0+S_0}}
 \sum_{m_0,s_0} \tr 
 (I- Q^{B_2}_{v_2(m_0,s_0,s_2,k_1)})\rho^{B_2}_{v_2(m_0,s_0,s_2,k_1)}\\
 &\hspace{4cm}+
 \frac{4}{2^{R_0+S_0}}\sum_{m_0,s_0}\sum_{ (m_0',s_0')\neq (m_0,s_0)} \sum_{s_2'=1}^{2^{S_2}}\sum_{k_1'=1}^{2^{r_1}}Q^{B_2}_{v_2(m_0,s_0,s_2',k_1')} \rho^{B_2}_{v_2(m_0,s_0,s_2,k_1)},
\end{align*}
where (a) follows by including the message $M_{10}$ (of rate $S_0$) which is encoded with the common message $M_0$ and will be inevitably decoded along the common message, even though it might be discarded later. (b) comes from Hayashi-Nagaoka inequality Lemma \ref{hayashi-nagaoka}, (c) from the definition and (d) comes about by discarding all but those corresponding to $(s_2,k_1)$ in the summation.

Next, we find the expectation of the average error probability over the choice of the random code:

\begin{align*}
    \E_{\mathcal{C}_1}(P_{e,2}(\mathcal{C}_1))
 &\le
\frac{2}{2^{R_0+S_0}}
 \sum_{m_0,s_0} \E_{\mathcal{C}_1}\left(\tr 
 (I- Q^{B_2}_{v_2(m_0,s_0,s_2,k_1)})\rho^{B_2}_{v_2(m_0,s_0,s_2,k_1)}\right)\\
 &\hspace{3cm}+
 \frac{4}{2^{R_0+S_0}}\sum_{m_0,s_0}\sum_{ (m_0',s_0')\neq (m_0,s_0)} \sum_{s_2'=1}^{2^{S_2}}\sum_{k_1'=1}^{2^{r_1}}\E_{\mathcal{C}_1}\left(Q^{B_2}_{v_2(m_0,s_0,s_2',k_1')} \rho^{B_2}\right)\\
 &=2\tr(I-Q^{V_2B_2})\rho^{UV_2}
 +
 4\times (2^{R_0+S_0}-1)2^{S_2+r_1}\tr Q^{V_2B_2}(\rho^{V_2}\otimes\rho^{B_2})\\
 &\le 4\tr(I-Q^{V_2B_2})\rho^{UV_2}
 +
 4\times 2^{R_0+S_0+S_2+r_1}\tr Q^{V_2B_2}(\rho^{V_2}\otimes\rho^{B_2})\\
 &\stackrel{\text{(a)}}{\leq} 2^{\alpha(R_0+S_0+S_2+r_1)+2} 2^{-\alpha D_{1-\alpha}\left(\mathcal{E}^{B_2}(\rho^{V_2B_2}) \| \rho^{V_2}\otimes \rho^{B_2}\right)}  \\
 &\stackrel{\text{(b)}}{\leq} \mu 2^{\alpha(R_0+S_0+S_2+r_1)+2} 2^{-\alpha \widetilde{D}_{1-\alpha}\left(\rho^{V_2B_2} \| \rho^{V_2}\otimes \rho^{B_2}\right)}  \\
&\stackrel{\text{(c)}}{\leq} \mu 2^{\alpha(R_0+S_0+S_2+r_1)+2} 
2^{-\alpha \tilde{I}_{\alpha}^{\uparrow}(V_2;B_2)_{\rho^{V_2B_2}}},
\end{align*}
where $\mu$ is the number of the distinct eigenvalues of the operator $\rho^{B_2}$.
Here, where (a) follows from Lemma \ref{hp-testing}, (b) from Lemma \ref{petz-sandwich}, and (c) from the definition of the R\'enyi quantum mutual information quantities.

\bigskip

\textbf{\emph{Receiver $B_3$:}} Receiver $B_3$ decodes the information $(m_0,s_0)$ encoded into the variable $U$via non-unique decoding using $V_2(m_0,s_0,s_3,k_2)$. Note that similar to receiver $B_2$, this receiver also successfully decodes $M_{10}$ as it is tied to the common message $M_0$. However, $B_3$ may discard this information after successfully decoding it. Note also that his differs from its decoding of $M_{13}$, where successful decoding may or may not occur -- either case is acceptable. The construction of the POVM and the analysis of error probability follow the same approach as for receiver $B_2$. The expectation of the average error probability can be expressed as follows:

 \begin{align}
\E_{\mathcal{C}_1} \left(P_{e,3}(\mathcal{C}_1)\right)
\le \eta 2^{\alpha(R_0+S_0+S_3+r_2)+2} 
2^{-\alpha \tilde{I}_{\alpha}^{\uparrow}(V_3;B_3)_{\rho^{V_3B_3}}},
\end{align}
where $\eta$ is the number of the distinct eigenvalues of the operator $\rho^{B_3}$.

\subsection{Asymptotic analysis}
As mentioned before, our proof consists in creating a one-shot code and then determining the rate region through asymptotic analysis.
 The key observation is that the number of distinct eigenvalues will vanish because they are only polynomial and that R\'enyi relative entropy tends to von Nuemann relative entropy when its parameter tends to $1$.

\begin{proofof}[Theroem \ref{general-two-degraded.t}]
Our proof consists in creating a one-shot code and then determining the rate region through asymptotic analysis. We break down the process and go through the details in distinct subsections.

Concerning the error probability of the encoder, leveraging Eq. \ref{max-asymp}, we can conclude that if the condition $r_1 + r_2 \geq I(V_2;V_3\vert U)$ holds true, then a pair of codewords in each product bin can be identified, such that they mimic a joint distribution. This is further supported by the asymptotic mutual covering lemma \cite{1056302}, \cite[Lemma 8.1]{elgamal-kim}. As for the error probability of the decoders, 
the upper bounds on the expectation of the average error probability imply the existence of at least one good code $\mathcal{C}_1$ that satisfies these bounds. Consequently, we can assess the upper bounds in the following manner:

\begin{align*}
    P_{e,1}^{(n)}(\mathcal{C}_n)&\leq 
    4 (n+1)^{\alpha (d_{B_1}-1)} 2^{\alpha n(R_0 +S_0+ S_{2} + S_3 + S_{1})} 2^{-\alpha\tilde{I}_{1-\alpha}^{\uparrow}(U^n,V_2^n,V^n_3,X^n;B_1^n)_{(\rho^{UV_2V_3XB_1})^{\otimes n}}}\\
    &\hspace{0.1cm}+4  (n+1)^{\alpha d_{U}(d_{B_1}+2)(d_{B_1}-1)/2} 2^{\alpha n(S_{2}+S_3+S_{1})} 2^{-\alpha \tilde{I}_{1-\alpha}^{\downarrow}\left(V_2^n,V_3^n,X^n;B_1^n|U^n\right)_{(\rho^{UV_2V_3XB_1})^{\otimes n}\vert(\rho^{UV_2V_3X})^{\otimes n}}}\\
    &\hspace{0.1cm}+4  (n+1)^{\alpha d_{V_2}d_U(d_{B_1}+2)(d_{B_1}-1)/2} 2^{\alpha n(S_3+S_1)} 2^{-\alpha \tilde{I}_{1-\alpha}^{\downarrow}\left(U^n,V_3^n,X^n;B_1^n|V_2^n\right)_{(\rho^{UV_2V_3XB_1})^{\otimes n}\vert(\rho^{UV_2V_3X})^{\otimes n}}}\\
    &\hspace{0.1cm}+4  (n+1)^{\alpha d_{V_3}d_U(d_{B_1}+2)(d_{B_1}-1)/2} 2^{\alpha n(S_2+S_1)} 2^{-\alpha \tilde{I}_{1-\alpha}^{\downarrow}\left(U^n,V_2^n,X^n;B_1^n|V_3^n\right)_{(\rho^{UV_3V_2XB_1})^{\otimes n}\vert(\rho^{UV_3V_2X})^{\otimes n}}}\\
    &\hspace{0.1cm}+4  (n+1)^{\alpha d_{V_2}d_{V_3}d_U(d_{B_1}+2)(d_{B_1}-1)/2} 2^{\alpha nS_1} 2^{-\alpha \tilde{I}_{1-\alpha}^{\downarrow}\left(U^n,X^n;B_1^n|V_2^n,V_3^n\right)_{(\rho^{UV_3V_2XB_1})^{\otimes n}\vert(\rho^{UV_3V_2X})^{\otimes n}}}\\
  P_{e,2}^{(n)}(\mathcal{C}_n)&\leq 
4 (n+1)^{\alpha (d_{B_2}-1)} 2^{\alpha n(R_0+S_0+S_2+r_1)} 2^{-\alpha\tilde{I}_{1-\alpha}^{\uparrow}(V_2^n;B_2^n)_{(\rho^{V_2B_2})^{\otimes n}}},\\
  P_{e,3}^{(n)}(\mathcal{C}_n)&\leq 
4 (n+1)^{\alpha (d_{B_3}-1)} 2^{\alpha n(R_0+S_0+S_3+r_2)} 2^{-\alpha\tilde{I}_{1-\alpha}^{\uparrow}(V_3^n;B_3^n)_{(\rho^{V_3B_3})^{\otimes n}}},
\end{align*}
where we have employed the polynomial bounds on the number of distinct eigenvalues from Proposition \ref{pinching-asymptotic}. We further obtain (notice that the number of distinct eigenvalues vanish because they are only polynomial and that R\'enyi relative entropy tends to von Nuemann relative entropy when its parameter tends to $1$):
\begin{align*}
  \lim_{n\to\infty }-\frac{1}{n}\log P_{e,1}^{(n)}(\mathcal{C}_n)\ge
  \min\bigg(&\alpha \big(\tilde{I}_{1-\alpha}^{\uparrow}(U,V_2,V_3,X;B_1)_{\rho^{UV_2V_3XB_1}}-(R_0 +S_0+ S_{2} + S_3 + S_{1})\big),\\
  &\alpha \big(\tilde{I}_{1-\alpha}^{\downarrow}\left(V_2,V_3,X;B_1|U\right)_{\rho^{UV_2V_3XB_1}|\rho^{UV_2V_3X}}-(S_{2}+S_3+S_{1})\big),\\
  &\alpha \big(\tilde{I}_{1-\alpha}^{\downarrow}\left(U,V_3,X;B_1|V_2\right)_{\rho^{UV_2V_3XB_1}|\rho^{UV_2V_3X}}-(S_3+S_1)\big)\\
  &\alpha \big(\tilde{I}_{1-\alpha}^{\downarrow}\left(U,V_2,X;B_1|V_3\right)_{\rho^{UV_2V_3XB_1}|\rho^{UV_2V_3X}}-(S_2+S_1)\big)\\
  &\alpha \big(\tilde{I}_{1-\alpha}^{\downarrow}\left(U,X;B_1|V_2,V_3\right)_{\rho^{UV_2V_3XB_1}|\rho^{UV_2V_3X}}-S_1\big)
  \bigg),
\end{align*}
\begin{align*}
\lim_{n\to\infty }-\frac{1}{n}\log P_{e,2}^{(n)}(\mathcal{C}_n)&\ge
  \alpha \big(\tilde{I}_{1-\alpha}^{\downarrow}(V_2;B_2)_{\rho^{V_2B_2}}-(R_0+S_0+S_2+r_1)\big),\\
\lim_{n\to\infty }-\frac{1}{n}\log P_{e,3}^{(n)}(\mathcal{C}_n)&\ge
  \alpha \big(\tilde{I}_{1-\alpha}^{\downarrow}(V_3;B_3)_{\rho^{V_3B_3}}-(R_0+S_0+S_3+r_2)\big).
\end{align*}
It now follows that as $n\to\infty$ and $\alpha\to 0$, there exists a sequence of codes $\mathcal{C}_n$ such that $P_{e,1}^{(n)}(\mathcal{C}_n),P_{e,2}^{(n)}(\mathcal{C}_n)$ and $P_{e,3}^{(n)}(\mathcal{C}_n)$ vanish, if
\begin{align*}
r_1 + r_2 &\geq I(V_2;V_3\vert U)\\
R_0 + R_1 &\leq I(X;B_1)_{\rho},\\
    S_{3}+S_2+S_{1}&\leq I(X;B_1|U)_{\rho},\\
     S_{3}+S_1&\leq I(X;B_1|V_2)_{\rho},\\
     S_{2}+S_1&\leq I(X;B_1|V_3)_{\rho},\\
     S_1&\leq I(X;B_1\vert V_2,V_3)_{\rho}, \\
    R_0 + S_{0} + S_2 +r_1 &\leq I(V_2;B_2)_{\rho},
    \\
    R_0 + S_{0} + S_3 +r_2 &\leq I(V_3;B_3)_{\rho},
\end{align*}
where the mutual information quantities are calculated for the state in the statement of the Theorem \ref{general-two-degraded.t}. Substituting $R_1=S_0 + S_2+ S_3+S_1$ and using the Fourier–Motzkin elimination procedure, we remove $S_{1},S_{2}$, it follows from the union bound $P_{e}^{(n)}(\mathcal{C}_n)\le P_{e,1}^{(n)}(\mathcal{C}_n)+P_{e,2}^{(n)}(\mathcal{C}_n)+P_{e,3}^{(n)}(\mathcal{C}_n)$ that the probability of decoding error tends to zero as $n\to \infty$ if the inequalities in Theorem are satisfied. This completes the proof of Theorem \ref{general-two-degraded.t}.

\end{proofof}

\bigskip
\bigskip

\subsection{No-go region}

The following theorem establishes a weak converse for the setting of Theorem \ref{general-two-degraded.t}. Unlike the converse of the three-receiver multilevel quantum broadcast channel, the converse below does not match the achievability region in general, meaning that it is not tight even in the classical setting. 

\begin{theorem}\label{converse-2}
      Assuming a transmitter and three receivers have access to many independent uses of a three-receiver quantum broadcast channel $x\to\rho_x^{B_1B_2B_3}$ two-degraded message set, if the set of rate pairs $(R_0,R_1)$ becomes asymptotically achievable, then they are contained in the following region:
    \begin{align*}
        R_0 & \leq \min\{I(U;B_1),I(V_2;B_2)-I(V_2;B_1|U),\\
        &\hspace{1.35cm} I(V_3;B_3)-I(V_3;B_1|U)\},\\
        R_1 & \leq I(X;B_1|U),
    \end{align*}
    for some state 
     \begin{align*}
        \rho^{UXV_1V_2B_1B_2B_3}=\sum_{x}p(x)\ketbra{x}\otimes\rho_x^{UV_2V_3}\otimes\rho^{B_1B_2B_3}_x,
    \end{align*}
    such there systems $(X,V_3)$ are degradable with respect to $U$, and simultaneously systems $(X,V_2)$ are degradable with respect to $U$ (systems $(U,V_2,V_3)$ are generally quantum).
\end{theorem}
\begin{proof}
   The proof follows along the same lines as its classical counterpart, with careful attention to quantum systems. Here, we offer only a high-level outline of the process.
We identify the auxiliary systems as follows: $U_i=(M_0,B_1^{i-1}),V_{2i}=(U_i,B_{2\,i+1}^{n})$, and $V_{3i}=(U_i,B_{3\,i+1}^{n})$. By standard arguments, the following two inequalities immediately follow $R_0\le I(U;B_1)$ and $R_1\leq I(X;B_1\vert U)$ (see \cite[Sec. 5.4.1]{elgamal-kim}) for these arguments). The other two inequalities are also follows from the classical counterpart, but we verify each step to ensure their applicability in the quantum setting.
By Fano's inequality, 
\begin{align*}
    S(M_0\vert B_2^n)\le nR_0 P_{e,2}^{(n)}+1\le n\varepsilon_n,
\end{align*}
where $\lim_{n\to\infty}\varepsilon_n=0$. Therefore,
\begin{align*}
    nR_0&\le I(M_0;B_2^n) +n\varepsilon_n \\
    &\leq  \sum_i I(M_0;B_{2i}\vert B_{2\, i+1}^{\,\,n}) + n\varepsilon_n\\
    &\le \sum_i I(M_0,B_{2\, i+1}^{\,\,n}, B_1^{i-1};B_{2i}) - I(B_1^{i-1},B_{2i}\vert M_0,B_{2\, i+1}^{\,\,n}) + n\varepsilon_n\\
&\stackrel{\text{(a)}}{=} \sum_i  I(M_0,B_{2\, i+1}^{\,\,n}, B_1^{i-1};B_{2i}) - I(B_{2\, i+1}^{\,\,n};B_{1\,i}\vert M_0,B_1^{i-1}) + n\varepsilon_n\\
&= \sum_i I(V_{2i};Y_{2i}) - I(V_{2i};Y_{1i}\vert U_i),
\end{align*}
where (a) follows from Csisz\'ar sum identity, replacing classical variables with quantum systems. The other inequality is straightforward.
\end{proof}
\begin{remark}
    Note that the aforementioned no-go region seems to be very different from the achievability region in Theorem \ref{general-two-degraded.t}. However, by summing these inequalities in an appropriate manner, we can observe a set of inequalities whose structure resembles those of the achievable region, although they are not equal.
\end{remark}
\begin{remark}
    The aforementioned outer bound reduces to the outer bound corresponding to the multilevel channel, Theorem \ref{converse-1}. To see this, recall that the defining property of a three-receiver multilevel quantum broadcast channel is that system $B_2$ is degradable with respect to system $B_1$. Then from quantum data processing \cite{Hayashi2017-cv,Wilde_2013}, we observe that
    \begin{align*}
        R_0&\leq I(V_2;B_2)-I(V_2;B_1\vert U)\\
         &\le I(V_2;B_2)-I(V_2;B_2\vert U)= I(U;B_2).     
    \end{align*}
Therefore, it is easily seen that the following region contains the converse region give above in Theorem \ref{converse-2}:
\begin{align*}
     R_0 & \leq \min\{I(U;B_2),I(V_3;B_3)\},\\
        R_1 & \leq I(X;B_1|U),\\
        R_0 +R_1 &\leq I(V_3;B_3)+ I(X;B_1\vert V_3).
\end{align*}
However, the achievability of this region was shown in Theorem \ref{two-degraded-multilevel} ; hence, the outer bound in Theorem \ref{converse-2} is tight in this sense.
\end{remark}

\bigskip

\section{General three-receiver quantum broadcast channel with three-degraded message set}
\label{general-three-degraded}
In this section we study the most general setting: The general three-receiver quantum broadcast channel $x\to\rho^{B_1B_2B_3}_x$ with three-degraded message sets. In this setup, a common message $M_0$ is intended for all three receivers, while $M_1$ is targeted for receivers $B_1$ and $B_2$, and $M_2$ is meant only for $B_1$. This configuration corresponds to Fig. \ref{three-degraded-bc}. 
We will present our results in the asymptotic setting, considering the scenario where the channel can be utilized $n$ times, with $n\to\infty$. For this reason, we don't define a one-shot explicitly. Instead, we define an asymptotic code for the three-receiver quantum broadcast channel with three-degraded message sets.

\begin{figure}[t]
\includegraphics[width=15cm]{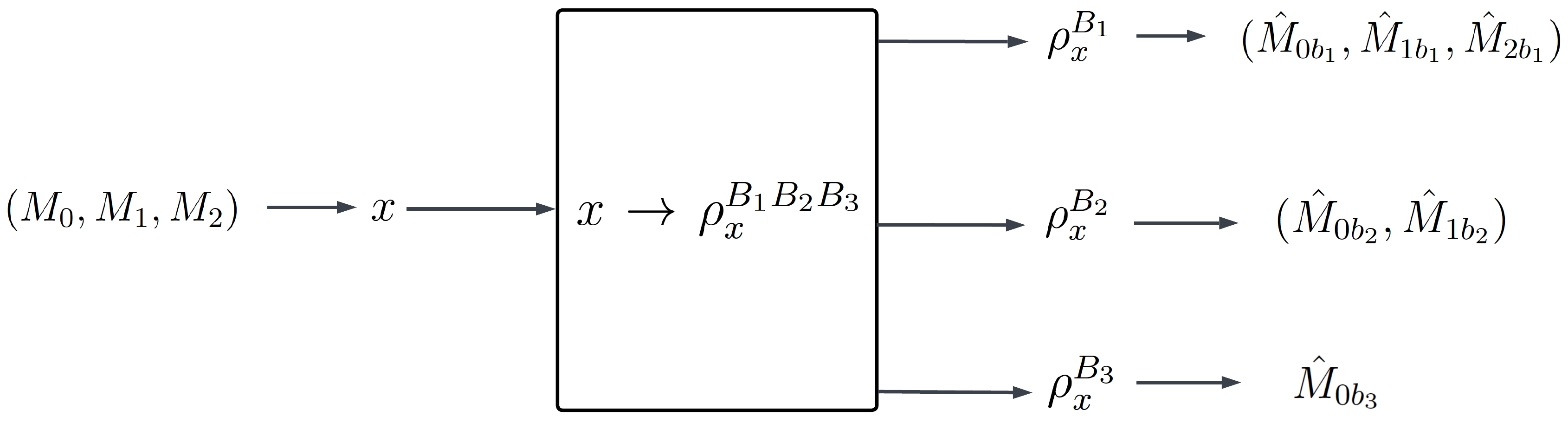}
\centering
\caption{General three-receiver quantum broadcast channel with three-degraded message set. There are three independent messages $(M_0,M_1,M_2)$ such that $M_0$ is a common message intended for all receivers, $M_1$ is intended for receivers $B_1$ and $B_2$ and $M_2$ is intended solely for receiver $B_1$. $\hat{M}_{jb_i}$ represents the estimate made by receiver $B_i$ for the message $j=0,1,2$.}
\label{three-degraded-bc}
\end{figure}

\begin{definition}\label{code-def}
A $(2^{nR_0},2^{nR_1},2^{nR_2},n,P^{(n)}_e)$ three-degraded message set code $\mathcal{C}_n$ for the three-receiver quantum broadcast channel $x\to\rho_x^{B_1B_2B_3}$ consists of the following components:
\begin{itemize}
    \item A triple of messages $(M_0,M_1,M_2)$ uniformly distributed over $[1:2^{nR_0}]\times[1:2^{nR_1}]\times[1:2^{nR_2}]$, where $R_0, R_1$ and $R_2$ represent the rates of the messages $M_0,M_1$ and $M_2$, respectively.
    \item An encoder that assigns a codeword $x^n(m_0,m_1,m_2)$ to each message triple $(m_0,m_1,m_2)\in[1:2^{nR_0}]\times[1:2^{nR_1}]\times[1:2^{nR_2}]$,
    \item Three POVMs, $\Lambda_{m_0m_1m_2}^{B_1^{n}}$ on $B_1^{\otimes n}$, $\Lambda_{m_0m_1}^{B_2^{n}}$ on $B_2^{\otimes n}$ and $\Lambda_{m_0^{n}}^{B_3}$ on $B_3^{\otimes n}$ (indicating which messages are intended for decoding by each receiver) which satisfy
    \begin{align}
\tr{\rho_{x^n(m_0,m_1,m_2)}^{B_1^{n}B_2^{n}B_3^{n}}\left(\Lambda_{m_0m_1m_2}^{B_1^{n}}\otimes\Lambda_{m_0m_1}^{B_2^{n}}\otimes\Lambda_{m_0}^{B_3^{n}}\right)}\geq 1-P^{(n)}_e(\mathcal{C}_n),
    \end{align}
    for every $(m_0,m_1,m_2)\in[1:2^{nR_0}]\times[1:2^{nR_1}]\times[1:2^{nR_2}]$. $P^{(n)}_e(\mathcal{C}_n)$ is the average probability of error defined as follows:
    \begin{align*}
        P^{(n)}_e(\mathcal{C}_n)= \text{Pr}\{\hat{M}_{0b_1},\hat{M}_{0b_2},\hat{M}_{0b_3}\neq M_0 \,\text{or}\,\hat{M}_{1b_1},\hat{M}_{1b_2}\neq M_1\,\text{or}\, \hat{M}_{2b_1}\neq M_2\},
    \end{align*}
 where $\hat{M}_{jb_i}$ represents the estimate made by receiver $B_i$ for the message $j$.
\end{itemize}
A rate triple $(R_0,R_1,R_2)$ is said to be achievable if there exists a sequence of $(2^{nR_0},2^{nR_1},2^{nR_2},n,P^{(n)}_e)$ three-degraded message set codes $\mathcal{C}_n$ such that $P^{(n)}_e(\mathcal{C}_n)\to 0$ as $n\rightarrow \infty$. The capacity region is the closure of the set of all such achievable rate triples. 
\end{definition}

The construction of one-shot code and the analysis of error probability follow a similar pattern to the preceding scenarios. Therefore, we only outline a high-level overview of the process.

\begin{theorem}\label{general-three-degraded.t}
    We assume a transmitter and three receivers have access to many independent uses of a three-receiver quantum broadcast channel $x\to\rho_x^{B_1B_2B_3}$. The set of rate triples $(R_0,R_1,R_2)$ becomes asymptotically achievable if
\begin{align*}
    R_0 &\le I(V_3;B_3),\\
    R_0 + R_1 & \leq \min\{I(V_2;B_2),I(V_2;B_2|U)+I(V_3;B_3)-I(V_2;V_3|U)\},\\
    2R_0+R_1 & \leq I(V_2;B_2)+ I(V_3;B_3)-I(V_2;V_3|U),\\
    R_0+R_1+R_2 &\leq \min\{I(X;B_1),I(V_2;B_2)+I(X;B_1|V_2),\\
    &I(V_3;B_3)+I(X;B_1|V_3),
    I(V_2;B_2|U)+I(V_3;B_3)+I(X;B_1|V_2,V_3)-I(V_2,V_3|U)\},\\
    2R_0+R_1+R_2 &\leq I(V_2;B_2) + I(V_3;B_3) + I(X;B_1|V_2,V_3)-I(V_2;V_3|U),\\
    2R_0+2R_1+R_2 &\leq I(V_2;B_2) + I(V_3;B_3) + I(X;B_1|V_3)-I(V_2;V_3|U),\\
     2R_0+2R_1+2R_2 &\leq \min\{I(V_2;B_2) + I(V_3;B_3) + I(X;B_1|V_2)+ I(X;B_1|V_3)-I(V_2;V_3|U),\\
     &I(V_2;B_2|U)+I(V_3;B_3)+I(X;B_1|U)+I(X;B_1|V_2,V_3)-I(V_2;V_3|U)\},
\end{align*}
for some probability distribution
\begin{align*}
    p(u,v_2,v_3,x) =& p(u)(v_2|u)p(x,v_3|v_2)\\
    =& p(u)(v_3|u)p(x,v_2|v_3),
\end{align*}
i.e. the Markov chains $U-V_2 - (V_3,X)$ and $U-V_3 - (V_2,X)$ are simultanoeusly satisfied,
 resulting in the cq-state $\rho^{UXV_2V_3B_1B_2B_3}$.
\end{theorem}

\bigskip

Our proof consists in creating a one-shot code and then determining the rate region through asymptotic analysis. We break down the process and go through the details in distinct subsections.

\subsection{One-shot code construction}
The asymtotic analysis of the code developped here recovers Theorem \ref{general-three-degraded.t}.
Consider a three-receiver  cq-broadcast channel $x\to\rho_x^{B_1B_2B_3}$, with corresponding marginal channels 
$x\to \rho^{B_1}_x$, $x\to \rho^{B_2}_x$, and $x\to \rho^{B_3}_x$. We define a three-degraded message set $\{M_0,M_1,M_2\}$, where
$M_0$ is to be sent to all three receivers, $M_1$ is to be sent to receivers $B_1$ and $B_2$, and $M_2$ is to be sent to $B_1$ only. We constrcut a one-shot code for this communication scenario where we use pinching technique to construct POVMs and bound the decoding errors and use Sen's mutual covering lemma for Marton coding \cite[Fact 2]{Sen2021}.  

\subsubsection{Rate splitting}
The message $M_1$ is split into two independent messages $(M_{10},M_{11})$ of corresponding rates $(R_{10},R_{11})$. Similarly, $M_2$ is split into four independent messages $\{M_{2i}\}_{i=0,1,2,3}$ with respective rates $\{S_{i}\}_{i=0,1,2,3}$, such that $R_1=S_0 + S_1 + S_2 + S_3$. We denote the realizations of these latter messages by the lowercase letter $s$, where each $s_i\in M_{2i}$ is uniformly distributed in $[1:2^{S_i}]$, for $i=0,1,2,3$. (However, we use lowercase $m$ for the realizations of the messages $M_0$ and $M_1$ with corresponding subscripts.) Using superposition and Marton coding, the message triple $(m_0,m_{10},s_{0})$ is represented by $U$, the pair $(m_{11},s_2)$ by $V_2$, the message $s_3$ by $V_3$, and message $s_1$ by $X$. Receiver $B_1$ find the message tuple $(m_0,m_{10},m_{11},s_0,s_1,s_2,s_3)$ by decoding for $X$, receiver $B_2$ find $(m_0,m_{10},m_{11})$ by decoding $(U,V_2)$ and $B_3$ find $m_0$ by non-unique decoding of $U$ from $V_3$ (so only receiver $B_3$ employs non-unique decoding). Note that the code resembles that of the scenario from the previous section, where the message $M_1$, intended for receivers $B_1$ and $B_2$, is split into two messages. One is encoded with variable $U$ as a common message, and the other is encoded into $V_2$ along with message $M_{22}$.

\hfill\\
\vspace{-0.4cm}

\subsubsection{Codebook generation}
Let $T_2\geq S_2$ and $T_3\geq S_3$ and fix a probability distribution satisfying the Markov chain $U-V_2-(X,V_3)$ and $U-V_3-(X,V_2)$, i.e. $p(u,v_2,v_3,x) = p(u)(v_2|u)p(x,v_3|v_2)
    = p(u)(v_3|u)p(x,v_2|v_3)$.
A code $\mathcal{C}_1$ is composed of tuples, each containing subsets of elements from $\mathcal{U}$, $\mathcal{V}_2,\mathcal{V}_3$, and $\mathcal{X}$. These tuples are generated as follows:
Randomly and independently generate $2^{R_0+R_{10}+S_0}$ elements 
$\{u(m_0,m_{10},s_0)\}$, $m_0 \in [1:2^{R_0}],m_{10} \in [1:2^{R_{10}}],s_0 \in [1:2^{S_0}]$ according to the pmf $P_U(u)$.
For each $u(m_0,m_{10},s_0)$, generate $(\text{a})$ $2^{R_{11}+S_2+r_1}$ elements
$\{v_2(m_0,m_{10},s_0,s_2,k_1)\}$, $_{(s_2,k_1)\in [1:2^{S_2}]\times[1:2^{r_1}]}$, randomly and conditionally independent from $P_{v_2|u}$, $(\text{b})$
$2^{S_3+r_2}$ elements
$\{v_3(m_0,m_{10},s_0,t_3)\}$, $_{(s_3,k_2)\in [1:2^{S_3}]\times[1:2^{r_2}]}$, randomly and conditionally independent from $P_{v_3|u}$.
Randomly partition $2^{R_{11}+S_2+r_1}$ elements $v_2(m_0,s_0,s_2,k_1)$ into $2^{R_{11}+S_2}$ equal size bins and $2^{S_3+r_2}$ elements $v_2(m_0,m_{10},s_0,s_3,k_2)$ into $2^{S_3}$ equal size bins.
To ensure that each product bin $((m_{11},s_2),s_3)$ contains a suitable pair $(v_2(m_0,m_{10},s_0,s_2,k_1),v_2(m_0,m_{10},s_0,s_3,k_2))$ for encoding with high probability, we require that
\begin{align}\label{covering-condition}
    r_1+r_2\geq I_{\text{max}}^{\varepsilon}(V_2;V_3|U)+2\log\frac{1}{\varepsilon}.
\end{align}
This follows from mutual covering lemma \cite[Fact 2]{Sen2021}.
Roughly speaking, this requirement means that there exists at least one pair $(v_2(m_0,m_{10},s_0,s_2,k_1),v_2(m_0,m_{10},s_0,s_3,k_2))$ in each product bin $((m_{11},s_2),s_3)$ whose joint distribution is close to $p(v_2,v_3|u)$ although they are generated from $p(v_2|u)p(v_3|u)$).
Finally, for each pair $(v_2(m_0,m_{10},s_0,s_2,k_1),v_3(m_0,m_{10},s_0,s_3,k_2))$ in each product bin $((m_{11},s_2),s_3)$ that satisfies the condition of mutual covering lemma, randomly and conditionally independent generate $2^{S_1}$ elements $x(m_0,m_{10},s_0,m_{11},s_2,s_3,s_1)$, $s_1 \in [1:2^{S_1}]$, each according to $P(x|v_2,v_3)$. This completes the codebook generation.

\subsubsection{Encoding} To send message triple $(m_0\in M_0,m_1\in M_1,m_2\in M_2)$, we express $m_1$ by the pair $(m_{10},m_{11})$ and $m_2$ by the quadruple $(s_0,s_1,s_2,s_3)$. The sender looks up a suitable pair $$(v_2(m_0,m_{10}, s_0,s_2,k_1),v_3(m_0,m_{10},s_0,s_3,k_2))$$ in the product bin $((m_{11},s_2),s_3)$ and eventually sends the codeword $x(m_0,m_{10},s_0,m_{11},s_2,s_3,s_1)$. 

\subsubsection{Decoding}
The failure of encoder contributes a constant additive term $f(\varepsilon)$ to the ultimate error probability. We go through the details of the decoding errors below assuming that the encoding was successful.. 
When the codeword $x(m_0,m_1,s_0,s_1,s_2,s_3)$, each receivers obtains its subsystem $\rho^{B_1B_2B_3}_{x(m_0,m_1,s_0,s_2,s_3,s_1)}$.
Subsequently, each receiver applies a POVM on its system to decode the intended messages. These POVMs are specifically designed to extract messages simultaneously. The high-level details are outlined below.

\medskip

\textbf{\emph{Receiver $B_1$:}} Let $(\hat{m}_0,\hat{m}_{10},\hat{m}_{11},\hat{s}_0,\hat{s}_2,\hat{s}_3,\hat{s}_1)$ denote the estimates made by the receiver $B_1$ for the messages. $B_1$ declares that $(\hat{m}_0,\hat{m}_{10},\hat{m}_{11},\hat{s}_0,\hat{s}_2,\hat{s}_3,\hat{s}_1)=(m_0,,m_{10},m_{11},s_0,s_1,s_2,s_3)$ is sent if it uniquely, but also up to $k_1$ and $k_2$, corresponds to the tuple
\begin{align*} 
\big(u(\hat{m}_0,\hat{s}_0,\hat{m}_{10}),v_2(\hat{m}_0,\hat{m}_{10},\hat{m}_{11},\hat{s}_0,\hat{s}_2,k_1),v_3(\hat{m}_0,\hat{m}_{10},\hat{s}_0,\hat{s}_3,k_2),x(\hat{m}_0,\hat{m}_{10},\hat{m}_{11},\hat{s}_0,\hat{s}_2,\hat{s}_3,\hat{s}_1)\big).
\end{align*}
Notice that $(m_{11},s_2)$ and $s_3$ are the product bin indices of $v_2(m_0,m_{10},s_0,s_2,k_1)$ and $v_3(m_0,m_{10},s_0,s_3,k_2)$, respectively. 
The five pinching maps as components of the decoding POVM for this receiver are defined similar to those in Sec. \ref{general-two-degraded-error} for receiver $B_1$. These maps are: $\mathcal{E}^{B_1}$ with respect to the spectral decomposition of $\rho^{B_1}$, $\mathcal{E}^{UB_1}_1$ with respect to the spectral decomposition of $\mathcal{E}^{B_1}(\rho^{XV_2V_3-U-B_1})$, $\mathcal{E}^{UV_2B_1}_2$ with respect to the spectral decomposition of $\mathcal{E}^{UB_1}_1(\rho_{UV_3X-V_2-B_1})$, $\mathcal{E}^{UV_3B_1}_{3}$ with respect to the spectral decomposition of $\mathcal{E}^{UB_1}_1(\rho_{UV_2X-V_3-B_1})$, and $\mathcal{E}^{UV_2V_3B_1}_4$ with respect to the spectral decomposition of $\mathcal{E}^{UB_1}_1(\rho_{UX-V_2V_3-B_1})$. 
We employ these pinching maps to define five projectors as follows:  
\begin{align*}
\Xi_{4}^{UV_2V_3XB_1}&\coloneqq\Big\{ \mathcal{E}^{UV_2V_3B_1}_4(\rho^{UV_2V_3XB_1}) \ge 2^{S_1}\mathcal{E}^{UB_1}_1(\rho^{UX- V_2V_3-B_1})\Big\},\\
\Xi_{3}^{UV_2V_3XB_1}&\coloneqq\Big\{ \mathcal{E}^{UV_3B_1}_3(\rho^{UV_2V_3XB_1}) \ge 2^{R_{11}+S_2+S_1}\mathcal{E}^{UB_1}_1(\rho^{UXV_2- V_3-B_1})\Big\},\\
\Xi_{2}^{UV_2V_3XB_1}&\coloneqq\Big\{ \mathcal{E}^{UV_2B_1}_2(\rho^{UV_2B_1}) \ge 2^{S_3+S_1}\mathcal{E}^{UB_1}(\rho^{UXV_3- V_2-B_1})\Big\},\\
\Xi_{1}^{UV_2V_3XB_1}&\coloneqq\Big\{\mathcal{E}^{UB_1}_1(\rho^{UV_2V_3X B_1}) \ge 2^{R_{11}+S_1+S_2+S_3}\mathcal{E}^{B_1}(\rho^{V_2 V_3 X- U-B_1})\Big\}, \\
\Xi_{0}^{UV_2V_3XB_1}&\coloneqq\Big\{ \mathcal{E}^{B_1}(\rho^{U V_2 V_3 XB_1}) \ge 2^{R_0+R_{1}+R_2}\rho^{UV_2V_3X}\otimes \rho^{B_1}\Big\}.
\end{align*}
The POVM is constructed in the same manner to that of the general three-receiver quantum broadcast channel with two-degraded message set, as depicted in Eq. \eqref{POVM-general-two}. The analysis of the average error probability also processes along the same lines; we use $\hat{M}_{ib_j}$ to denote the estimate made by receiver $B_j$ about the message $M_{i}$. We obtain:
\begin{align*}
\mathbb{E}_{\mathcal{C}_1}\big(P_{e,1}(\mathcal{C}_1)\big)&=\mathbb{E}_{\mathcal{C}_1}\big(\text{Pr}\{\hat{M}_{0b_1}\neq M_0,\hat{M}_{1b_1}\neq M_1,\hat{M}_{2b_1}\neq M_2|\mathcal{C}_1\}\big)\\
&\le \nu 2^{\alpha(R_0+R_1+R_2)+2-\alpha\tilde{I}_{\alpha}^{\uparrow}\left(U V_2 V_3 X; B_1\right)_{\rho^{U V_2 V_3 XB_1}}}\\
 &\hspace{1cm}+\nu_1 2^{\alpha(R_{11}+S_2+S_3+S_1)+2-\alpha\tilde{I}_{\alpha}^{\downarrow}\left(V_2 V_3 X; B_1|U\right)_{\rho^{U V_2 V_3 XB_1}|\rho^{U V_2 V_3 X}}}\\
 &\hspace{1cm}+\nu_2 2^{\alpha(S_3+S_1)+2-\alpha\tilde{I}_{\alpha}^{\downarrow}\left(U V_3 X; B_1|V_2\right)_{\rho^{U V_2 V_3 XB_1}|\rho^{U V_2 V_3 X}}}\\
 &\hspace{1cm}+\nu_3 2^{\alpha(R_{11}+S_2+S_1)+2-\alpha\tilde{I}_{\alpha}^{\downarrow}\left(U V_2 X; B_1|V_3\right)_{\rho^{U V_2 V_3 XB_1}|\rho^{U V_2 V_3 X}}}\\
 &\hspace{1cm}+\nu_4 2^{\alpha S_1+2-\alpha\tilde{I}_{\alpha}^{\downarrow}\left(U X; B_1|V_2,V_3\right)_{\rho^{U V_2 V_3 XB_1}|\rho^{U V_2 V_3 X}}}.
\end{align*}
where $\nu,\nu_1,\nu_2,\nu_3,\nu_4$ are the maximum number of distinct eigenvalues of the operators $\rho^{B_1},\{\mathcal{E}^{B_1}(\rho^{B_1}_{1|u})\}_u$, $\{\mathcal{E}^{B_1}_{1|u}(\rho^{B_1}_{v_2})\}_{u,v_2},\{\mathcal{E}^{B_1}_{1|u}(\rho^{B_1}_{v_3})\}_{u,v_3}$, and $\{\mathcal{E}^{B_1}_{1|u}(\rho^{B_1}_{v_2,v_3})\}_{u,v_2,v_3}$, respectively.

\bigskip

\textbf{\emph{Receiver $B_2$:}} Receiver $B_2$ extracts $M_0$ and $M_{10}$ by decoding $U$; it also inevitably decodes $M_{20}$ by decoding $U$, but may discard this information later. Unlike the case of two-degraded message set, here receiver $B_2$ does not use non-unique decoding because the variable $V_2$ includes the message $M_{11}$ that is intended for it. Similar to variable $U$, variable $V_2$ includes the message $M_{22}$ which is not intended for $B_2$, but will inevitably be decoded correctly. This information too can be discarded afterwards. Let $\mathcal{E}^{B_2}$ and $\mathcal{E}^{UB_2}_1$ be pinching maps with respect to the spectral decomposition of the operators $\rho^{B_2}$ and $\mathcal{E}^{B_2}(\rho^{V_2-U-B_2})$, respectively. We define the following projectors:
\begin{align*}
\Phi_1^{UV_2B_2} &\coloneqq \{\mathcal{E}^{UB_2}_{1}(\rho^{UV_2B_2})\ge2^{T_2}\mathcal{E}^{B_2}(\rho^{V_2-U-B_2})\}\\
    \Phi_0^{UV_2B_2} &\coloneqq
\{\mathcal{E}^{B_2}(\rho^{UV_2B_2})\ge2^{R_0+R_{10}+S_0+R_{11}+S_2+r_1}\rho^{UV_2}\otimes \rho^{B_2}\},
\end{align*}
and accordingly
\begin{align*}
   \Phi_{1\, u,v_2}^{B_2} &=\bra{u,v_2}\Phi_1^{UV_2B_2}\ket{u,v_2} \\
   \Phi_{0\, u,v_2}^{B_2} &=\bra{u,v_2}\Phi_0^{UV_2B_2}\ket{u,v_2},\\
   \Phi_{u,v_2}^{B_2} &=\Phi_{0\, u,v_2}^{B_2}\Phi_{1\, u,v_2}^{B_2}.
\end{align*}
The above projectors are not related to any code. However, as the decoder of $B_2$, we can build the following POVM for the constructed $\mathcal{C}_1$: 
\begin{align*}
&\Lambda_{m_0,m_{10},s_0,m_{11},s_2}^{B_2}(\mathcal{C}_1)\coloneqq\\
&\hspace{1.5cm}\left(\sum_{\substack{m_0',m_{10}',s_0',\\m_{11}',s_2'}} \Phi^{B_2}_{v_2(m_0',m_{10}',s_0',m_{11}',s_2')}\right)^{-\frac{1}{2}} \Phi^{B_2}_{v_2(m_0,m_{10},s_0,m_{11},s_2)}
\left(\sum_{\substack{m_0',m_{10}',s_0',\\m_{11}',s_2'}} \Phi^{B_2}_{v_2(m_0',m_{10}',s_0',m_{11}',s_2')}\right)^{-\frac{1}{2}},
\end{align*} 
where 
\begin{align}
\label{three-general.d}
\Phi^{B_2}_{v_2(m_0,m_{10},s_0,m_{11},s_2)}=\sum_{k_1'=1}^{2^{r_1}}\Phi^{B_2}_{v_2(m_0',m_{10}',s_0',m_{11}',s_2',k_1')}.
\end{align}
We can bound the average error probability as follows:
\begin{align*}
P_{e,2}(\mathcal{C}_1)&=\text{Pr}\{\hat{M}_{0b_2}\neq M_0, \hat{M}_{10b_2}\neq \hat{M}_{10},\hat{M}_{11b_2}\neq M_{11}\}\\
&\stackrel{\text{(a)}}{\le}
\text{Pr}\{\hat{M}_{0b_2}\neq M_0, \hat{M}_{10b_2}\neq \hat{M}_{10},\hat{M}_{11b_2}\neq M_{11},\hat{M}_{20b_2}\neq M_{20},\hat{M}_{22b_2}\neq M_{22}\}\\
&\coloneqq
 \frac{1}{2^{R_0+R_1+S_0+S_{2}}}\sum_{m_0,m_{10},s_0,m_{11},s_2} \tr (I-\Lambda_{m_0,m_{10},s_0,m_{11},s_2}^{B_2}(\mathcal{C}_1))
\rho^{B_2}_{v_2(m_0,m_{10},s_0,m_{11},s_2,k_1)}\\
&\stackrel{\text{(b)}}{\le}
\frac{1}{2^{R_0+R_1+S_0+S_{2}}}
 \sum_{m_0,m_{10},s_0,m_{11},s_2} \tr \Big(
 2(I- Q_{v_2(m_0,m_{10},s_0,m_{11},s_2)}^{B_2})
 \\&\hspace{4cm}+
 4\sum_{ \substack{(m_0',m_{10}',s_0',m_{11}',s_2')\neq\\ (m_0,m_{10},s_0,m_{11},s_2)}} Q_{v_2(m_0',m_{10}',s_0',m_{11}',s_2')}
\Big) \rho^{B_2}_{v_2(m_0,m_{10},s_0,m_{11},s_2)} \\
&\stackrel{\text{(c)}}{\le}
\frac{2}{2^{R_0+R_1+S_0+S_{2}}}
 \sum_{m_0,m_{10},s_0,m_{11},s_2} \tr
 (I- Q^{B_2}_{v_2(m_0,m_{10},s_0,m_{11},s_2,k_1)})\rho^{B_2}_{v_2(m_0,m_{10},s_0,m_{11},s_2,k_1)}\\
 &\hspace{2cm}+
\frac{4}{2^{R_0+R_1+S_0+S_{2}}}\sum_{ \substack{(m_{11}',s_2')\neq\\ (m_{11},s_2)}} \sum_{k_1'=1}^{2^{r_1}}Q^{B_2}_{v_2(m_0,m_{10},s_0,m_{11}',s_2',k_1')} \rho^{B_2}_{v_2(m_0,m_{10},s_0,m_{11},s_2,k_1)}\\
 &\hspace{2cm}+
 \frac{4}{2^{R_0+R_1+S_0+S_{2}}}\sum_{ \substack{(m_0',m_{10}',s_0')\neq\\ (m_0,m_{10},s_0),\\
 m_{11}',s_2'}} \sum_{k_1'=1}^{2^{r_1}}Q^{B_2}_{v_2(m_0',m_{10}',s_0',m_{11}',s_2',k_1')} \rho^{B_2}_{v_2(m_0,m_{10},s_0,m_{11},s_2,k_1)}\\
\end{align*}
where (a) follows by including the message $M_{10}$ (of rate $S_0$) which is encoded with the common message $M_0$ and will be inevitably decoded along the common message, even though it might be discarded later. (b) comes from Hayashi-Nagaoka inequality Lemma \ref{hayashi-nagaoka},
where in (c), for the first term, we use the definition in \eqref{three-general.d}
to discard all but one POVM element corresponding to $(m_0,m_{10},s_0,m_{11},s_2,k_1)$; and for the second term, we expanded the summand into two terms.
Next, we find the expectation of the average error probability over the choice of the random code:
\begin{align*}
    \mathbb{E}\{P_{e,2}(\mathcal{C}_1)\}& \leq \mu_12^{\alpha (R_{11}+S_2+r_1) +2 -\alpha \tilde{I}^{\downarrow}_{\alpha}(V_2;B_2|U)_{\rho^{UV_2B_2}|\rho^{UV_2}}} \\
    &+ \mu2^{\alpha(R_0+R_{10}+S_0+R_{11}+S_2+r_1)+2-\alpha\tilde{I}^{\uparrow}_{\alpha}(UV_2;B_2)_{\rho^{UV_2B_2}}},
\end{align*}
where $\mu$ and $\mu_1$ are the maximum number of distinct eigenvalues of the $\rho^{B_2}$ and $\mathcal{E}^{B_2}(\rho^{V_2-U-B_2})$, respectively.

\medskip

\textbf{\emph{Receiver $B_3$:}} Receiver $B_3$ is only interested in the common message $M_0$. However, it inevitably decodes $M_{10}$ and $M_{20}$ successfully. It decodes these messages non-uniquely via the information encoded into $V_3$. Let $\mathcal{E}^{B_3}$ be the pinching map with respect to the spectral decomposition of the operator $\rho^{B_3}$. We define the following projector:
\begin{align*}
    W^{B_3} \coloneqq \{\mathcal{E}^{UV_3B_3}\left(\rho^{UV_3B_3}\right)\ge 2^{R_0+R_{10}+S_0+S_3+r_2} \rho^{UV_3}\otimes \rho^{B_3}\},
\end{align*}
as well as $\bra{u,v_3}W^{UV_3B_3}\ket{u,v_3}=W^{B_3}_{u,v_3}$.
We now build the following POVM as the decoder of $B_3$: $\mathcal{C}_1$: 
\begin{align*}
\Lambda_{m_0,m_{10},s_0}^{B_3}(\mathcal{C}_1)\coloneqq
\left(\sum_{m_0',m_{10}',s_0'} Q^{B_3}_{v_3(m_0',m_{10}',s_0')}\right)^{-\frac{1}{2}} Q^{B_3}_{v_3(m_0,m_{10},s_0)}
\left(\sum_{m_0',m_{10}',s_0'} Q^{B_3}_{v_3(m_0',m_{10}',s_0')}\right)^{-\frac{1}{2}},
\end{align*} 
where 
\begin{align*}
Q^{B_2}_{v_3(m_0,m_{10},s_0)}=\sum_{s_3'=1}^{2^{S_2}}\sum_{k_2'=1}^{2^{r_2}}Q^{B_2}_{v_3(m_0,m_{10},s_0,s_3',k_2')}.
\end{align*}
We find an upper bound on the average error probability along the same lines as previous non-unique decoder (e.g. receivers $B_2$ and $B_3$ in the preceding section). Let $\hat{M}_{ib_3}$ denote the estimate made by receiver $B_3$ about the message $M_{i}$.We finally obtain:
 \begin{align*}
\E_{\mathcal{C}_1} \left(P_{e,3}(\mathcal{C}_1)\right) & =\mathbb{E}_{\mathcal{C}_1}\big(\text{Pr}\{\hat{M}_{0b_3}\neq M_0,\hat{M}_{10b_3}\neq M_{10},\hat{M}_{20b_3}\neq M_{20}|\mathcal{C}_1\}\big) \\
&\le \eta 2^{\alpha(R_0+R_{10}+S_0+S_3+r_2)+2} 
2^{-\alpha \tilde{I}_{\alpha}^{\uparrow}(UV_3;B_3)_{\rho^{UV_3B_3}}},
\end{align*}
where $\eta$ is the number of the distinct eigenvalues of the operator $\rho^{B_3}$.

\subsection{Asymptotic Analysis}

\begin{proofof}[Theorem \ref{general-three-degraded.t}]

Concerning the error probability of the encoder, leveraging Eq. \eqref{max-asymp}, we can conclude that if the condition $r_1 + r_2 \geq I(V_2;V_3\vert U)$ holds true, then a pair of codewords in each product bin can be identified, such that they mimic a joint distribution. This is further supported by the asymptotic mutual covering lemma \cite{1056302}, \cite[Lemma 8.1]{elgamal-kim}. As for the error probability of the decoders, 
the upper bounds on the expectation of the average error probability imply the existence of at least one good code $\mathcal{C}_1$ that satisfies these bounds. Consequently, we can assess the upper bounds in the following manner:
\begin{align*}
    P_{e,1}^{(n)}(\mathcal{C}_n)&\leq 
    4 (n+1)^{\alpha (d_{B_1}-1)} 2^{\alpha n(R_0 +R_1 + R_2)} 2^{-\alpha\tilde{I}_{1-\alpha}^{\uparrow}(U^n,V_2^n,V^n_3,X^n;B_1^n)_{(\rho^{UV_2V_3XB_1})^{\otimes n}}}\\
    &+4  (n+1)^{\alpha d_{U}(d_{B_1}+2)(d_{B_1}-1)/2} 2^{\alpha n(R_{11} + S_{2}+S_3+S_{1})} 2^{-\alpha \tilde{I}_{1-\alpha}^{\downarrow}\left(V_2^n,V_3^n,X^n;B_1^n|U^n\right)_{(\rho^{UV_2V_3XB_1})^{\otimes n}\vert(\rho^{UV_2V_3X})^{\otimes n}}}\\
    &+4  (n+1)^{\alpha d_{V_2}d_U(d_{B_1}+2)(d_{B_1}-1)/2} 2^{\alpha n(S_3+S_1)} 2^{-\alpha \tilde{I}_{1-\alpha}^{\downarrow}\left(U^n,V_3^n,X^n;B_1^n|V_2^n\right)_{(\rho^{UV_2V_3XB_1})^{\otimes n}\vert(\rho^{UV_2V_3X})^{\otimes n}}}\\
    &+4  (n+1)^{\alpha d_{V_3}d_U(d_{B_1}+2)(d_{B_1}-1)/2} 2^{\alpha n(R_{11}+S_2+S_1)} 2^{-\alpha \tilde{I}_{1-\alpha}^{\downarrow}\left(U^n,V_2^n,X^n;B_1^n|V_3^n\right)_{(\rho^{UV_3V_2XB_1})^{\otimes n}\vert(\rho^{UV_3V_2X})^{\otimes n}}}\\
    &+4  (n+1)^{\alpha d_{V_2}d_{V_3}d_U(d_{B_1}+2)(d_{B_1}-1)/2} 2^{\alpha nS_1} 2^{-\alpha \tilde{I}_{1-\alpha}^{\downarrow}\left(U^n,X^n;B_1^n|V_2^n,V_3^n\right)_{(\rho^{UV_3V_2XB_1})^{\otimes n}\vert(\rho^{UV_3V_2X})^{\otimes n}}}\\
  P_{e,2}^{(n)}(\mathcal{C}_n)&\leq 
4 (n+1)^{\alpha (d_{B_2}-1)} 2^{\alpha n(R_0+R_{10}+S_0+R_{11}+S_2+r_1)} 2^{-\alpha\tilde{I}_{1-\alpha}^{\uparrow}(U^n,V_2^n;B_2^n)_{(\rho^{UV_2B_2})^{\otimes n}}}\\
&+4 (n+1)^{\alpha d_U(d_{B_2}+2)(d_{B_2}-1)/2} 2^{\alpha n(R_{11}+S_2+r_1)} 2^{-\alpha\tilde{I}_{1-\alpha}^{\downarrow}(V_2^n;B_2^n\vert U^n)_{(\rho^{UV_2B_2})^{\otimes n}\vert(\rho^{UV_2})^{\otimes n} }},\\
  P_{e,3}^{(n)}(\mathcal{C}_n)&\leq 
4 (n+1)^{\alpha (d_{B_3}-1)} 2^{\alpha n(R_0+R_{10}+S_0+S_3+r_2)} 2^{-\alpha\tilde{I}_{1-\alpha}^{\uparrow}(U^n,V_3^n;B_3^n)_{(\rho^{UV_3B_3})^{\otimes n}}},
\end{align*}
where we have employed the polynomial bounds on the number of distinct eigenvalues from Proposition \ref{pinching-asymptotic}. We further obtain (notice that the number of distinct eigenvalues vanish because they are only polynomial and that R\'enyi relative entropy tends to von Nuemann relative entropy when its parameter tends to $1$):
\begin{align*}
  \lim_{n\to\infty }-\frac{1}{n}\log P_{e,1}^{(n)}(\mathcal{C}_n)\ge
  \min\bigg(&\alpha \big(\tilde{I}_{1-\alpha}^{\uparrow}(U,V_2,V_3,X;B_1)_{\rho^{UV_2V_3XB_1}}-(R_0+R_1+R_2)\big),\\
  &\alpha \big(\tilde{I}_{1-\alpha}^{\downarrow}\left(V_2,V_3,X;B_1|U\right)_{\rho^{UV_2V_3XB_1}|\rho^{UV_2V_3X}}-(R_{11}+S_{2}+S_3+S_{1})\big),\\
  &\alpha \big(\tilde{I}_{1-\alpha}^{\downarrow}\left(U,V_3,X;B_1|V_2\right)_{\rho^{UV_2V_3XB_1}|\rho^{UV_2V_3X}}-(S_3+S_1)\big)\\
  &\alpha \big(\tilde{I}_{1-\alpha}^{\downarrow}\left(U,V_2,X;B_1|V_3\right)_{\rho^{UV_2V_3XB_1}|\rho^{UV_2V_3X}}-(R_{11}+S_2+S_1)\big)\\
  &\alpha \big(\tilde{I}_{1-\alpha}^{\downarrow}\left(U,X;B_1|V_2,V_3\right)_{\rho^{UV_2V_3XB_1}|\rho^{UV_2V_3X}}-S_1\big)
  \bigg),\\
\lim_{n\to\infty }-\frac{1}{n}\log P_{e,2}^{(n)}(\mathcal{C}_n)\ge \min\bigg(
  &\alpha \big(\tilde{I}_{1-\alpha}^{\uparrow}(UV_2;B_2)_{\rho^{UV_2B_2}}-(R_0+R_{10}+S_0+R_{11}+S_2+r_1)\big),\\  
  &\alpha \big(\tilde{I}_{1-\alpha}^{\downarrow}(V_2;B_2\vert U)_{\rho^{UV_2B_2}\vert\rho^{UV_2}}-(R_{11}+S_2+r_1)\big) \bigg)\\
\lim_{n\to\infty }-\frac{1}{n}\log P_{e,3}^{(n)}(\mathcal{C}_n)\ge
  &\alpha \big(\tilde{I}_{1-\alpha}^{\downarrow}(UV_3;B_3)_{\rho^{UV_3B_3}\vert\rho^{UV_3}}-(R_0+R_{10}+S_0+S_3+r_2)\big).
\end{align*}
It now follows that as $n\to\infty$ and $\alpha\to 0$, there exists a sequence of codes $\mathcal{C}_n$ such that $P_{e,1}^{(n)}(\mathcal{C}_n),P_{e,2}^{(n)}(\mathcal{C}_n)$ and $P_{e,3}^{(n)}(\mathcal{C}_n)$ vanish, if
\begin{align*}
r_1 + r_2 &\geq I(V_2;V_3\vert U)\\
R_0 + R_1 + R_2&\leq I(X;B_1)_{\rho},\\
   R_{11}+ S_{3}+S_2+S_{1}&\leq I(X;B_1|U)_{\rho},\\
     S_{3}+S_1&\leq I(X;B_1|V_2)_{\rho},\\
     R_{11}+S_{2}+S_1&\leq I(X;B_1|V_3)_{\rho},\\
     S_1&\leq I(X;B_1\vert V_2,V_3)_{\rho}, \\
    R_0 + S_{0} + R_{10}+ S_2 +r_1 &\leq I(V_2;B_2)_{\rho},
    \\
    R_{11}+S_2+r_1 &\leq I(V_2;B_2\vert U)_{\rho},\\
    R_0 + S_{0} + R_{10}+S_3 +r_2 &\leq I(V_3;B_3)_{\rho},
\end{align*}
where the mutual information quantities are calculated for the state in the statement of the Theorem \ref{general-three-degraded.t}. Substituting $R_2=S_0 + S_2+ S_3+S_1, R_1=R_{10}+R{11}$ and using the Fourier–Motzkin elimination procedure, it follows from the union bound $P_{e}^{(n)}(\mathcal{C}_n)\le P_{e,1}^{(n)}(\mathcal{C}_n)+P_{e,2}^{(n)}(\mathcal{C}_n)+P_{e,3}^{(n)}(\mathcal{C}_n)$ that the probability of decoding error tends to zero as $n\to \infty$ if the inequalities in Theorem are satisfied. This completes the proof of Theorem \ref{general-three-degraded.t}. 
\end{proofof}

\medskip

\section{Discussion}
Increasing the number of receivers in broadcast channels naturally leads to new settings, and the emerging problems become at least as challenging with the addition of each receiver. Indeed, it is easy to see that straightforward superposition coding cannot establish the capacity of the three-receiver multilevel classical broadcast channel \cite{4557122,Nair-Elgamal}, and in case of quantum channels, it cannot reach the single-letter no-go region established in Theorem \ref{converse-1}. Specifically, we can immediately observe that straightforward superposition coding results in the following region for the three-receiver multilevel quantum broadcast channel $x\to\rho_{x}^{B_1B_2B_3}$ with two-degraded message set:
\begin{align*}
    R_0&\leq \min\{I(U;B_2),I(U;B_3)\},\\
    R_1&\leq I(X:B_1\vert U),
\end{align*}
for some state $\rho^{UXB_1B_2B_3}$.
For the case of classical outputs, it is shown by means of explicit examples that this bound is strictly smaller than capacity region achieved by non-unique decoding. It is an interesting question to find analogues examples to those in \cite[Sec. IV]{Nair-Elgamal}, showing that a similar result holds for quantum broadcast channels. 

  In superposition coding for the two-receiver broadcast channels, requiring the (potentially) stronger receiver to decode the cloud center (as a common message, even if that receiver is not interested in it), will not change the inner bound. However, this facilitates finding a converse. Indeed, the capacity of the two-receiver broadcast channel with two-degraded message sets is fully known in the classical case, whereas in the quantum case, it is understood up to conditioning on quantum systems. This is not the case for non-unique decoding \cite[Remark 8.2]{elgamal-kim}. In non-unique decoding corresponding to the three-receiver multilevel broadcast channel, if the receiver $B_3$ is required to decode part of the private message correctly, the inner bound will not change change, but only finding the converse will be difficult, meaning that changing the definition of the problem in the sense of the non-unique decoding does not change the converse, which is interesting. 

The transmission of classical information over quantum broadcast channels introduces new challenges that do not arise with classical broadcast channels. Transmitting quantum information over quantum broadcast channels amplifies these challenges compared to classical message transmission. Since quantum information cannot be copied, one needs to consider the effects of underlying environment, to which classical information may be leaked. In a future work, we are going to study the ultimate capabilities of quantum broadcast channels for quantum information communication. Another interesting problem is to show whether the single-letter no-go region in Theorem \ref{converse-1} corresponds to a strong converse, meaning that exceeding this region results in error probability approach $1$ exponentially fast. We anticipate this to be the case, and one might demonstrate it using the techniques outlined in \cite{hao-chung-strong}.

\medskip

\section{Appendix}
\label{appendix}

\subsection{Pinching maps}

We begin with the following lemma, known as the pinching inequality by Hayashi \cite{Hayashi_2002}, which presents a crucial inequality in quantum information and we frequently employ in our work. 
\begin{lemma}[Pinching inequality \cite{Hayashi_2002}, see also \cite{ogawa-hayashi}]\label{Hayashi-pinching}
Let $\rho$ and $\sigma$ be two quantum states and let $\mathcal{E}_\sigma$ be the pinching operation with respect to eigendecomposition of $\sigma$. We have
\begin{align*}
    \rho\leq\nu\mathcal{E}_{\sigma}(\rho),
\end{align*}
where $\nu$ is the number of distinct eigenvalues of the operator $\sigma$.
\end{lemma}

Next, we define nested pinching maps on a single Hilbert space. Consider the following tripartite states:
    \begin{align*}
         \rho^{UVB} = \sum_{u,v}p(u,v)\ketbra{u}\otimes \ketbra{v}\otimes\rho^{B}_{u,v},\\ 
   \rho^{V-U-B} = \sum_{u,v}p(u,v)\ketbra{u}\otimes \ketbra{v}\otimes\rho^{B}_{u},\\
   \rho^{U-V-B} = \sum_{u,v}p(u,v)\ketbra{u}\otimes \ketbra{v}\otimes\rho^{B}_{v},  
    \end{align*}
where $U$ and $V$ are classical, while $B$ is quantum. Let us denote the dimension of systems $U,V$ and $B$ by $d_U,d_V$, and $d_B$, respectively. We further define three pinching maps as follows: A pinching map with respect to the eigenspace decomposition of $\rho^{B}$, denoted by $\mathcal{E}^{B}$. If we apply this pinching map to the state $\rho^{V-U-B}$, we obtain
\begin{align*}
    \mathcal{E}^{B}(\rho^{V-U-B}) = \sum_{u,v}p(u,v)\ketbra{u}\otimes \ketbra{v}\otimes\mathcal{E}^{B}(\rho^{B}_{u}).
\end{align*}
We are now interested in a pinching map with respect to the spectral decomposition of $\mathcal{E}^{B}(\rho^{V-U-B})$. Notice that now we need to define a pinching map with respect to a collection of states $\{\mathcal{E}^{B}(\rho^{B}_u)\}_u$ instead of one state. In other words, we have $d_U$ different pinching maps rather than one. Hence, for every $u\in\mathcal{U}$, we define a pinching map with respect to the spectral decomposition of $\mathcal{E}^{B}(\rho^{B}_u)$ and denote it by $\mathcal{E}_{1|u}^{B}$. In our nested pinching technique, this pinching map is extended to a pinching map on $UB$ rather than a map on $B$ alone, even though its action on $U$ is trivial. Hence, we define the pinching map on $UB$ with respect to the spectral decomposition of $\mathcal{E}^{B}(\rho^{V-U-B})$ as follows:
\begin{align*}
\mathcal{E}_{1}^{UB}=\sum_{u}\ketbra{u}\otimes\mathcal{E}_{1|u}^{B}.
\end{align*}
Applied to an arbitrary state $\sigma^{UB}$, the pinching map $\mathcal{E}_{1|u}^{B}$ acts on $\sigma_u^{B}$. 
It is important to note that that the pinching map $\mathcal{E}_1^{UB}$ is an operator on composite systems $UB$. For pedagogical purposes, let us consider the action of this map on the \textit{tensor product state} $\sum_{u}q(u)\ketbra{u}\otimes\sigma^{B}$ (notice that this is not a cq-state, as $\sigma^B$ does not depend on $u$). We see that the resulting pinched state $\tilde{\sigma}^{UB}$ is generically a cq-state on $UB$ as follows:
\begin{align*}
   \tilde{\sigma}^{UB}=& \mathcal{E}_1^{UB}\Big(\sum_{u}q(u)\ketbra{u}\otimes\sigma^{B}\Big) \\
   =& 
\sum_{u}q(u)\ketbra{u}\otimes\mathcal{E}_{1|u}^{B}(\sigma^{B})\\
=& \sum_{u}q(u)\ketbra{u}\otimes\Tilde{\sigma}^{B}.
\end{align*}
In other words, pinching map turned the product states into a cq-state.
Let us consider pinching the state $\rho^{U-V-B}$ with the operator $\mathcal{E}_1^{UB}$. Note that in the latter state, the system $B$ is independent of $U$ conditioned on $V$. However, after the application of the pinching map, we obtain a generic cq-state on $UVB$ as follows:
\begin{align*}
    \tilde{\tau}^{UVB}=&\mathcal{E}_1^{UB}(\rho^{U-V-B}) \\
    =& \sum_{u,v}p(u,v)\ketbra{u}\otimes \ketbra{v}\otimes\mathcal{E}_{1|u}^{B}(\rho^{B}_{v})\\
    = & \sum_{u,v}p(u,v)\ketbra{u}\otimes \ketbra{v}\otimes\Tilde{\tau}_{u,v}^{B},
\end{align*}
where for every $u\in\mathcal{U}$ and $v\in\mathcal{V}$, we define $\Tilde{\tau}_{u,v}^{B}=\mathcal{E}_{1|u}^{B}(\rho^{B}_{v})$.
Now, we aim to define another pinching map with respect to the spectral decomposition of the state $\tilde{\tau}^{UVB}$. This pinching map is defined with respect to the spectral decomposition of the operators $\Tilde{\tau}_{u,v}^{B}=\mathcal{E}_{1|u}^{B}(\rho^{B}_{v})$ for every $(u,v)\in \mathcal{U}\times\mathcal{V}$. Let us denote the pinching map with respect to $\mathcal{E}_{1|u}^{B}(\rho^{B}_{v})$ by $\mathcal{E}_{2|u,v}^{B}$. This map is extended to a pinching map on $UVB$ as follows:
\begin{align*}
    \mathcal{E}_{2}^{UVB} = \sum_{u,v}\ketbra{u}\otimes \ketbra{v}\otimes \mathcal{E}_{2|u,v}^{B}.
\end{align*}
The number of distinct eigenvalues of an operator whose spectral decomposition gives rise to the pinching map plays a significant role in asymptotic analysis of our one-shot codes. In case of $\mathcal{E}^{B}$, this number corresponds to the distinct eigenvalues of the operator $\rho^{B}$. In case of $\mathcal{E}_1^{UB}$, we deal with operators $\{\mathcal{E}^{B}(\rho^{B}_u)\}_u$ and we are interested in the maximum number of distinct eigenvalues of these operators over $u\in\mathcal{U}$. Similarly, in case of $ \mathcal{E}_{2}^{UVB}$, we deal with operators $\{\mathcal{E}^{B}_{1|u}(\rho^{B}_v)\}_{u,v}$ and we are interested in the maximum number of distinct eigenvalues of these operators for $(u,v)\in(\mathcal{U},\mathcal{V})$. 

We can now proceed to construct analogous pinching maps when there are $n$ identical copies of the aforementioned operators.
But the important point is to show how the number of the distinct eigenvalues of the operators in $n$-fold space scale with $n$. More precisely, we want to find upper bounds on the maximum number of the distinct eigenvalues of the operators $(\rho^{B})^{\otimes n}$, $\{\mathcal{E}^{B^n}\left((\rho^{B}_u)^{\otimes n}\right)\}_{u^n}$, and $\{\mathcal{E}^{B^n}_{1|u^n}\left((\rho^{B}_v)^{\otimes n}\right)\}_{u^n,v^n}$, where $\mathcal{E}^{B^n}$ is the pinching map with respect to the spectral decomposition of $(\rho^{B})^{\otimes n}$, and  $\mathcal{E}^{B^n}_{1|u^n}$ is the pinching map with respect to the spectral decomposition of $\{\mathcal{E}^{B^n}\left((\rho^{B}_u)^{\otimes n}\right)\}_{u^n}$. We aim to demonstrate that the number of distinct eigenvalues in these scenarios is polynomial in the number of tensor product spaces. 

\medskip

The proof of these polynomial upper bounds on the number of distinct eigenvalues is based on the Weyl-Schur duality \cite[Sec. 4.4.1]{Hayashi2017-group1}, \cite[Sec. 6.2.1]{Hayashi2017-group2}. 
Schur-Weyl duality is a relation between the finite-dimensional irreducible representations of the general linear group $GL(\mathbb{C}^{d})$ and the symmetric group $S_n$.
Any finite-dimensional irreducible representation of the general linear group $GL(\mathbb{C}^{d})$
has the same form as
an irreducible representation of the special unitary group $SU_d$. 
It states that for any $n,d \in \mathbb{N}$, the $(n \times d)$-dimensional complex $n$-fold tensor product vector space
  $(\mathbb{C}^{d})^{\otimes n}$ decomposes into a direct sum of irreducible representations of the general linear group $GL(\mathbb{C}^{d})$ and the $n$-th symmetric group $S_n$.
These groups in fact identify each other via Young diagrams: When the vector of integers $\lambda = (\lambda_1, \lambda_2, \ldots, \lambda_d)$ satisfies the condition $\lambda_1 \geq \lambda_2 \geq \ldots \geq \lambda_d \geq 0$ and $\sum_{i=1}^{d} \lambda_i = n$, the vector $\lambda$ is called a Young diagram (frame) of size $n$ and depth $d$; the set of such vectors of size $n$ and depth (at most) $d$ is denoted as $Y_{d}^{n}$. It can be shown that 
\begin{align*}
    |Y_{d}^{n}|\le {n + d -1 \choose d-1} \le (n+1)^{d-1}.
\end{align*} 

\begin{lemma}[Schur-Weyl Duality]
\label{schur-weyl duality}
Using the Young diagram, the irreducible decomposition of $(\mathbb{C}^{d})^{\otimes n}$ can be described as follows:
\begin{align*} 
  (\mathbb{C}^{d})^{\otimes n} = \bigoplus_{\lambda \in Y_d^n} \mathcal{U}_{\lambda}(SU_d) \otimes \mathcal{U}_{\lambda}(S_n),
\end{align*}
where for an element $\lambda \in Y_{d}^{n}$,
$\mathcal{U}_{\lambda}(SU_d)$ is the irreducible representation space of the speacial unitary group $SU_d$ characterized by $\lambda$, and $\mathcal{U}_{\lambda}(S_n)$ is the irreducible representation space of the $n$-th symmetric group $S_n$ characterized also by $\lambda$.
\end{lemma}

\medskip

Schur-Weyl decomposition provides a profound insight into the structure of tensor product Hilbert spaces, revealing that any operator within this space can be expressed as irreducible representations of both $SU_d$ and $S_n$.  
Specifically, 
any tensor product state residing in $B^{\otimes n}$ takes the form:
\begin{align*}
(\rho^B)^{\otimes n} = 
\sum_{\lambda \in Y_{d_B}^{n}}
p_\lambda \rho_\lambda\otimes \rho_{\text{mix}, U_\lambda(S_{n})}
\end{align*} 
where $\rho_\lambda$ is a state on the space $\mathcal{U}_{\lambda}(SU_d)$
and $\rho_{\text{mix}, U_\lambda(S_{n})}$ is the completely mixed state on
$ \mathcal{U}_{\lambda}(S_n)$, and 
$p_\lambda$ 
is a probability distribution on $Y_{d_B}^{n}$.
In other words,
$(\rho^B)^{\otimes n}$ can be written as 
\begin{align}
\label{decomposition}
(\rho^B)^{\otimes n} = 
\sum_{\lambda \in Y_{d_B}^{n}}
X_\lambda\otimes I_{\mathcal{U}_\lambda(S_{n})},
\end{align} 
where $X_\lambda$ is 
an operator on $\mathcal{U}_\lambda(SU_{d_B})$
and 
$I_{\mathcal{U}_\lambda(S_{n})}$ is the identity operator on $\mathcal{U}_\lambda(S_{n})$.
We now present the polynomial upper bounds in the following lemma.
\begin{proposition}\label{pinching-asymptotic-appendix}
Consider the cq-states $\rho^{UVXB},\rho^{XV-U-B},\rho^{XU-V-B}$, and $\rho^{U-VX-B}$, classical on $U,V,X$ and quantum on $B$. Let $\nu,\nu_1,\nu_2$ and $\nu_3$ be defined as follows:
\begin{align*}
 \nu & = \left\{\text{distinct eigenvalues of}\hspace{0.1cm} (\rho^B)^{\otimes n}\right\},\\
    \nu_1 & = \max_{u^n\in\mathcal{U}^n} {\left\{\text{distinct eigenvalues of}\hspace{0.1cm} \mathcal{E}^{B^n}\big(\rho^{B^n}_{u^n}\big)\right\}},\\
     \nu_2 & = \max_{u^n\in\mathcal{U}^n,v^n\in\mathcal{V}^n} {\left\{\text{distinct eigenvalues of}\hspace{0.1cm} \mathcal{E}^{B^n}_{1|u^n}\big(\rho^{B^n}_{v^n}\big)\right\}},\\
     \nu_3 & = \max_{\substack{v^n\in\mathcal{V}^n,x^n\in\mathcal{X}^n\\u^n\in\mathcal{U}^n}} {\left\{\text{distinct eigenvalues of}\hspace{0.1cm} \mathcal{E}^{B^n}_{1|u^n}\big(\rho^{B^n}_{v^n,x^n}\big)\right\}}
\end{align*}
where $\rho^{B^n}_{u^n}= \rho_{u_1}^{B}\otimes \rho_{u_2}^{B}\otimes\cdots\otimes\rho_{u_n}^{B}$, $\rho^{B^n}_{v^n}= \rho_{v_1}^{B}\otimes \rho_{v_2}^{B}\otimes\cdots\otimes\rho_{v_n}^{B}$, $\rho^{B^n}_{v^n,x^n}=\rho_{u_1,v_1}^B\otimes\cdots,\otimes\rho_{u_n,v_n}^B$, and $\mathcal{E}^{B^n}$ and $\mathcal{E}^{B^n}_{1|u^n}$ are pinching maps corresponding to the spectral decomposition of the operators $(\rho^B)^{\otimes n}$ and $\mathcal{E}^{B^n}\left(\rho^{B^n}_{u^n}\right)$, respectively. We show that the number of distinct eigenvalues are polynomial in the number of tensor product Hilbert spaces. More precisely, we obtain:
\begin{align*}
\nu & \le (n+1)^{d_B-1} \\
\nu_1 & \le 
(n+1)^{d_U(d_B+2)(d_B-1)/2} \\
\nu_2 & \le (n+1)^{d_V d_U(d_B+2)(d_B-1)/2}\\
\nu_3 & \le (n+1)^{d_Xd_V d_U(d_B+2)(d_B-1)/2} .
\end{align*}
\end{proposition}

\begin{proof}
The values of $\nu$ and $\nu_1$ are determined in \cite[Lemma 3]{AHW}. However, the evaluations of $\nu_2$ and $\nu_3$ are not provided in the previous work. To find the upper bounds $\nu_2$ and $\nu_3$, we build on the techniques used to find $\nu$ and $\nu_1$. So we begin with the latter numbers.
The number of eigenvalues of $(\rho^{B})^{\otimes n}$
is upper bounded by $|Y_{d_B}^n|$, the cardinality of the Young diagram of size $n$ and depth at most $d_B$. The latter can be bounded by the number of combinations of $d_B-1$ objects from
$n+d_B-1$ elements, which
itself is upper bounded by $(n+1)^{d_B-1}$. This follows from counting the number of different type classes in information theory \cite{Csiszar_Körner_2011} (for further details between Schur duality and type classes see \cite[Sec. 3]{Hayashi2009-gj}, \cite[Sec. 6.2.1]{Hayashi2017-group2} and \cite[Chapter 6]{harrow2005applications}). Next, our focus turns to $\nu_1$. 
Suppose
$u^n$ takes the following form:
\begin{align}
\label{sample-u}
   u^n = (\underbrace{u_1, \ldots, u_1}_{n_1}, 
\underbrace{u_2, \ldots, u_2}_{n_2}, 
\ldots, 
\underbrace{u_{d_U}, \ldots, u_{d_U}}_{n_{d_U}}).
\end{align}
Accordingly, we have $\rho^{B^n}_{u^n}=(\rho^B_{u_1})^{\otimes n_1}\otimes \cdots \otimes 
(\rho^B_{u_{d_U}})^{\otimes n_{d_U}}$.
For $j=1, \ldots, d_U$, from Schur-Weyl duality given in Lemma \ref{schur-weyl duality},
the tensor product Hilbert space $B^{\otimes n_j}$
is decomposed as
\begin{align*}
    \bigoplus_{\lambda \in Y_{d_B}^{n_j}}
{\cal U}_\lambda (SU_{d_B}) \otimes \mathcal{U}_\lambda(S_{n_j}),
\end{align*}
where
$Y_d^n$ is the set of Young indexes of size $n$ and depth at most $d_B$.
Let ${\cal F}_j$ 
be the pinching map on $B^{\otimes n_j}$
with respect to the above Schur-Weyl decomposition. Hence, $\mathcal{F}_j$ is described by the following projectors:
\begin{align}
\label{pinching-n_j}
    \{
\mathbbm{1}_{{\cal U}_\lambda (SU_{d_B}) }\otimes \mathbbm{1}_{\mathcal{U}_\lambda(S_{n_j})}
\}_{\lambda \in Y_{d_B}^{n_j}},
\end{align}
where $\mathbbm{1}_{{\cal U}_\lambda (SU_{d_B}) }$ and $\mathbbm{1}_{\mathcal{U}_\lambda(S_{n_j})}$ are the identity operators on the subspaces corresponding to ${\cal U}_\lambda (SU_{d_B})$ and $\mathcal{U}_\lambda(S_{n_j})$, respectively.
Notice that this decomposition does not depend on $(\rho^{B}_{u_j})^{\otimes n_j}$, nor on the eigenvectors of $\rho^B_{u_j}$ (recall that Schur-Weyl duality is a statement generally on the tensor product space). Therefore, the pinching map $\mathcal{F}_j$ defined by operators in Eq. \eqref{pinching-n_j}, commutes with all tensor product operators residing in $B^{\otimes n_j}$, such as $(\rho^B_{u_j})^{\otimes n_j}$ or $(\rho^B)^{\otimes n_j}$.
Therefore, the pinching map $\mathcal{F}^{B^n} = {\cal F}_1 \otimes \cdots \otimes {\cal F}_{d_U}$ on $B^{\otimes n}$ has no effect on $(\rho^B)^{\otimes n}$. To put it another way, if $\{F_i\}_i$ and $\{E_j\}_j$ are the pinching projectors with respect to $\mathcal{F}^{B^n}$ and $\mathcal{E}^{B^n}$, respectively, $F_iE_j=E_jF_i$ for every $i$ and $j$.

In order to further understand the the pinching map with respect to the spectral decomposition of the operator $\mathcal{E}^{B^n}(\rho_{u^n}^{B^n})$, we rewrite this operator as follows:
\begin{align*}
    \mathcal{E}^{B^n}(\rho_{u^n}^{B^n})=
    \mathcal{E}^{B^{n_1}}\left((\rho_{u_1}^B)^{\otimes n_1}\right)\otimes\mathcal{E}^{B^{n_2}}\left((\rho_{u_2}^B)^{\otimes n_2}\right)\otimes\cdots\otimes\mathcal{E}^{B^{n_{d_U}}}\left((\rho_{u_{d_U}}^B)^{\otimes n_{d_U}}\right).
\end{align*}
For any $j=1,2,\ldots,d_U$, based on Eq. \eqref{decomposition}, we can express
\begin{align*}
    (\rho_{u_j}^B)^{\otimes n_j}&=\sum_{\lambda \in Y_{d_B}^{n_j}}
X_\lambda\otimes \mathbbm{1}_{U_\lambda(S_{n_j})},\\
(\rho^B)^{\otimes n_j}&=\sum_{\lambda \in Y_{d_B}^{n_j}}
Z_\lambda\otimes \mathbbm{1}_{U_\lambda(S_{n_j})},
\end{align*}
where $X_\lambda$ and $Z_\lambda$ are operators acting on $\mathcal{U}_\lambda(SU_{d_B})$. Moreover, since the pinching map $\mathcal{F}_j$ commutes with both operators, we have $\mathcal{E}^{B^{n_j}}\bigl(\mathcal{F}_j\big((\rho_{u_j}^B)^{\otimes n_j}\big)\bigl)=
\mathcal{F}_j\bigl(\mathcal{E}^{B^{n_j}}\big((\rho_{u_j}^B)^{\otimes n_j}\big)\bigl)$. This means that in fact
\begin{align*}
  \mathcal{E}^{B^{n_j}}\big((\rho_{u_j}^B)^{\otimes n_j}\big)=\sum_{\lambda \in Y_{d_B}^{n_j}}
T_\lambda\otimes \mathbbm{1}_{U_\lambda(S_{n_j})}
\end{align*}
for an operator $T_\lambda$ on $\mathcal{U}_\lambda(SU_{d_B})$. Now considering the tensor product operators, since
\begin{align*}
\mathcal{E}^{B^n}\big(\rho_{u^n}^{B^{\otimes n}}\big)
=
\mathcal{E}^{B^n} 
\bigl(({\cal F}_1 \otimes \cdots \otimes {\cal F}_{d_U})
\big(\rho_{u^n}^{B^{\otimes n}}\big)\bigl)
= 
({\cal F}_1 \otimes \cdots \otimes {\cal F}_{d_U})
\big(\mathcal{E}^{B^n}\big(\rho_{u^n}^{B^{\otimes n}}\big)\big),
\end{align*}
the operator $\mathcal{E}^{B^n}\big(\rho_{u^n}^{B^{\otimes n}}\big)
$ takes the form
\begin{align*}
\mathcal{E}^{B^n}\big(\rho_{u^n}^{B^{\otimes n}}\big)=
    \sum_{\lambda_1 \in Y_{d_B}^{n_1}, \ldots,
\lambda_{d_U} \in Y_{d_B}^{n_{d_U}}}
Q_{\lambda_1, \ldots, \lambda_{d_U}}
\otimes 
\mathbbm{1}_{{\cal U}_{\lambda_1}(S_{n_1})} \otimes \cdots
\otimes \mathbbm{1}_{{\cal U}_{\lambda_{d_U}}(S_{n_{d_U}})}
\end{align*}
with an operator 
$Q_{\lambda_1, \ldots, \lambda_{d_U}}$
on 
${\cal U}_{\lambda_1}(SU_{d_B})\otimes \cdots \otimes
{\cal U}_{\lambda_{d_U}}(SU_{d_B})$. We can now find an upper bound on the number of the eigenvalues of the state. Note that the identity operators $\{\mathbbm{1}_{{\cal U}_{\lambda_{i}}(S_{n_{i}})}\}_i$ do not increase the number of eigenvalues, therefore we focus on the operators $Q_{\lambda_1, \ldots, \lambda_{d_U}}$. First notice that the number of these operators in the summation equals the number of choices of 
$(\lambda_1, \ldots, \lambda_{d_U})$, which
is upper bounded by 
$(n+1)^{d_U (d_B-1)}$.
The dimension of 
${\cal U}_{\lambda_j}(SU_{d_B})$ is upper bounded by
$(n_j+1)^{d_B(d_B-1)/2}$ \cite[(6.16)]{Hayashi2017-group2}.
Since $(n_j+1)^{d_B(d_B-1)/2}
\le (n+1)^{d_B(d_B-1)/2}$,
the dimension of 
${\cal U}_{\lambda_1}(SU)\otimes \cdots \otimes
{\cal U}_{\lambda_{d_U}}(SU)$
is upper bounded by 
$(n+1)^{d_U d_B(d_B-1)/2}$.
Bringing together these findings,
the number of eigenvalues of $\mathcal{E}^{B^n}\big(\rho_{u^n}^{B^{\otimes n}}\big)$ is upper bound by
$(n+1)^{d_U d_B(d_B-1)/2} \cdot
(n+1)^{d_U (d_B-1)}
=(n+1)^{d_U(d_B+2)(d_B-1)/2} $.
A general sequence
$\tilde{u}^n$ can be written as 
the application of a certain permutation 
$g \in S_n$
to the above-mentioned sequence in Eq. \eqref{sample-u}.
In this case, the finding for that particular choice of $u^n$ also holds for a general $\tilde{u}^n$ after application of $g \in S_n$. Therefore, the pinching $\mathcal{F}^{B^n}$ becomes 
$U_g 
({\cal F}_1 \otimes \cdots \otimes {\cal F}_{d_U})
U_g^\dagger$.
That is, this pinching
depends only on an element of
$d_U$-nomial combinatorics
among $n$ elements.

We are now in a position to find an upper bound on $\nu_2$.
We again assume that $u^n$ has the particular form given by Eq. \eqref{sample-u}, which we repeat here for clarity: 
\begin{align}
\label{sample-u-repeat}
       u^n = (\underbrace{u_1, \ldots, u_1}_{n_1}, 
\underbrace{u_2, \ldots, u_2}_{n_2}, 
\ldots, 
\underbrace{u_{d_U}, \ldots, u_{d_U}}_{n_{d_U}}).
\end{align}
We next make an assumption about the structure of the sequence $v^n$. We suppose the first $n_1$ elements of 
$v^n$
are arranged as follows:
\begin{align*}
   (\underbrace{v_1, \ldots, v_1}_{n_{11}}, 
\underbrace{v_2, \ldots, v_2}_{n_{12}}, 
\ldots, 
\underbrace{v_{d_V}, \ldots, v_{d_V}}_{n_{1d_V}}),
\end{align*}
such that $\sum_{i=1}^{d_V}n_{1i}=n_1$.
We assume a similar structure for the subsequent $n_2$ elements, followed by the next $n_3$ elements, and so forth, up to the final $n_{d_U}$ elements.
So the sequence $v^n$ is written as follows: 
\begin{align}
\label{dissect}
    \big(
\overbrace{\underbrace{v_1, \ldots, v_1}_{n_{11}}, 
\underbrace{v_2, \ldots, v_2}_{n_{12}}, 
\ldots, 
\underbrace{v_{d_V}, \ldots, v_{d_V}}_{n_{1d_V}}}^{\text{first $n_1$ elements}},
\overbrace{\underbrace{v_1, \ldots, v_1}_{n_{21}}, 
\ldots, \underbrace{v_{d_V}, \ldots, v_{d_V}}_{n_{2d_V}}}^{\text{subsequent $n_2$ elements} },\ldots,\overbrace{
\underbrace{v_1, \ldots, v_1}_{n_{d_U1}}, 
\ldots, 
\underbrace{v_{d_U}, \ldots, v_{d_V}}_{n_{d_Ud_V}}}^{\text{final $n_{d_U}$ elements}}\big).
\end{align}
Consider the first $n_{11}$ elements of both $u^n$ and $v^n$.
Let $\mathcal{R}_{11}$ be the pinching map with respect to the Schur-Weyl decomposition of the tensor product space $B^{\otimes n_{11}}$. Analogous to Eq. \eqref{pinching-n_j}, $\mathcal{R}_{11}$ is described by the following projectors: 
\begin{align*}
    \{
\mathbbm{1}_{{\cal U}_\lambda (SU_{d_B}) }\otimes \mathbbm{1}_{\mathcal{U}_\lambda(S_{n_{11}})}
\}_{\lambda \in Y_{d_B}^{n_{11}}},
\end{align*}
where $\mathbbm{1}_{{\cal U}_\lambda (SU_{d_B}) }$ and $\mathbbm{1}_{\mathcal{U}_\lambda(S_{n_{11}})}$ are the identity operators on the subspaces corresponding to ${\cal U}_\lambda (SU_{d_B})$ and $\mathcal{U}_\lambda(S_{n_{11}})$, respectively.
If we denote the pinching map with respect to the spectral decomposition of $\mathcal{E}^{B^{n_{11}}}\big(\rho_{u_1^{n_{11}}}^{B^{{n_{11}}}}\big)$ by $\mathcal{E}^{B^{n_{11}}}_{1|u_{1}^{n_{11}}}$ (which is the notation we have introduced before), from the structure of the operators $\mathcal{E}^{B^{n_{11}}}_{1|u_{1}^{n_{11}}}\equiv\mathcal{E}^{B^{n_{11}}}\big(\rho_{u_1^{n_{11}}}^{B^{{n_{11}}}}\big)$ and $\rho_{{v_1}^{n_{11}}}^{n_{11}}$, we know that
\begin{align*}
\mathcal{E}^{B^{n_{11}}}_{1|u_{1}^{n_{11}}}\circ \mathcal{R}_{11} (\rho_{{v_1}^{n_{11}}}^{n_{11}}) 
=
\mathcal{R}_{11}\circ \mathcal{E}^{B^{n_{11}}}_{1|u_{1}^{n_{11}}}(\rho_{{v_1}^{n_{11}}}^{n_{11}}).
\end{align*}
We proceed with this procedure for each subset of elements within $v^n$, as indicated by the underbraces in Eq. \eqref{dissect}. 
Let $\mathcal{R}_{ij}$ denote the pinching map according to the Schur-Weyl decomposition of the tensor product Hilbert space $B^{\otimes n_{ij}}$, so a pinching map on $B^{\otimes n}$ is defined as ${\cal R}_{11} \otimes {\cal R}_{12} \otimes 
\cdots \otimes {\cal R}_{d_Ud_V}
$. We have:
\begin{align*}
\mathcal{E}^{B^n}_{1|u^n}\big(\rho_{v^n}^{B^{ n}}\big)
&=
\mathcal{E}^{B^n}_{1|u^n}\circ
(
{\cal R}_{11} \otimes {\cal R}_{12} \otimes 
\cdots \otimes {\cal R}_{d_Ud_V}
)\big(\rho_{v^n}^{B^{ n}}\big)\\
&=
(
{\cal R}_{11} \otimes {\cal R}_{12} \otimes 
\cdots \otimes {\cal R}_{d_Ud_V}
)\circ\mathcal{E}^{B^n}_{1|u^n}
\big(\rho_{v^n}^{B^{ n}}\big)
\end{align*}
We can now apply the same evaluation as for $\nu_1$.
We find that 
the number of eigenvalues of 
$\mathcal{E}^{B^n}_{1|u^n}\big(\rho_{v^n}^{B^{ n}}\big)$
 is upper bound by
$(n+1)^{d_U d_V(d_B+2)(d_B-1)/2} $.

Next, we obtain an upper bound on $\nu_3$, the maximum number of distinct eigenvalues of the operator $\mathcal{E}^{B^n}_{1|u^n}\big(\rho^{B^n}_{v^n,x^n}\big)$. We again assume that the sequences $u^n$ and $v^n$ have the structure given by Eqs. \eqref{sample-u-repeat} and \eqref{dissect}. Now, we assume that the first $n_{11}$ the sequence $x^n$ have the following structure:
\begin{align*}
   (\underbrace{x_1, \ldots, x_1}_{n_{111}}, 
\underbrace{x_2, \ldots, x_2}_{n_{112}}, 
\ldots, 
\underbrace{x_{d_X}, \ldots, x_{d_X}}_{n_{11d_X}}),
\end{align*}
such that $\sum_{i=1}^{d_X}n_{11i}=n_{11}$. Therefore, similar to Eq. \eqref{dissect}, we can write $x^n$ as follows:

\begin{align*}
& \overbrace{
\overbrace{\underbrace{x_1, \ldots, x_1}_{n_{111}}, 
\underbrace{x_2, \ldots, x_2}_{n_{112}}, 
\ldots, 
\underbrace{x_{d_X}, \ldots, x_{d_X}}_{n_{11d_X}}}^{\text{First $n_{11}$ elements of $v^n$}},
\overbrace{\underbrace{x_1, \ldots, x_1}_{n_{121}}, 
\ldots, \underbrace{x_{d_X}, \ldots, x_{d_X}}_{n_{12d_X}}}^{\text{subsequent $n_{12}$ elements of $v^n$} },\ldots,
\overbrace{
\underbrace{x_1, \ldots, x_1}_{n_{1d_V1}}, 
\ldots, 
\underbrace{x_{d_X}, \ldots, x_{d_X}}_{n_{1d_Vd_X}}}^{\text{final $n_{1d_V}$ elements of the first $n_{11}$}}}^{\text{First $n_1$ elements of $u^n$}}
\end{align*}
 where $\sum_{j=1}^{d_X}n_{11j}=n_{11}$, $\sum_{i=1}^{d_V}\sum_{j=1}^{d_X}n_{1ij}=n_{1}$, and $\sum_{j=1}^{d_U}n_{j}=n$.
Now, considering the first $n_{111}$ elements of $x^n$, we need to find an upper bound on the distinct number of eigenvalues of the operator $\mathcal{E}_{1|u^{n_{111}}}\left(\rho^{B^{n_{111}}}_{v_1,x_{1}^{n_{111}}}\right)$. By using the lessons we learned in the previous evaluations, we obtain
\begin{align*}
    (n_{111}+1)^{d_B-1}\times (n_{111}+1)^{d_B(d_B-1)/2}\le(n+1)^{d_B-1}\times (n+1)^{d_B(d_B-1)/2}=(n+1)^{(d_B-1)(d_B+2)/2}.
\end{align*}
This will be repeated for the next $(d_X-1)$ sub-sequences corresponding to $(x_2,\ldots,x_2),\ldots,(x_{d_X},\ldots,x_{d_X})$, so that for the first $n_{11}$ elements of $v^n$, the number of distinct eigenvalues are upper bounded by $(n+1)^{d_X(d_B-1)(d_B+2)/2}$. This will again be repeated for the subsequent $n_{12}$ elements of $v^n$, and so forth, yielding $(n+1)^{d_Vd_X(d_B-1)(d_B+2)/2}$. So far we have considered the first $n_1$ elements of $u^n$; eventually by taking into account the rest of $d_U-1$ sub-sequences of $u^n$, we obtain the upper bound $(n+1)^{d_Ud_Vd_X(d_B-1)(d_B+2)/2}$. This concludes the proof.
\end{proof}

\bigskip
\subsection{Binary hypothesis-testing}
We next study a binary hypothesis-testing problem between a null state $\rho$ and an alternative state $\sigma$, where the corresponding two-outcome POVM is constructed using pinching maps. For two Hermitian operators $T$ and $O$, consider the operator $(T-O)$ with the following eigenspace decomposition
\begin{align*}
    T - O=\sum_{i}g_i P_i,
\end{align*} 
where $P_i$ is the projector onto the eigenspace corresponding to the eigenvalue $g_i$.
We define the operator $\{T\ge O\}$ as 
\begin{align*}
    \{T\ge O\}\coloneqq\sum_{i:g_i\ge 0} P_i.
\end{align*}
We use this definition to define the projector $\Pi\coloneqq\{\mathcal{E}_{\sigma}(\rho)\geq M\sigma\}$, where $\mathcal{E}_{\sigma}$ is the pinching map with respect to the spectral decomposition of the state $\sigma$ and $M$ is a positive number. Projectors of this kind are naturally encountered in the derivation of coding theorems. In particular, when we employ Hayashi-Nagaoka inequality Lemma \ref{hayashi-nagaoka} in the coding theorem, two terms emerge. The first term denotes the probability of making an error in detecting the true state (type I error), while the second term encompasses all other codewords, represented collectively by $\sigma$ due to random coding, are erroneously identified as true (type II error). We now find an upper bound on combination of these errors using the aforementioned projector.

\begin{lemma}
\label{hp-testing}
Consider the binary POVM $
\{\Pi, I-\Pi\}$ constructed from the projector $\Pi\coloneqq\{\mathcal{E}_{\sigma}(\rho)\geq M\sigma\}$, where $\mathcal{E}_{\sigma}$ is the pinching map with respect to the spectral decomposition of the state $\sigma$ and $M$ is a positive number. The following holds for $\alpha\in(0,1)$:
 \begin{align*}
     \tr(I-\Pi)\rho+M\tr\Pi\sigma&\leq M^{\alpha}2^{-\alpha D_{1-\alpha}(\mathcal{E}_{\sigma}(\rho)\|\sigma)}.
 \end{align*}
\end{lemma}
\begin{proof}
We find an upper bound on each of the terms on the left-hand-side as follows: Note that we adopt the notation $\mathcal{E}_{\sigma}(\rho)^{x}=\left(\mathcal{E}_{\sigma}(\rho)\right)^{x}$.
For the first term we obtain: 
\begin{align*}
    \tr(I-\Pi)\rho&\stackrel{\text(a)}{=}\tr\{\big(\mathcal{E}_\sigma(I-\Pi)\big)\rho\}\\
    &\stackrel{\text(b)}{=}\tr(I-\Pi)\mathcal{E}_\sigma(\rho)\\
    &=\tr(I-\Pi)\mathcal{E}_\sigma(\rho)^{1-\alpha}\mathcal{E}_\sigma(\rho)^{\alpha}\\
    &\stackrel{\text{(c)}}{\le} \tr(I-\Pi)\mathcal{E}_\sigma(\rho)^{1-\alpha}(M\sigma)^{\alpha},
\end{align*}
where (a) follows from the fact the the projector $\Pi$ commutes with the operator $\sigma$, thus the pinching map $\mathcal{E}_\sigma$ has no effect on the projector, (b) uses the cyclicity of the pinching map inside the trace, and (c) follows from the definition of the pinching projector as well as the operator monotonicity of the function $x^{\alpha}$ for $\alpha\in(0,1)$. For the second term, we obtain:
\begin{align*}
     \tr\Pi\sigma&=\tr\Pi\sigma^{1-\alpha}\sigma^\alpha\\
     &\le\tr\Pi(M^{-1}\mathcal{E}_\sigma(\rho))^{1-\alpha}\sigma^\alpha,
\end{align*}
where the inequality follows from the definition of the projector and the monotonicity of the function $x^{1-\alpha}$. Summing up, we obtain: 
\begin{align*}
    \tr(I-\Pi)\rho+M\tr\Pi\sigma\leq M^\alpha\tr\mathcal{E}_\sigma(\rho)^{1-\alpha}\sigma^\alpha=M^{\alpha}2^{-\alpha D_{1-\alpha}(\mathcal{E}_\sigma(\rho)\|\sigma)}.
\end{align*}
\end{proof}

\medskip

In the previous lemma, the error probability is upper bounded in terms of the Petz's R\'enyi entropy where the first argument involves a pinching map. The following lemma demonstrates that removing the pinching from the first argument is possible, but the Petz's R\'enyi entropy transforms to the Sandwich R\'enyi entropy.

\begin{lemma}
\label{petz-sandwich}
    Let $\mathcal{E}_{\sigma}$ be the pinching map with respect to the spectral decomposition of the operator $\sigma$. The following inequaliy holds:
    \begin{align*}
        \widetilde{D}_{1-\alpha} (\rho\| \sigma) \le \log\nu +  D_{1-\alpha}(\mathcal{E}_{\sigma}(\rho)\|\sigma),
    \end{align*}
    where $\nu$ is the number of the distinct eigenvalues of the operator $\sigma$. It is useful to write this inequality equivalently as follows:
    \begin{align*}
        2^{-\alpha D_{1-\alpha}(\mathcal{E}_{\sigma}(\rho)\|\sigma)}\leq
        \nu^{\alpha} 2^{-\alpha \widetilde{D}_{1-\alpha} (\rho\| \sigma) }.
    \end{align*}
\end{lemma}
\begin{proof}
From the pinching inequality in Lemma \ref{Hayashi-pinching}, 
$\rho\leq\nu \mathcal{E}_{\sigma}(\rho)$,
we have
\begin{align*}
\sigma^{\frac{\alpha}{2(1-\alpha)}}
\rho
\sigma^{\frac{\alpha}{2(1-\alpha)}}\leq\nu\sigma^{\frac{\alpha}{2(1-\alpha)}}
\mathcal{E}_{\sigma}(\rho)
\sigma^{\frac{\alpha}{2(1-\alpha)}}.
\end{align*}
Since $x \mapsto -x^{-\alpha}$ is operator monotone for $\alpha \in (0,1)$, we have
\begin{align}
\Big(\nu\sigma^{\frac{\alpha}{2(1-\alpha)}}
\mathcal{E}_\sigma(\rho)
\sigma^{\frac{\alpha}{2(1-\alpha)}}\Big)^{-\alpha}
\leq
\Big(\sigma^{\frac{\alpha}{2(1-\alpha)}}\rho
\sigma^{\frac{\alpha}{2(1-\alpha)}}\Big)^{-\alpha}.\label{BNR}
\end{align}
Similar to Eq. (65) of \cite{AHW}, we have
\begin{align*}
& 2^{-\alpha D_{1-\alpha} (\mathcal{E}_\sigma(\rho)\| \sigma) }
=\tr \mathcal{E}_\sigma(\rho)^{1-\alpha}
\sigma^\alpha\\
=&
\tr 
\Big(\sigma^{\frac{\alpha}{2(1-\alpha)}}
\mathcal{E}_\sigma(\rho)
\sigma^{\frac{\alpha}{2(1-\alpha)}}
\Big)^{1-\alpha}
\\
=&
\tr 
\Big(\sigma^{\frac{\alpha}{2(1-\alpha)}}
\mathcal{E}_\sigma(\rho)
\sigma^{\frac{\alpha}{2(1-\alpha)}}
\Big)
\Big(\sigma^{\frac{\alpha}{2(1-\alpha)}}
\mathcal{E}_\sigma(\rho)
\sigma^{\frac{\alpha}{2(1-\alpha)}}
\Big)^{-\alpha}
\\
\stackrel{\text(a)}{=}&
\tr 
\Big(\mathcal{E}_\sigma\Big(\sigma^{\frac{\alpha}{2(1-\alpha)}}
\rho
\sigma^{\frac{\alpha}{2(1-\alpha)}}
\Big)\Big)
\Big(\sigma^{\frac{\alpha}{2(1-\alpha)}}
\mathcal{E}_\sigma(\rho)
\sigma^{\frac{\alpha}{2(1-\alpha)}}
\Big)^{-\alpha}
\\
\stackrel{\text(b)}{=}  &
\tr 
\Big(\sigma^{\frac{\alpha}{2(1-\alpha)}}
\rho
\sigma^{\frac{\alpha}{2(1-\alpha)}}
\Big)
\Big(\sigma^{\frac{\alpha}{2(1-\alpha)}}
\mathcal{E}_\sigma(\rho)
\sigma^{\frac{\alpha}{2(1-\alpha)}}
\Big)^{-\alpha}
\\
\stackrel{\text(c)}{\le}  &
\nu^{\alpha} \tr 
\Big(\sigma^{\frac{\alpha}{2(1-\alpha)}}
\rho
\sigma^{\frac{\alpha}{2(1-\alpha)}}
\Big)
\Big(\sigma^{\frac{\alpha}{2(1-\alpha)}}
\rho
\sigma^{\frac{\alpha}{2(1-\alpha)}}
\Big)^{-\alpha}
\\
=&\nu^{\alpha} \tr 
\Big(\sigma^{\frac{\alpha}{2(1-\alpha)}}
\rho
\sigma^{\frac{\alpha}{2(1-\alpha)}}
\Big)^{1-\alpha}\\
= &
\nu^{\alpha} 2^{-\alpha \widetilde{D}_{1-\alpha} (\rho\| \sigma) },
\end{align*}
where
(a) follows from $
\sigma^{\frac{\alpha}{2(1-\alpha)}}
\mathcal{E}_\sigma(\rho)
\sigma^{\frac{\alpha}{2(1-\alpha)}}=\mathcal{E}_\sigma\left(\sigma^{\frac{\alpha}{2(1-\alpha)}}
\rho
\sigma^{\frac{\alpha}{2(1-\alpha)}}
\right)$, (b) follows from the cyclicity of pinching operator inside trace, and (c) follows from \eqref{BNR}.
\end{proof}

\begin{proposition}
Consider the following four cq-states 
    \begin{align*}
    \label{cq-states}
         \rho^{UVXB} &= \sum_{u,v,x}p(u,v,x)\ketbra{u}\otimes \ketbra{v}\otimes\ketbra{x}\otimes\rho^{B}_{u,v,x},\\ 
\rho^{VX-U-B} &= \sum_{u,v,x}p(u,v,x)\ketbra{u}\otimes \ketbra{v}\otimes\ketbra{x}\otimes\rho^{B}_{u},\\
\rho^{UX-V-B} &= \sum_{u,v,x}p(u,v,x)\ketbra{u}\otimes \ketbra{v}\otimes\ketbra{x}\otimes\rho^{B}_{v},\\
\rho^{UV-X-B} &= \sum_{u,v,x}p(u,v,x)\ketbra{u}\otimes \ketbra{v}\otimes\ketbra{x}\otimes\rho^{B}_{x}.
    \end{align*}
Let $\mathcal{E}^{B}$, $\mathcal{E}^{UB}_1$, $\mathcal{E}^{UVB}_2$, and $\mathcal{E}^{UVXB}_3$ be the pinching maps with respect to the spectral decompositions of the operators $\rho^{B}$, $\mathcal{E}^{B}\left(\rho^{V-U-B}\right)$, and $\mathcal{E}^{UB}_1\left(\rho^{U-V-B}\right)$, and $\mathcal{E}^{UVB}_2(\rho^{UV-X-B})$, respectively. The following inequalities hold:
    \begin{align*}
     2^{-\alpha D_{1-\alpha} \left({\cal E}_2^{UVB}(\rho^{VU XB})\| \mathcal{E}_1^{UB}(\rho^{UX-V- B})\right) }
&\le 
\nu_2^\alpha 2^{-\alpha \widetilde{D}_{1-\alpha} \left(\rho^{VU XB}\| \mathcal{E}_1^{UB}(\rho^{UX-V- B})\right) } \\
2^{-\alpha D_{1-\alpha} \left(\mathcal{E}_3^{UVXB}(\rho^{VU XB})\|
\mathcal{E}_2^{UVB}(\rho^{UV-X- B})\right) }  
&\le  \nu_3^\alpha 2^{-\alpha \widetilde{D}_{1-\alpha} \left(\rho^{VU XB}\|
\mathcal{E}_2^{UVB}( \rho^{UV-X-B})\right) },
\end{align*}
where $\nu_2$ and $\nu_3$ are the number of the distinct eigenvalues of the operators $\mathcal{E}^{UB}_1\left(\rho^{U-V-B}\right)$ and $\mathcal{E}^{UVB}_2(\rho^{UV-X-B})$, respectively.
\end{proposition}
\begin{proof}
   The proofs resemble those of Lemma \ref{petz-sandwich}. We defer them to the interested reader.
\end{proof}

\bigskip

\textbf{Acknowledgement:}
We are grateful to Amin Aminzadeh Gohari and Chandra Nair for insightful discussions. 
This work originated from a discussion held at the ``Workshop on Information Theory and Related Fields: In Memory of Ning Cai'' in Bielefeld, Germany. We extend our sincere gratitude to the organizers, particularly Christian Deppe, for their warm hospitality and gracious invitation. FS gratefully acknowledges the support of the Walter Benjamin Fellowship, DFG project No. 524058134.
PH was supported by ARO (award W911NF2120214), DOE (Q-NEXT), CIFAR and the Simons Foundation. MH was supported in part by the National Natural Science Foundation of China under Grant 62171212.

\bibliographystyle{ieeetr}
\bibliography{3-broadcast.bib}

\end{document}